\documentclass[a4paper,12pt]{article}
\usepackage[T1]{fontenc}
\pdfoutput=1

\usepackage{jheppub}
\usepackage{macro}
\usepackage{appendix}
\usepackage{comment}
\usepackage{array}
\usepackage{yfonts}
\usepackage{comment}
\usepackage{amsthm}
\usepackage{tikz-cd}
\usepackage{amsmath,amssymb}
\usepackage{tikz}
\usepackage{pgfkeys}
\usepackage{pgfopts}
\usepackage{xcolor}
\usepackage{young}
\usepackage{youngtab}
\usepackage{ytableau}
\usepackage{titlesec}
\usetikzlibrary{decorations.pathreplacing}

\usepackage{tikz}

\titleformat{\paragraph}
{\normalfont\normalsize\bfseries}{\theparagraph}{1em}{}
\titlespacing*{\paragraph}
{0pt}{3.25ex plus 1ex minus .2ex}{1.5ex plus .2ex}

\newtheorem{Theorem}{Theorem}[section]

\newtheorem{Definition}{Definition}[section]
\newtheorem{Lemma}{Lemma}[section]
\newtheorem{Proposition}{Proposition}[section]

\usetikzlibrary{matrix,calc,positioning,decorations.markings,decorations.pathmorphing,decorations.pathreplacing
}
\usetikzlibrary{arrows,cd,shapes}

\title{\textbf Duality, asymptotic charges and algebraic topology in mixed symmetry tensor gauge theories and applications}

\author[a,b]{Federico Manzoni}

\affiliation[a]{Mathematics and Physics department, Roma Tre, Via della Vasca Navale 84, Rome, Italy}
\affiliation[b]{INFN Roma Tre Section, Physics department, Via della Vasca Navale 84, Rome, Italy}
\emailAdd{federico.manzoni@uniroma3.it, ORCID ID: 0000-0002-9979-6154}

\abstract{Recently the duality map between electric-like asymptotic charges of $p$-form gauge theories is studied. The outcome is an existence and uniqueness theorem and the topological nature of the duality map. The goal of this work is to extend that theorem in the case of mixed symmetry tensor gauge theories in order to have a deeper understanding of exotic gauge theories, of the non-trivial charges associated to them and of the duality of their observables. Unlike the simpler case of $p$-form gauge theories, here we need to develop some mathematical tools. The crucial points are to view a mixed symmetry tensor as a Young projected object of the $N$-multi-form space and to develop an analogue of de Rham complex for mixed symmetry tensors. As a result, if the underlying manifold satisfy appropriate conditions, the duality map can be proven to exist and to be unique ensuring the charge of
a description has information on the dual ones. Moreover, we provide some physical applications ranging form fractons and higher symmetries to string theory and holography.}

\begin{document}

\maketitle

\section{Introduction}
Since Maxwell's modern theory of electromagnetism \cite{max} until today, gauge theories have played a major role in understanding and describing Nature. The invariance of physics under local transformations has led to the description of all fundamental interactions from General Relativity to the Standard Model of Particle Physics. However, in these theories only a very small subclass of representations of the Lorentz group are used as gauge fields and between the seventies and the eighties, eyes were directed towards the possibility of using arbitrarily high spin fields as gauge fields leading to the Higher-Spin theories \cite{Curtright:1980un,Francia:2004lbf, Campoleoni:2017qot,Campoleoni:2018uib,Campoleoni:2020ejn,sagnotti2012notes,Didenko:2014dwa,Campoleoni:2017mbt}. Systematic study of massless arbitrary spin fields was initiated by Fronsdal in 1978 \cite{fron1},\cite{fron2}. Usually, the spectrum of such theories contains the graviton as a massless spin-two field ans since Higher-Spin theories are supposed to be consistent quantum theories and, for this reason, to give examples of quantum gravity theories. However, Higher-Spin theories, suffers of no-go theorems which, for example, trivialize interactions \cite{mandu,Argyres:1989cu,Aragone:1979bm,Aragone:1983sz,Weinberg:1964ew,sagnotti2012notes}. \\ More or less in the same years, Curtright showed how mixed symmetry tensor fields can be used as consisted gauge field, generalizing the concept of a gauge field to include higher rank Lorentz tensors which are neither totally symmetric nor totally antisymmetric \cite{Curtright:1980yk}. The main interest is due to the fact that an infinite family of mixed symmetry gauge fields arises in the zero tension limit of String Theory \cite{Bonelli_2003}. The simplest of these mixed symmetry tensor gauge theories has as fundamental object a three indexes mixed symmetry tensor field called Curtright hooke field; moreover, its the gauge-invariant dynamics is dual to those of the graviton in $D=5$ dimensions. This is due to an underlying duality between the Young tableaux defining the two irreducible representations. This duality goes well beyond this simple example and a generic mixed symmetry tensor fields posses a bunch of on-shell dual descriptions \cite{hamermesh1989group}. Therefore, studying the Curtright hooke field in generic dimension can lead to new insights on the nature of gravity, at least on its perturbative sector.

As a general feature of gauge theories, we can find some large gauge transformations that act in a non-trivial way on the physical states at infinity, where infinity means in the nearby of a boundary \cite{Bondi:1962px, PhysRev.128.2851,Esmaeili:2020eua,Campoleoni:2018uib, Romoli:2024hlc}. One of the first examples is provided by the BMS group in the context of Einstein gravity \cite{Sachs:1962wk,Bondi:1962px,PhysRev.128.2851}. We stress also that recently an asymptotic algebra for the case of gravisolitons spacetime was found \cite{Belinski:2001ph,Manzoni:2021dij,Manzoni:2024agc}.

In recent past, was shown as the asymptotic symmetries of the scalar in $D=4$ can be interpreted as the asymptotic symmetries of a $2$-form; hence the asymptotic charge of a description contains information about the dual one \cite{Francia:2018jtb,Campiglia:2018see}. Furthermore, in \cite{Manzoni:2024tow},\cite{manzf} we show that also in the realm of $p$-form gauge theories, asymptotic charges contains information on the dual formulation. The main goal of the work is to extend the results of \cite{Manzoni:2024tow} to more general case of mixed symmetry tensor gauge theories. The possibility of a well-defined and unique map, under appropriate topological assumptions, would be of interest to various fields of physics, as will be discussed in the paragraph
\ref{phy}. The route towards such a map goes through the construction of a de Rham-like complex with the idea that a mixed symmetry tensor associated with a Young diagram is a suitably projected N-multi-form, following the idea of Hull and Medeiros \cite{Medeiros_2003}. The tricky conceptual point is to define an appropriate cohomology  to be used as de Rham cohomology for the case of differential forms or 1-multi-forms in this language. These points are covered in paragraphs \ref{nmulti} and \ref{MSTT} where we discuss also the existence and uniqueness of duality maps in the case of mixed symmetry tensors. These paragraphs are presented in a more formal way in order to elucidate better the construction steps.

The physical applications can cover a lot of fields. First of all, we introduce mixed symmetry memory effects and thanks to the duality theorem we propose a possibility to detect extra dimensions measuring memory effects of dual descriptions. The applications to higher symmetry and string theory are essentially based on the isomorphism between de Rham-like cohomology and standard de Rham cohomology: in a nutshell, extended defect can source also mixed symmetry fields and the spectrum of allowed charges is determined by de Rham cohomology of the compactification. Moreover, the closure property in the de Rham-like complex seem to be the right framework for the reinterpretation of fracton dynamics in full cohomological features while the application in holography lead to the possibility that certain operator algebras in the CFT must close into each other, thereby reducing the freedom in bootstrap equations.

\section{The duality map for $p$-form gauge theories}\label{dualitymap}

In a previous work \cite{Manzoni:2024tow}, we proposed a duality map between the electric-like charge $Q^{(e)}_{p,D}$ and the electric-like charge of the dual theory $Q^{(e)}_{q,D}$ with $q=D-p-2$.
In this section we review the theorem on the duality map for the case of $p$-form gauge theories.
This map is between well defined charges, i.e. charges with radiation fall-off for the fields and charges with Coulomb fall-off for the field in $D=2p+2$. The existence and uniqueness of the duality map is stated by the following theorem
\begin{Theorem}[Existence and uniqueness of the duality map for well defined charges]\label{THM3.1} 
Let $(M_D,\boldsymbol{\eta})$ be the $D$-dimensional Minkowski spacetime. Then a duality map $f \in \mathrm{GL}(n_p,\mathbb{C})$, such that the following diagram
\begin{equation}
\centering
\begin{tikzcd}
\Omega_{\mathrm{AS}}^{p+1}(M_D) \arrow[rr, "\star"] & & \Omega_{\mathrm{AS}}^{q+1}(M_D)\\
\Omega_{\mathrm{AS}}^{p}(M_D) \arrow[u, "d"] \arrow[d, "\pi_1", swap ] \arrow[rr, "\star_{D-2}"] & & \Omega_{\mathrm{AS}}^{q}(M_D) \arrow[u, "d", swap] \arrow[d, "\pi_2"]\\
\mathbb{C}^{n_p}  \arrow[rr, "f"] & & \mathbb{C}^{n_{q}} \\
\end{tikzcd}
\end{equation}
commutes, exists. Moreover $f$ admits a unique restriction to a 1-dimensional subspace such that $f|_{Q}: Q^{(e)}_{p,D} \mapsto Q_{q,D}^{(e)}$ and $f^{-1}|_{Q}: Q^{(e)}_{q,D} \mapsto Q^{(e)}_{p,D}$.
\end{Theorem}
Let us explain the content of Theorem \ref{THM3.1} which can be interpreted in the realm of algebraic topology; the interesting point is the topological nature of the duality map. Let us consider two copies of the de Rham complex, one labelled by $p$, $C_{\text{dR}}^{*(p)}$, and one labelled by $q=D-p-2$, $C_{\text{dR}}^{*(q)}$
\begin{equation}
\centering
\begin{tikzpicture}
[x=0.75pt,y=0.75pt,yscale=-1,xscale=1]

\draw    (252,175) -- (314.77,175.36) ;
\draw [shift={(316.77,175.37)}, rotate = 180.33] [color={rgb, 255:red, 0; green, 0; blue, 0 }  ][line width=0.75]    (10.93,-3.29) .. controls (6.95,-1.4) and (3.31,-0.3) .. (0,0) .. controls (3.31,0.3) and (6.95,1.4) .. (10.93,3.29)   ;
\draw    (342,175) -- (404.77,175.36) ;
\draw [shift={(406.77,175.37)}, rotate = 180.33] [color={rgb, 255:red, 0; green, 0; blue, 0 }  ][line width=0.75]    (10.93,-3.29) .. controls (6.95,-1.4) and (3.31,-0.3) .. (0,0) .. controls (3.31,0.3) and (6.95,1.4) .. (10.93,3.29)   ;
\draw    (449,176) -- (511.77,176.36) ;
\draw [shift={(513.77,176.37)}, rotate = 180.33] [color={rgb, 255:red, 0; green, 0; blue, 0 }  ][line width=0.75]    (10.93,-3.29) .. controls (6.95,-1.4) and (3.31,-0.3) .. (0,0) .. controls (3.31,0.3) and (6.95,1.4) .. (10.93,3.29)   ;
\draw    (150,175) -- (212.77,175.36) ;
\draw [shift={(214.77,175.37)}, rotate = 180.33] [color={rgb, 255:red, 0; green, 0; blue, 0 }  ][line width=0.75]    (10.93,-3.29) .. controls (6.95,-1.4) and (3.31,-0.3) .. (0,0) .. controls (3.31,0.3) and (6.95,1.4) .. (10.93,3.29)   ;
\draw    (252,115) -- (314.77,115.36) ;
\draw [shift={(316.77,115.37)}, rotate = 180.33] [color={rgb, 255:red, 0; green, 0; blue, 0 }  ][line width=0.75]    (10.93,-3.29) .. controls (6.95,-1.4) and (3.31,-0.3) .. (0,0) .. controls (3.31,0.3) and (6.95,1.4) .. (10.93,3.29)   ;
\draw    (342,115) -- (404.77,115.36) ;
\draw [shift={(406.77,115.37)}, rotate = 180.33] [color={rgb, 255:red, 0; green, 0; blue, 0 }  ][line width=0.75]    (10.93,-3.29) .. controls (6.95,-1.4) and (3.31,-0.3) .. (0,0) .. controls (3.31,0.3) and (6.95,1.4) .. (10.93,3.29)   ;
\draw    (449,116) -- (511.77,116.36) ;
\draw [shift={(513.77,116.37)}, rotate = 180.33] [color={rgb, 255:red, 0; green, 0; blue, 0 }  ][line width=0.75]    (10.93,-3.29) .. controls (6.95,-1.4) and (3.31,-0.3) .. (0,0) .. controls (3.31,0.3) and (6.95,1.4) .. (10.93,3.29)   ;
\draw    (150,115) -- (212.77,115.36) ;
\draw [shift={(214.77,115.37)}, rotate = 180.33] [color={rgb, 255:red, 0; green, 0; blue, 0 }  ][line width=0.75]    (10.93,-3.29) .. controls (6.95,-1.4) and (3.31,-0.3) .. (0,0) .. controls (3.31,0.3) and (6.95,1.4) .. (10.93,3.29)   ;
\draw    (124,125) -- (124.37,161.43) ;
\draw [shift={(124.39,163.43)}, rotate = 269.42] [color={rgb, 255:red, 0; green, 0; blue, 0 }  ][line width=0.75]    (10.93,-3.29) .. controls (6.95,-1.4) and (3.31,-0.3) .. (0,0) .. controls (3.31,0.3) and (6.95,1.4) .. (10.93,3.29)   ;
\draw    (530,125) -- (530.37,161.43) ;
\draw [shift={(530.39,163.43)}, rotate = 269.42] [color={rgb, 255:red, 0; green, 0; blue, 0 }  ][line width=0.75]    (10.93,-3.29) .. controls (6.95,-1.4) and (3.31,-0.3) .. (0,0) .. controls (3.31,0.3) and (6.95,1.4) .. (10.93,3.29)   ;
\draw    (420,125) -- (420.37,161.43) ;
\draw [shift={(420.39,163.43)}, rotate = 269.42] [color={rgb, 255:red, 0; green, 0; blue, 0 }  ][line width=0.75]    (10.93,-3.29) .. controls (6.95,-1.4) and (3.31,-0.3) .. (0,0) .. controls (3.31,0.3) and (6.95,1.4) .. (10.93,3.29)   ;
\draw    (327,125) -- (327.37,161.43) ;
\draw [shift={(327.39,163.43)}, rotate = 269.42] [color={rgb, 255:red, 0; green, 0; blue, 0 }  ][line width=0.75]    (10.93,-3.29) .. controls (6.95,-1.4) and (3.31,-0.3) .. (0,0) .. controls (3.31,0.3) and (6.95,1.4) .. (10.93,3.29)   ;
\draw    (224,125) -- (224.37,161.43) ;
\draw [shift={(224.39,163.43)}, rotate = 269.42] [color={rgb, 255:red, 0; green, 0; blue, 0 }  ][line width=0.75]    (10.93,-3.29) .. controls (6.95,-1.4) and (3.31,-0.3) .. (0,0) .. controls (3.31,0.3) and (6.95,1.4) .. (10.93,3.29)   ;

\draw (319,164.4) node [anchor=north west][inner sep=0.75pt]    {$\Omega^{q}$};
\draw (216,164.4) node [anchor=north west][inner sep=0.75pt]    {$\Omega^{q-1}$};
\draw (519,164.4) node [anchor=north west][inner sep=0.75pt]    {$\Omega^{q+2}$};
\draw (411,164.4) node [anchor=north west][inner sep=0.75pt]    {$\Omega^{q+1}$};
\draw (114,164.4) node [anchor=north west][inner sep=0.75pt]    {$\Omega^{q-2}$};
\draw (562,172.4) node [anchor=north west][inner sep=0.75pt]    {$...$};
\draw (97,172.4) node [anchor=north west][inner sep=0.75pt]    {$...$};
\draw (163,149.4) node [anchor=north west][inner sep=0.75pt]    {$d_{q-2}$};
\draw (466,149.4) node [anchor=north west][inner sep=0.75pt]    {$d_{q+1}$};
\draw (359,149.4) node [anchor=north west][inner sep=0.75pt]    {$d_{q}$};
\draw (267,149.4) node [anchor=north west][inner sep=0.75pt]    {$d_{q-1}$};
\draw (319,104.4) node [anchor=north west][inner sep=0.75pt]    {$\Omega^{p}$};
\draw (216,104.4) node [anchor=north west][inner sep=0.75pt]    {$\Omega^{p-1}$};
\draw (519,104.4) node [anchor=north west][inner sep=0.75pt]    {$\Omega^{p+2}$};
\draw (411,104.4) node [anchor=north west][inner sep=0.75pt]    {$\Omega^{p+1}$};
\draw (114,104.4) node [anchor=north west][inner sep=0.75pt]    {$\Omega^{p-2}$};
\draw (562,114.4) node [anchor=north west][inner sep=0.75pt]    {$...$};
\draw (97,114.4) node [anchor=north west][inner sep=0.75pt]    {$...$};
\draw (163,89.4) node [anchor=north west][inner sep=0.75pt]    {$d_{p-2}$};
\draw (466,89.4) node [anchor=north west][inner sep=0.75pt]    {$d_{p+1}$};
\draw (359,89.4) node [anchor=north west][inner sep=0.75pt]    {$d_{p}$};
\draw (267,89.4) node [anchor=north west][inner sep=0.75pt]    {$d_{p-1}$};
\draw (422,128.4) node [anchor=north west][inner sep=0.75pt]    {$\star _{D}$};
\draw (329,128.4) node [anchor=north west][inner sep=0.75pt]    {$\star _{D-2}$};
\draw (226,128.4) node [anchor=north west][inner sep=0.75pt]    {$\star _{D-4}$};
\draw (126,128.4) node [anchor=north west][inner sep=0.75pt]    {$\star _{D-6}$};
\draw (532,128.4) node [anchor=north west][inner sep=0.75pt]    {$\star _{D+2}$};
\end{tikzpicture}
\label{topocosa}
\end{equation}
where every $\star_{D+2n}$ with $n \in \mathbb{Z} \setminus \{-\infty,+\infty\}$ is required to be a group homomorphism and $\star_D$ is the Hodge operator. Noting that $q-p=D-p-2-p=D-2p-2$, we have that 
\begin{equation}
    \star^{*} : C_{\text{dR}}^{*(p)} \mapsto C_{\text{dR}}^{*(q)} \ \ \ \ [D-2p-2]
\end{equation}
is an homotopy of cochain complexes; in critical dimension, i.e. $p=\frac{D-2}{2}=q$, the $p$-form gauge theory is self dual and the homotopy $\star^*$ is an isomorphism of cochain complexes since every space of forms is mapped in itself. Now, since we are interested in considering gauge field theories on Minkowski spacetime we can assume trivial topology, i.e. all the cohomology groups ($n \neq 0$) of the de Rham complex are trivial
\begin{equation}
    H^n=\frac{Z_n:=\{B\in \Omega^n | d_{n}B=0\}}{B_n:=\{d_{n-1}A \in \Omega^n | A \in \Omega^{n-1}\}}=0;
\end{equation}
this means that every cocycle is also a coboundary\footnote{In other words, every closed form is exact.}. Therefore, the de Rham complexes in \eqref{topocosa} are exact sequences of Abelian group. Now, let us restrict to only one de Rham complex\footnote{The same considerations hold for $C_{dR}^{*(q)}$.}, $C_{\textnormal{dR}}^{*(p)}$
\begin{equation}
\centering
\tikzset{every picture/.style={line width=0.75pt}} 
\begin{tikzpicture}[x=0.75pt,y=0.75pt,yscale=-1,xscale=1]

\draw    (252,115) -- (314.77,115.36) ;
\draw [shift={(316.77,115.37)}, rotate = 180.33] [color={rgb, 255:red, 0; green, 0; blue, 0 }  ][line width=0.75]    (10.93,-3.29) .. controls (6.95,-1.4) and (3.31,-0.3) .. (0,0) .. controls (3.31,0.3) and (6.95,1.4) .. (10.93,3.29)   ;
\draw    (342,115) -- (404.77,115.36) ;
\draw [shift={(406.77,115.37)}, rotate = 180.33] [color={rgb, 255:red, 0; green, 0; blue, 0 }  ][line width=0.75]    (10.93,-3.29) .. controls (6.95,-1.4) and (3.31,-0.3) .. (0,0) .. controls (3.31,0.3) and (6.95,1.4) .. (10.93,3.29)   ;
\draw    (449,116) -- (511.77,116.36) ;
\draw [shift={(513.77,116.37)}, rotate = 180.33] [color={rgb, 255:red, 0; green, 0; blue, 0 }  ][line width=0.75]    (10.93,-3.29) .. controls (6.95,-1.4) and (3.31,-0.3) .. (0,0) .. controls (3.31,0.3) and (6.95,1.4) .. (10.93,3.29)   ;
\draw    (150,115) -- (212.77,115.36) ;
\draw [shift={(214.77,115.37)}, rotate = 180.33] [color={rgb, 255:red, 0; green, 0; blue, 0 }  ][line width=0.75]    (10.93,-3.29) .. controls (6.95,-1.4) and (3.31,-0.3) .. (0,0) .. controls (3.31,0.3) and (6.95,1.4) .. (10.93,3.29)   ;

\draw (319,104.4) node [anchor=north west][inner sep=0.75pt]    {$\Omega^{p}$};
\draw (216,104.4) node [anchor=north west][inner sep=0.75pt]    {$\Omega^{p-1}$};
\draw (519,104.4) node [anchor=north west][inner sep=0.75pt]    {$\Omega^{p+2}$};
\draw (411,104.4) node [anchor=north west][inner sep=0.75pt]    {$\Omega^{p+1}$};
\draw (114,104.4) node [anchor=north west][inner sep=0.75pt]    {$\Omega^{p-2}$};
\draw (562,114.4) node [anchor=north west][inner sep=0.75pt]    {$...$};
\draw (97,114.4) node [anchor=north west][inner sep=0.75pt]    {$...$};
\draw (163,89.4) node [anchor=north west][inner sep=0.75pt]    {$d_{p-2}$};
\draw (466,89.4) node [anchor=north west][inner sep=0.75pt]    {$d_{p+1}$};
\draw (359,89.4) node [anchor=north west][inner sep=0.75pt]    {$d_{p}$};
\draw (267,89.4) node [anchor=north west][inner sep=0.75pt]    {$d_{p-1}$};
\end{tikzpicture},
\end{equation}
and let us taken into account the fact we are interested in asymptotic symmetries. We start with a $p$-form gauge theory, with gauge field $B \in \Omega^p$ and the asymptotic charge is written in terms of the field strength $H=d_pB \in \Omega^{p+1}$ which we require to be non-vanishing\footnote{Otherwise the asymptotic charge would be zero and would be associated to a trivial gauge transformation.}. Since $C_{\mathrm{dR}}^{*(p)}$ on Minkowski spacetime is exact  we have $H=0 \Leftrightarrow B=d_{p-1}A$ for some $A \in \Omega^{p-1}$; hence we need to throw away all those elements $B \in \Omega^p$ such that $B=d_{p-1}A$ for some $A \in \Omega^{p-1}$. Moreover, only the zero form can have vanishing field strength but, again for exactness $B=0 \Leftrightarrow A=d_{p-2}C$ for some $C \in \Omega^{p-2}$. Therefore, for asymptotic symmetries scopes we can replace $\Omega^{p+1}$ with $\Omega^{p+1}_{\text{AS}}:=\{H \in \Omega^{p+1} | H=d_pB, H \neq 0\} \cup \{H=0\}$, $\Omega^{p}$ remains the same but we dubbed it as with $\Omega^{p}_{\text{AS}}$ and $\Omega^{p-1}$ and $\Omega^{p+1}$ with 0 to get the asymptotic symmetries de Rham complex $C_{\text{ASdR}}^{*(p)}$
\begin{equation}
\centering
\begin{tikzpicture}[x=0.75pt,y=0.75pt,yscale=-1,xscale=1]

\draw    (233,115) -- (295.77,115.36) ;
\draw [shift={(297.77,115.37)}, rotate = 180.33] [color={rgb, 255:red, 0; green, 0; blue, 0 }  ][line width=0.75]    (10.93,-3.29) .. controls (6.95,-1.4) and (3.31,-0.3) .. (0,0) .. controls (3.31,0.3) and (6.95,1.4) .. (10.93,3.29)   ;
\draw    (332,115) -- (394.77,115.36) ;
\draw [shift={(396.77,115.37)}, rotate = 180.33] [color={rgb, 255:red, 0; green, 0; blue, 0 }  ][line width=0.75]    (10.93,-3.29) .. controls (6.95,-1.4) and (3.31,-0.3) .. (0,0) .. controls (3.31,0.3) and (6.95,1.4) .. (10.93,3.29)   ;
\draw    (439,116) -- (501.77,116.36) ;
\draw [shift={(503.77,116.37)}, rotate = 180.33] [color={rgb, 255:red, 0; green, 0; blue, 0 }  ][line width=0.75]    (10.93,-3.29) .. controls (6.95,-1.4) and (3.31,-0.3) .. (0,0) .. controls (3.31,0.3) and (6.95,1.4) .. (10.93,3.29)   ;

\draw (300,104.4) node [anchor=north west][inner sep=0.75pt]    {$\Omega_{\text{AS}}^{p}$};
\draw (216,108.4) node [anchor=north west][inner sep=0.75pt]    {$0$};
\draw (509,108.4) node [anchor=north west][inner sep=0.75pt]    {$0$};
\draw (401,104.4) node [anchor=north west][inner sep=0.75pt]    {$\Omega_{\text{AS}}^{p+1}$};
\draw (456,89.4) node [anchor=north west][inner sep=0.75pt]    {$d_{p+1}$};
\draw (349,89.4) node [anchor=north west][inner sep=0.75pt]    {$d_{p}$};
\draw (248,89.4) node [anchor=north west][inner sep=0.75pt]    {$d_{p-1}$};
\end{tikzpicture}
\end{equation} 
which is a short exact sequence. By general reasoning or by explicit computations follows that $d_p$ is an isomorphism. Looking now at diagram \eqref{topocosa} and reducing the de Rham complexes to the asymptotic symmetries de Rham complexes, Theorem \eqref{THM3.1} can be used to construct $f$ and then the duality map since now diagram \eqref{topocosa} reduces to the upper part of the diagram of Theorem \eqref{THM3.1}. Therefore, the duality map is topological in nature and can be constructed if and only if 
\begin{equation}
    H^p=H^{p+1}=0=H^{q+1}=H^{q}.
    \label{cohom}
\end{equation}
Indeed the vanishing of these cohomology groups is sufficient to reduce the full de Rham complex to the asymptotic symmetries de Rham complex and it is also necessary since to construct the duality map we need that $d_p$ and $d_q$ are isomorphisms.

\section{The $N$-multi-form space}\label{nmulti}
Let us review and clarify the $N$-multi-form space starting with the case of the bi-form space \cite{Medeiros_2003,Francia:2004lbf}. First of all, the basic definition 
\begin{Definition}[Bi-form space]
    \textit{The bi-form space on a $D$-dimensional differential manifold $(M,\mathcal{A})$ with metric $\boldsymbol{g}$, dubbed $\Omega^{p \otimes q}(M)$, is the tensor product space between the space of $p$-forms $\Omega^{p}(M)$ and the space of $q$-form $\Omega^{q}(M)$.}
\end{Definition}
In general the resulting tensor $T$ can be written as
\begin{equation}
    T=\frac{1}{p!q!}T_{[\mu_1...\mu_p][\nu_1...\nu_q]}(dx^{\mu_1} \wedge ... \wedge dx^{\mu_p}) \otimes (dx^{\nu_1} \wedge ... \wedge dx^{\nu_q}),
\end{equation}
and, always in general, carries a reducible representation of $\mathrm{GL}(D)$. To extract the irreducible representation labelled by the Young tableau $\lambda=(p,q)$ we imake us of the Young projector $\Pi_{(p,q)} : \Omega^{p \otimes q}(M) \rightarrow \Omega^{p \otimes q}(M)$. In the space of bi-forms are well defined the left and right differentials 
\begin{Definition}[Left and right differential]
    \textit{The left differential $d_L$ and right differential $d_R$ are the usual de Rham differential that act only on one of the spaces of the differential forms used to construct the space of the bi-forms. Therefore}
    \begin{equation}
        d_L: \Omega^{p \otimes q}(M) \rightarrow \Omega^{p+1\otimes q}(M), \ \ \ \ \ d_R: \Omega^{p \otimes q}(M) \rightarrow \Omega^{p\otimes q+1}(M).
    \end{equation}
\end{Definition}
The following result is easy to show
\begin{Proposition}[Left and right differential basic properties]\label{lrd}
    The left and right differentials are such that $d_L \circ d_L=0$, $d_R \circ d_R=0$, moreover they commute $d_L\circ d_R=d_R\circ d_L$.
\end{Proposition}
\begin{proof}
    It follows immediately from the nilpotency property of the standard de Rham differential and from the commutativity of derivatives since differential forms have smooth components.
\end{proof}
Since we are interested in the application of this formalism in gauge theories we can define the left and right field strength as follows.
\begin{Definition}[Left and right field strength]
    \textit{Given a gauge field $B$ whose writing in a chart $(\phi,U) \in \mathcal{A}$ is a mixed symmetry tensor carrying the irreducible representation corresponding to the Young tableau $\lambda=(p,q)$, the left field strength $H_L$ and the right field strength $H_R$ are the Young projected bi-forms given by}
    \begin{equation}
        H_L:=d_L B \in \Omega^{p+1\otimes q}(M), \ \ \ \ \ H_R:=d_R B \in \Omega^{p\otimes q+1}(M),
    \end{equation}
    \textit{when they are meaningful, i.e. when the application of Young projector extracts a bi-form with a well defined associated Young tableau.}
\end{Definition}
Roughly speaking, these are partial field strengths 
which are not completely gauge invariant but that are useful to write down lagrangians density; they correspond to the irreducible representation given by the Young tableaux $\lambda_L=(p+1,q)$ and $\lambda_R=(p,q+1)$. The full gauge invariant field strength can be constructed using the left and right differentials
\begin{Definition}[Field strength]
    \textit{Given a gauge field $B$ whose writing in a chart $(\phi,U) \in \mathcal{A}$ is a mixed symmetry tensor carrying the irreducible representation corresponding to the Young tableau $\lambda=(p,q)$, the field strength $H$ is the Young projected bi-forms given by}
    \begin{equation}
        H:=d_L \circ d_R B= d_R \circ d_L B\in \Omega^{p+1\otimes q+1}(M).
    \end{equation}
\end{Definition}
This object is by construction fully gauge invariant due to the properties proven in Proposition \ref{lrd} and its writing in a chart corresponds to a mixed symmetry tensor which carries the irreducible representation associated to the Young tableau $\lambda=(p+1,q+1)$. In a similar way we can introduce the left and right Hodge morphism 
\begin{Definition}[Left and right Hodge morphism]
\textit{The left Hodge morphism $\star_L$ and right Hodge morphism $\star_R$ are the usual Hodge morphism that act only on one of the spaces of the differential forms used to construct the space of the bi-forms. Therefore}
    \begin{equation}
        \star_L: \Omega^{p \otimes q}(M) \rightarrow \Omega^{D-p\otimes q}(M), \ \ \ \ \ \star_R: \Omega^{p \otimes q}(M) \rightarrow \Omega^{p\otimes D-q}(M).
    \end{equation}
 \textit{The full Hodge morphism is simply given by $\star:=\star_L\circ \star_R=\star_R\circ \star_L$.}    
\end{Definition}
This formalism can be generalized to the case of more than two form spaces as follows
\begin{Definition}[$N$-multi-form space]
    \textit{The $N$-multi-form space on a $D$-dimensional differential manifold $(M,\mathcal{A})$ with metric $\boldsymbol{g}$, dubbed $\Omega^{p_1 \otimes ... \otimes p_N}(M)$, is the tensor product space between the spaces of $p_i$-forms $\Omega^{p_i}(M)$ with $i \in [1,N]$.}
\end{Definition}
An object living in the $N$-multi-form space is given, in local coordinates, by a mixed symmetry tensor $T_{[\mu_1^{(1)}...\mu_{p_1}^{(1)}]...[\mu_1^{(N)}...\mu_{p_N}^{(N)}]}$ and to extract the irreducible representation labelled
by the Young tableaux $\lambda=(p_1,...,p_N)$ we need to introduce the Young projector $\Pi_{(p_1,..,p_N)} : \Omega^{p_1 \otimes ... \otimes p_N}(M) \rightarrow \Omega^{p_1 \otimes ... \otimes p_N}(M)$. In the $N$-multi-form space is defined the $i$-th differential
\begin{Definition}[$i$-th differential]
    \textit{The $i$-th differential $d^{(i)}$ is the usual de Rham differential that acts only on one of the spaces of the differential forms used to construct the space of the $N$-multi-forms. Therefore}
    \begin{equation}
        d^{(i)}: \Omega^{p_1 \otimes ... \otimes p_i \otimes ... \otimes p_N}(M) \rightarrow  \Omega^{p_1 \otimes ... \otimes p_i+1 \otimes ... \otimes p_N}(M).
    \end{equation}
\end{Definition}
Results of Proposition \ref{lrd} generalize quite obviously to $d^{(i)} \circ d^{(i)}=0 \ \ \forall i \in [1,N]$ and $d^{(i)} \circ d^{(j)}=d^{(j)} \circ d^{(i)} \ \ \forall i \neq j.$ Thanks to these differentials we can build up the $i$-th field strength
\begin{Definition}[$i$-th field strength]
    \textit{Given a gauge field $B$ whose writing in a chart $(\phi,U) \in \mathcal{A}$ is a mixed symmetry tensor carrying the irreducible representation corresponding to the Young tableau $\lambda=(p_1,...p_N)$, the $i$-th field strength $H^{(i)}$ is the Young projected $N$-multi-forms given by}
    \begin{equation}
        H^{(i)}:=d^{(i)} B \in \Omega^{p_1 \otimes ... \otimes p_i+1 \otimes ... \otimes p_N}(M),
    \end{equation}
    \textit{when they are meaningful, i.e. when the application of Young projector extracts a $N$-multi-form with a well defined associated Young tableau.}
\end{Definition}
As before, these are not fully gauge invariant and their writing in a chart corresponds to a mixed symmetry tensor which carries the irreducible representation associated to the Young tableaux $\lambda^{(i)}=(p_1,...,p_i+1,...,p_N)$. To define the field strength we introduce the following definition 
\begin{Definition}[$k$-cumulative field strength]
\textit{Given a gauge field $B$ whose writing in a chart $(\phi,U) \in \mathcal{A}$ is a mixed symmetry tensor carrying the irreducible representation corresponding to the Young tableau $\lambda=(p_1,...p_N)$, the $k$-cumulative field strength $H^{(k)}$ is the Young projected $N$-multi-forms given by}
    \begin{equation}
        H^{(k)}:=d^{(k)} \circ d^{(k-1)} \circ ... \circ d^{(2)} \circ d^{(1)}B \in \Omega^{p_1+1 \otimes ... \otimes p_k+1 \otimes p_{k+1} \otimes ... \otimes p_N}(M) \qquad k\leq N
    \end{equation} 
when it is meaningful, i.e. when the application of Young projector extracts a $N$-multi-form with a well defined associated Young tableau.
\end{Definition}
These objects will be useful in the description of the duality of HS gauge theories in the next Paragraph. In the end we have the definition of the field strength
\begin{Definition}[Field strength]\label{def2.5.10}
    \textit{Given a gauge field $B$ whose writing in a chart $(\phi,U) \in \mathcal{A}$ is a mixed symmetry tensor carrying the irreducible representation corresponding to the Young tableau $\lambda=(p_1,...p_N)$, the field strength $H$ is the Young projected $N$-multi-forms given by}
    \begin{equation}
        H:=d^{(N)} \circ d^{(N-1)} \circ ... \circ d^{(2)} \circ d^{(1)}B \in \Omega^{p_1+1 \otimes ... \otimes p_i+1 \otimes ... \otimes p_N+1}(M).
    \end{equation}
\textit{Equivalently, the field strength $H$ is the $N$-cumulative field strength.}
\end{Definition}
This object is by construction fully gauge invariant due to the properties generalized from Proposition \ref{lrd} and its writing in a chart correspond to a mixed symmetry tensor which carries the irreducible representation associated to the Young tableau $\lambda=(p_1+1,...,p_N+1)$. We can also introduce the $i$-th Hodge morphism as follows
\begin{Definition}[$i$-th Hodge morphism]
\textit{The $i$-th Hodge morphism $\star^{(i)}$ is the usual Hodge morphism that acts only on one of the spaces of the differential forms used to construct the space of the $N$-multi-forms. Therefore}
    \begin{equation}
        \star^{(i)}: \Omega^{p_1 \otimes ... \otimes p_i \otimes ... \otimes p_N}(M) \rightarrow  \Omega^{p_1 \otimes ... \otimes D-p_i \otimes ... \otimes p_N}(M).
    \end{equation}
 \textit{The full Hodge morphism is simply given by $\star:=\star^{(N)} \circ \star^{(N-1)} \circ ... \circ \star^{(2)} \circ \star^{(1)}$.}    
\end{Definition}
In the case $N=1$, i.e. the tensor product is trivial and we have only a standard space of forms of some degree, all the $i$-th field strength degenerate to the only field strength $H$.\\

\section{Duality for well defined charges in mixed symmetry gauge theories}\label{MSTT}
In this Paragraph we want to extend the Theorem \ref{THM3.1} to the case of mixed symmetry tensors. We suppose we are dealing with a gauge field $B$ whose writing in a chart is a mixed symmetry tensor carrying the irreducible representation corresponding to the Young tableau $\lambda=\{p_1,...,p_N\}$.  We proceed in fully general and abstract way, supposing that we have already calculated the asymptotic charges of both the original description and its $r$ dual descriptions and we assume these are well defined charges. We refer to these charges as $Q_0, ..., Q_r$ where $Q_0$ is the asymptotic charge of the original description while $Q_r$ that of its $r$-th dual description. The point is to look at the mixed symmetry tensor as a Young projected element of the $N$-multi-form space using the Young projector $\Pi_{(p_1,...,p_N)}$. Doing so we can construct de Rham-like complexes for the $N$-multi-form space; moreover, requiring the vanishing of some de Rham-like cohomology groups we can reduce them to the asymptotic symmetries de Rham-like complexes and following the idea of Theorem \ref{THM3.1}' proof the game is over. \\

\subsection{The de Rham-like complexes for differential mixed symmetry tensors}\label{derhamcomgen}
We now generalize the de Rham complex where space of differential forms are replaced by space of differential $N$-multi-form. Specifically, we search for the generalization to differential mixed symmetry tensors where the term differential means that, in chart, components are $C^{\infty}$ functions. Before moving on, let us give an useful definition 
\begin{Definition}[Principal $\lambda$-subspace of the $N$-multi-form space]
    Let $\Omega^{p_1\otimes ... \otimes p_N}(M)$ be the $N$-multi-form space on a differential manifold $(M,\mathcal{A})$ of dimension $D$, its principal $\lambda$-subspace $\Omega_{\lambda}^{p_1\otimes ... \otimes p_N}(M)$ is the subspace of mixed symmetry tensors which carry the irreducible representation labelled by the Young tableau $\lambda=\{\lambda_1,...,\lambda_N\}$ such that $\lambda_i \leq p_i\leq D \ \ \forall i\in [1,N]$.
\end{Definition}
Hence, for example, the principal $\{1,1\}$-subspace of $\Omega^{1 \otimes 1}(M)$ is the subspace containing tensors which carry the irreducible representation labelled by the Young tableau $\lambda=\{1,1\}$, i.e. tensors with the indexes symmetry of the graviton field. The Young projector $\Pi_{\lambda}$ furnishes a natural projection from the $N$-multi-form space to its principal subspace 
\begin{equation}
   \Pi_{\lambda}: \Omega^{p_1\otimes ... \otimes p_N}(M) \rightarrow \Omega_{\lambda}^{p_1\otimes ... \otimes p_N}(M)
\end{equation}
which sends a generic $N$-multi-form in a mixed symmetry tensor carrying the irreducible representation labelled by the Young tableau $\lambda$. 
In order to formulate a de Rham-like complex for the differential mixed symmetry tensors, we have at our disposal the $i$-th differentials $d^{(i)}$ with $i \in [1,N]$. The main idea is to follow the structure of the Young lattice, or more precisely, the Hasse diagram of the Young lattice we report in the following Figure \ref{hassedia}. 
\begin{figure}[H]
    \centering
    \includegraphics[width = 5 cm]{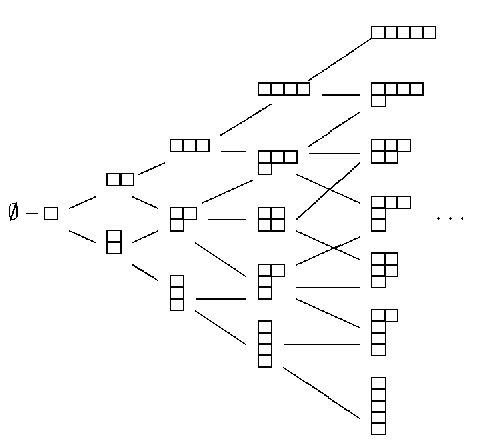}
    \caption{\textit{Hasse diagram of Young's lattice}}
    \label{hassedia}
\end{figure}
In the following and until the end (and also in Appendix \ref{appfasci}), we are going to call the $\{p_1,...,p_N\}$-subspace $\Omega_{\{p_1,...,p_N\}}^{p_1\otimes ... \otimes p_N}(M)$ of $\Omega^{p_1\otimes ... \otimes p_N}(M)$ simply with $\Omega^{p_1\otimes ... \otimes p_N}$ unless explicitly specified. Looking at the Hasse diagram of the Young's lattice we can construct, for every fixed $N$, a $N$-complex where every square commutes. In the case $N=2$, we can draw it as 
\begin{equation}
    \centering
    \tikzset{every picture/.style={line width=0.75pt}} 
\begin{tikzpicture}[x=0.65pt,y=0.65pt,yscale=-1,xscale=1]

\draw    (146.72,137.5) -- (190.76,137.7) ;
\draw [shift={(192.76,137.71)}, rotate = 180.26] [color={rgb, 255:red, 0; green, 0; blue, 0 }  ][line width=0.75]    (10.93,-3.29) .. controls (6.95,-1.4) and (3.31,-0.3) .. (0,0) .. controls (3.31,0.3) and (6.95,1.4) .. (10.93,3.29)   ;
\draw    (146.34,217.5) -- (190.38,217.7) ;
\draw [shift={(192.38,217.71)}, rotate = 180.26] [color={rgb, 255:red, 0; green, 0; blue, 0 }  ][line width=0.75]    (10.93,-3.29) .. controls (6.95,-1.4) and (3.31,-0.3) .. (0,0) .. controls (3.31,0.3) and (6.95,1.4) .. (10.93,3.29)   ;
\draw    (243.34,221.5) -- (287.38,221.7) ;
\draw [shift={(289.38,221.71)}, rotate = 180.26] [color={rgb, 255:red, 0; green, 0; blue, 0 }  ][line width=0.75]    (10.93,-3.29) .. controls (6.95,-1.4) and (3.31,-0.3) .. (0,0) .. controls (3.31,0.3) and (6.95,1.4) .. (10.93,3.29)   ;
\draw [color={rgb, 255:red, 65; green, 117; blue, 5 }  ,draw opacity=1 ]   (111,148) -- (110.77,204.36) ;
\draw [shift={(110.76,206.36)}, rotate = 270.24] [color={rgb, 255:red, 65; green, 117; blue, 5 }  ,draw opacity=1 ][line width=0.75]    (10.93,-3.29) .. controls (6.95,-1.4) and (3.31,-0.3) .. (0,0) .. controls (3.31,0.3) and (6.95,1.4) .. (10.93,3.29)   ;
\draw    (112,230) -- (111.77,286.36) ;
\draw [shift={(111.76,288.36)}, rotate = 270.24] [color={rgb, 255:red, 0; green, 0; blue, 0 }  ][line width=0.75]    (10.93,-3.29) .. controls (6.95,-1.4) and (3.31,-0.3) .. (0,0) .. controls (3.31,0.3) and (6.95,1.4) .. (10.93,3.29)   ;
\draw    (112,314) -- (111.77,370.36) ;
\draw [shift={(111.76,372.36)}, rotate = 270.24] [color={rgb, 255:red, 0; green, 0; blue, 0 }  ][line width=0.75]    (10.93,-3.29) .. controls (6.95,-1.4) and (3.31,-0.3) .. (0,0) .. controls (3.31,0.3) and (6.95,1.4) .. (10.93,3.29)   ;
\draw    (110.62,67) -- (110.38,123.36) ;
\draw [shift={(110.37,125.36)}, rotate = 270.24] [color={rgb, 255:red, 0; green, 0; blue, 0 }  ][line width=0.75]    (10.93,-3.29) .. controls (6.95,-1.4) and (3.31,-0.3) .. (0,0) .. controls (3.31,0.3) and (6.95,1.4) .. (10.93,3.29)   ;
\draw    (202,147) -- (201.77,203.36) ;
\draw [shift={(201.76,205.36)}, rotate = 270.24] [color={rgb, 255:red, 0; green, 0; blue, 0 }  ][line width=0.75]    (10.93,-3.29) .. controls (6.95,-1.4) and (3.31,-0.3) .. (0,0) .. controls (3.31,0.3) and (6.95,1.4) .. (10.93,3.29)   ;
\draw    (203,230) -- (202.77,286.36) ;
\draw [shift={(202.76,288.36)}, rotate = 270.24] [color={rgb, 255:red, 0; green, 0; blue, 0 }  ][line width=0.75]    (10.93,-3.29) .. controls (6.95,-1.4) and (3.31,-0.3) .. (0,0) .. controls (3.31,0.3) and (6.95,1.4) .. (10.93,3.29)   ;
\draw    (301,228) -- (300.77,284.36) ;
\draw [shift={(300.76,286.36)}, rotate = 270.24] [color={rgb, 255:red, 0; green, 0; blue, 0 }  ][line width=0.75]    (10.93,-3.29) .. controls (6.95,-1.4) and (3.31,-0.3) .. (0,0) .. controls (3.31,0.3) and (6.95,1.4) .. (10.93,3.29)   ;
\draw    (149.34,298.5) -- (193.38,298.7) ;
\draw [shift={(195.38,298.71)}, rotate = 180.26] [color={rgb, 255:red, 0; green, 0; blue, 0 }  ][line width=0.75]    (10.93,-3.29) .. controls (6.95,-1.4) and (3.31,-0.3) .. (0,0) .. controls (3.31,0.3) and (6.95,1.4) .. (10.93,3.29)   ;
\draw    (244.34,295.5) -- (288.38,295.7) ;
\draw [shift={(290.38,295.71)}, rotate = 180.26] [color={rgb, 255:red, 0; green, 0; blue, 0 }  ][line width=0.75]    (10.93,-3.29) .. controls (6.95,-1.4) and (3.31,-0.3) .. (0,0) .. controls (3.31,0.3) and (6.95,1.4) .. (10.93,3.29)   ;
\draw    (341.34,298.5) -- (385.38,298.7) ;
\draw [shift={(387.38,298.71)}, rotate = 180.26] [color={rgb, 255:red, 0; green, 0; blue, 0 }  ][line width=0.75]    (10.93,-3.29) .. controls (6.95,-1.4) and (3.31,-0.3) .. (0,0) .. controls (3.31,0.3) and (6.95,1.4) .. (10.93,3.29)   ;
\draw    (202,314) -- (201.77,370.36) ;
\draw [shift={(201.76,372.36)}, rotate = 270.24] [color={rgb, 255:red, 0; green, 0; blue, 0 }  ][line width=0.75]    (10.93,-3.29) .. controls (6.95,-1.4) and (3.31,-0.3) .. (0,0) .. controls (3.31,0.3) and (6.95,1.4) .. (10.93,3.29)   ;
\draw    (302,313) -- (301.77,369.36) ;
\draw [shift={(301.76,371.36)}, rotate = 270.24] [color={rgb, 255:red, 0; green, 0; blue, 0 }  ][line width=0.75]    (10.93,-3.29) .. controls (6.95,-1.4) and (3.31,-0.3) .. (0,0) .. controls (3.31,0.3) and (6.95,1.4) .. (10.93,3.29)   ;
\draw    (400,313) -- (399.77,369.36) ;
\draw [shift={(399.76,371.36)}, rotate = 270.24] [color={rgb, 255:red, 0; green, 0; blue, 0 }  ][line width=0.75]    (10.93,-3.29) .. controls (6.95,-1.4) and (3.31,-0.3) .. (0,0) .. controls (3.31,0.3) and (6.95,1.4) .. (10.93,3.29)   ;
\draw    (428.34,392.5) -- (472.38,392.7) ;
\draw [shift={(474.38,392.71)}, rotate = 180.26] [color={rgb, 255:red, 0; green, 0; blue, 0 }  ][line width=0.75]    (10.93,-3.29) .. controls (6.95,-1.4) and (3.31,-0.3) .. (0,0) .. controls (3.31,0.3) and (6.95,1.4) .. (10.93,3.29)   ;
\draw    (130.34,392.5) -- (174.38,392.7) ;
\draw [shift={(176.38,392.71)}, rotate = 180.26] [color={rgb, 255:red, 0; green, 0; blue, 0 }  ][line width=0.75]    (10.93,-3.29) .. controls (6.95,-1.4) and (3.31,-0.3) .. (0,0) .. controls (3.31,0.3) and (6.95,1.4) .. (10.93,3.29)   ;
\draw    (231.34,392.5) -- (275.38,392.7) ;
\draw [shift={(277.38,392.71)}, rotate = 180.26] [color={rgb, 255:red, 0; green, 0; blue, 0 }  ][line width=0.75]    (10.93,-3.29) .. controls (6.95,-1.4) and (3.31,-0.3) .. (0,0) .. controls (3.31,0.3) and (6.95,1.4) .. (10.93,3.29)   ;
\draw    (326.34,392.5) -- (370.38,392.7) ;
\draw [shift={(372.38,392.71)}, rotate = 180.26] [color={rgb, 255:red, 0; green, 0; blue, 0 }  ][line width=0.75]    (10.93,-3.29) .. controls (6.95,-1.4) and (3.31,-0.3) .. (0,0) .. controls (3.31,0.3) and (6.95,1.4) .. (10.93,3.29)   ;
\draw [color={rgb, 255:red, 74; green, 144; blue, 226 }  ,draw opacity=1 ]   (123.98,62.54) -- (195.48,125.22) ;
\draw [shift={(196.98,126.54)}, rotate = 221.24] [color={rgb, 255:red, 74; green, 144; blue, 226 }  ,draw opacity=1 ][line width=0.75]    (10.93,-3.29) .. controls (6.95,-1.4) and (3.31,-0.3) .. (0,0) .. controls (3.31,0.3) and (6.95,1.4) .. (10.93,3.29)   ;
\draw [color={rgb, 255:red, 74; green, 144; blue, 226 }  ,draw opacity=1 ]   (217.98,143.54) -- (289.48,206.22) ;
\draw [shift={(290.98,207.54)}, rotate = 221.24] [color={rgb, 255:red, 74; green, 144; blue, 226 }  ,draw opacity=1 ][line width=0.75]    (10.93,-3.29) .. controls (6.95,-1.4) and (3.31,-0.3) .. (0,0) .. controls (3.31,0.3) and (6.95,1.4) .. (10.93,3.29)   ;
\draw [color={rgb, 255:red, 74; green, 144; blue, 226 }  ,draw opacity=1 ]   (315.98,225.54) -- (387.48,288.22) ;
\draw [shift={(388.98,289.54)}, rotate = 221.24] [color={rgb, 255:red, 74; green, 144; blue, 226 }  ,draw opacity=1 ][line width=0.75]    (10.93,-3.29) .. controls (6.95,-1.4) and (3.31,-0.3) .. (0,0) .. controls (3.31,0.3) and (6.95,1.4) .. (10.93,3.29)   ;
\draw [color={rgb, 255:red, 74; green, 144; blue, 226 }  ,draw opacity=1 ]   (408.98,305.54) -- (486.21,372.85) ;
\draw [shift={(487.72,374.16)}, rotate = 221.07] [color={rgb, 255:red, 74; green, 144; blue, 226 }  ,draw opacity=1 ][line width=0.75]    (10.93,-3.29) .. controls (6.95,-1.4) and (3.31,-0.3) .. (0,0) .. controls (3.31,0.3) and (6.95,1.4) .. (10.93,3.29)   ;
\draw [color={rgb, 255:red, 245; green, 166; blue, 35 }  ,draw opacity=1 ]   (121.98,142.54) -- (193.48,205.22) ;
\draw [shift={(194.98,206.54)}, rotate = 221.24] [color={rgb, 255:red, 245; green, 166; blue, 35 }  ,draw opacity=1 ][line width=0.75]    (10.93,-3.29) .. controls (6.95,-1.4) and (3.31,-0.3) .. (0,0) .. controls (3.31,0.3) and (6.95,1.4) .. (10.93,3.29)   ;
\draw [color={rgb, 255:red, 245; green, 166; blue, 35 }  ,draw opacity=1 ]   (219,225.5) -- (289.51,288.17) ;
\draw [shift={(291,289.5)}, rotate = 221.63] [color={rgb, 255:red, 245; green, 166; blue, 35 }  ,draw opacity=1 ][line width=0.75]    (10.93,-3.29) .. controls (6.95,-1.4) and (3.31,-0.3) .. (0,0) .. controls (3.31,0.3) and (6.95,1.4) .. (10.93,3.29)   ;
\draw [color={rgb, 255:red, 245; green, 166; blue, 35 }  ,draw opacity=1 ]   (312.98,306.54) -- (384.48,369.22) ;
\draw [shift={(385.98,370.54)}, rotate = 221.24] [color={rgb, 255:red, 245; green, 166; blue, 35 }  ,draw opacity=1 ][line width=0.75]    (10.93,-3.29) .. controls (6.95,-1.4) and (3.31,-0.3) .. (0,0) .. controls (3.31,0.3) and (6.95,1.4) .. (10.93,3.29)   ;
\draw [color={rgb, 255:red, 65; green, 117; blue, 5 }  ,draw opacity=1 ]   (121.98,224.54) -- (193.48,287.22) ;
\draw [shift={(194.98,288.54)}, rotate = 221.24] [color={rgb, 255:red, 65; green, 117; blue, 5 }  ,draw opacity=1 ][line width=0.75]    (10.93,-3.29) .. controls (6.95,-1.4) and (3.31,-0.3) .. (0,0) .. controls (3.31,0.3) and (6.95,1.4) .. (10.93,3.29)   ;
\draw [color={rgb, 255:red, 65; green, 117; blue, 5 }  ,draw opacity=1 ]   (216.98,308.54) -- (288.48,371.22) ;
\draw [shift={(289.98,372.54)}, rotate = 221.24] [color={rgb, 255:red, 65; green, 117; blue, 5 }  ,draw opacity=1 ][line width=0.75]    (10.93,-3.29) .. controls (6.95,-1.4) and (3.31,-0.3) .. (0,0) .. controls (3.31,0.3) and (6.95,1.4) .. (10.93,3.29)   ;
\draw [color={rgb, 255:red, 144; green, 19; blue, 254 }  ,draw opacity=1 ]   (120.98,312.54) -- (192.48,375.22) ;
\draw [shift={(193.98,376.54)}, rotate = 221.24] [color={rgb, 255:red, 144; green, 19; blue, 254 }  ,draw opacity=1 ][line width=0.75]    (10.93,-3.29) .. controls (6.95,-1.4) and (3.31,-0.3) .. (0,0) .. controls (3.31,0.3) and (6.95,1.4) .. (10.93,3.29)   ;
\draw    (112,402) -- (111.77,458.36) ;
\draw [shift={(111.76,460.36)}, rotate = 270.24] [color={rgb, 255:red, 0; green, 0; blue, 0 }  ][line width=0.75]    (10.93,-3.29) .. controls (6.95,-1.4) and (3.31,-0.3) .. (0,0) .. controls (3.31,0.3) and (6.95,1.4) .. (10.93,3.29)   ;
\draw    (414.34,474.5) -- (458.38,474.7) ;
\draw [shift={(460.38,474.71)}, rotate = 180.26] [color={rgb, 255:red, 0; green, 0; blue, 0 }  ][line width=0.75]    (10.93,-3.29) .. controls (6.95,-1.4) and (3.31,-0.3) .. (0,0) .. controls (3.31,0.3) and (6.95,1.4) .. (10.93,3.29)   ;
\draw    (143.34,475.5) -- (187.38,475.7) ;
\draw [shift={(189.38,475.71)}, rotate = 180.26] [color={rgb, 255:red, 0; green, 0; blue, 0 }  ][line width=0.75]    (10.93,-3.29) .. controls (6.95,-1.4) and (3.31,-0.3) .. (0,0) .. controls (3.31,0.3) and (6.95,1.4) .. (10.93,3.29)   ;
\draw    (244.34,475.5) -- (288.38,475.7) ;
\draw [shift={(290.38,475.71)}, rotate = 180.26] [color={rgb, 255:red, 0; green, 0; blue, 0 }  ][line width=0.75]    (10.93,-3.29) .. controls (6.95,-1.4) and (3.31,-0.3) .. (0,0) .. controls (3.31,0.3) and (6.95,1.4) .. (10.93,3.29)   ;
\draw    (331.34,475.5) -- (375.38,475.7) ;
\draw [shift={(377.38,475.71)}, rotate = 180.26] [color={rgb, 255:red, 0; green, 0; blue, 0 }  ][line width=0.75]    (10.93,-3.29) .. controls (6.95,-1.4) and (3.31,-0.3) .. (0,0) .. controls (3.31,0.3) and (6.95,1.4) .. (10.93,3.29)   ;
\draw    (206,401) -- (205.77,457.36) ;
\draw [shift={(205.76,459.36)}, rotate = 270.24] [color={rgb, 255:red, 0; green, 0; blue, 0 }  ][line width=0.75]    (10.93,-3.29) .. controls (6.95,-1.4) and (3.31,-0.3) .. (0,0) .. controls (3.31,0.3) and (6.95,1.4) .. (10.93,3.29)   ;
\draw [color={rgb, 255:red, 74; green, 144; blue, 226 }  ,draw opacity=1 ]   (512.97,397.41) -- (580.43,456.39) ;
\draw [shift={(581.94,457.7)}, rotate = 221.16] [color={rgb, 255:red, 74; green, 144; blue, 226 }  ,draw opacity=1 ][line width=0.75]    (10.93,-3.29) .. controls (6.95,-1.4) and (3.31,-0.3) .. (0,0) .. controls (3.31,0.3) and (6.95,1.4) .. (10.93,3.29)   ;
\draw    (521.34,475.5) -- (565.38,475.7) ;
\draw [shift={(567.38,475.71)}, rotate = 180.26] [color={rgb, 255:red, 0; green, 0; blue, 0 }  ][line width=0.75]    (10.93,-3.29) .. controls (6.95,-1.4) and (3.31,-0.3) .. (0,0) .. controls (3.31,0.3) and (6.95,1.4) .. (10.93,3.29)   ;
\draw [color={rgb, 255:red, 245; green, 166; blue, 35 }  ,draw opacity=1 ]   (413.98,397.54) -- (485.48,460.22) ;
\draw [shift={(486.98,461.54)}, rotate = 221.24] [color={rgb, 255:red, 245; green, 166; blue, 35 }  ,draw opacity=1 ][line width=0.75]    (10.93,-3.29) .. controls (6.95,-1.4) and (3.31,-0.3) .. (0,0) .. controls (3.31,0.3) and (6.95,1.4) .. (10.93,3.29)   ;

\draw (104.38,126.4) node [anchor=north west][inner sep=0.75pt]    {$\Omega ^{1\otimes 0}$};
\draw (152.81,113.4) node [anchor=north west][inner sep=0.75pt]    {$d^{( 2)}$};
\draw (196.69,126.4) node [anchor=north west][inner sep=0.75pt]    {$\Omega ^{1\otimes 1}$};
\draw (104,206.4) node [anchor=north west][inner sep=0.75pt]    {$\Omega ^{2\otimes 0}$};
\draw (152.43,193.4) node [anchor=north west][inner sep=0.75pt]    {$d^{( 2)}$};
\draw (199.31,207.4) node [anchor=north west][inner sep=0.75pt]    {$\Omega ^{2\otimes 1}$};
\draw (106,290.4) node [anchor=north west][inner sep=0.75pt]    {$\Omega ^{3\otimes 0}$};
\draw (249.43,197.4) node [anchor=north west][inner sep=0.75pt]    {$d^{( 2)}$};
\draw (391.81,386.4) node [anchor=north west][inner sep=0.75pt]    {$...$};
\draw (295.81,386.4) node [anchor=north west][inner sep=0.75pt]    {$...$};
\draw (116,163.4) node [anchor=north west][inner sep=0.75pt]    {$d^{( 1)}$};
\draw (117,245.4) node [anchor=north west][inner sep=0.75pt]    {$d^{( 1)}$};
\draw (117,329.4) node [anchor=north west][inner sep=0.75pt]    {$d^{( 1)}$};
\draw (104,45.4) node [anchor=north west][inner sep=0.75pt]    {$\Omega ^{0\otimes 0}$};
\draw (115.62,82.4) node [anchor=north west][inner sep=0.75pt]    {$d^{( 1)}$};
\draw (207,162.4) node [anchor=north west][inner sep=0.75pt]    {$d^{( 1)}$};
\draw (297.31,208.4) node [anchor=north west][inner sep=0.75pt]    {$\Omega ^{2\otimes 2}$};
\draw (198,289.4) node [anchor=north west][inner sep=0.75pt]    {$\Omega ^{3\otimes 1}$};
\draw (208,245.4) node [anchor=north west][inner sep=0.75pt]    {$d^{( 1)}$};
\draw (294,287.4) node [anchor=north west][inner sep=0.75pt]    {$\Omega ^{3\otimes 2}$};
\draw (306,243.4) node [anchor=north west][inner sep=0.75pt]    {$d^{( 1)}$};
\draw (155.43,274.4) node [anchor=north west][inner sep=0.75pt]    {$d^{( 2)}$};
\draw (248.43,273.4) node [anchor=north west][inner sep=0.75pt]    {$d^{( 2)}$};
\draw (347.43,274.4) node [anchor=north west][inner sep=0.75pt]    {$d^{( 2)}$};
\draw (390,289.4) node [anchor=north west][inner sep=0.75pt]    {$\Omega ^{3\otimes 3}$};
\draw (207,329.4) node [anchor=north west][inner sep=0.75pt]    {$d^{( 1)}$};
\draw (307,328.4) node [anchor=north west][inner sep=0.75pt]    {$d^{( 1)}$};
\draw (405,328.4) node [anchor=north west][inner sep=0.75pt]    {$d^{( 1)}$};
\draw (435.43,368.4) node [anchor=north west][inner sep=0.75pt]    {$d^{( 2)}$};
\draw (194.81,386.4) node [anchor=north west][inner sep=0.75pt]    {$...$};
\draw (107.81,386.4) node [anchor=north west][inner sep=0.75pt]    {$...$};
\draw (489.72,386.4) node [anchor=north west][inner sep=0.75pt]    {$...$};
\draw (137.43,368.4) node [anchor=north west][inner sep=0.75pt]    {$d^{( 2)}$};
\draw (238.43,367.4) node [anchor=north west][inner sep=0.75pt]    {$d^{( 2)}$};
\draw (333.43,368.4) node [anchor=north west][inner sep=0.75pt]    {$d^{( 2)}$};
\draw (117,417.4) node [anchor=north west][inner sep=0.75pt]    {$d^{( 1)}$};
\draw (104.38,462.4) node [anchor=north west][inner sep=0.75pt]    {$\Omega ^{D\otimes 0}$};
\draw (389.81,468.4) node [anchor=north west][inner sep=0.75pt]    {$...$};
\draw (301.81,468.4) node [anchor=north west][inner sep=0.75pt]    {$...$};
\draw (420.43,450.4) node [anchor=north west][inner sep=0.75pt]    {$d^{( 2)}$};
\draw (149.43,451.4) node [anchor=north west][inner sep=0.75pt]    {$d^{( 2)}$};
\draw (250.43,450.4) node [anchor=north west][inner sep=0.75pt]    {$d^{( 2)}$};
\draw (337.43,451.4) node [anchor=north west][inner sep=0.75pt]    {$d^{( 2)}$};
\draw (199,462.4) node [anchor=north west][inner sep=0.75pt]    {$\Omega ^{D\otimes 1}$};
\draw (211,416.4) node [anchor=north west][inner sep=0.75pt]    {$d^{( 1)}$};
\draw (571,464.4) node [anchor=north west][inner sep=0.75pt]    {$\Omega ^{D\otimes D}$};
\draw (527.43,451.4) node [anchor=north west][inner sep=0.75pt]    {$d^{( 2)}$};
\draw (465,464.4) node [anchor=north west][inner sep=0.75pt]    {$\Omega ^{D\otimes D-1}$};
\end{tikzpicture}.
\label{derhamcom2multi}
\end{equation}
We note that the first column of this diagram is exactly the standard de Rham complex with differential $d^{(1)}$ thanks to the canonical isomorphism $\Omega^{p \otimes 0} \cong \Omega^p$ for every $p\leq D$.\\
However, for the general $N$ case it is not so easy to draw the diagram, since we need $N$ independent directions to taken into account all the ramifications; however, also in this case there is a standard de Rham complex with differential $d^{(1)}$ thanks to the canonical isomorphism $\Omega^{p \otimes 0 \otimes ... \otimes 0} \cong \Omega^p$ for every $p\leq D$. \\
In order to reproduce and extend Theorem \ref{THM3.1}, a necessary property we want to emulate of the de Rham complex is the fact that given a form $B \in \Omega^p$ for same $p$, the form $dB \in \Omega^{p+1}$ is its field strength. In other words, the space $\Omega^{p+1}$ contains the field strengths of all the forms in $\Omega^{p}$. In the case of a $N$-multi-form whose writing, in a chart, is mixed symmetry tensor carrying the irreducible representation associated to the Young tableau $\lambda$, there is a unique unambiguous way to construct a field strength, that is, acting with the composition of all the $i$-th differential, i.e. Definition \ref{def2.5.10}. Therefore we can define
\begin{Definition}[De Rham-like differential]
    Given the $N$-multi-form space $\Omega^{p_1 \otimes ... \otimes p_N}(M)$ and the $i$-th differential $d^{(i)}$ with $i \in [1,N]$ the de Rham-like differential $d$ is given by the composition of all the $i$-th differentials
    \begin{equation}
        \delta^{(N)}:=d^{(N)} \circ d^{(N-1)} \circ ... \circ d^{(2)} \circ d^{(1)}.
    \end{equation}
\end{Definition}
We have easely the following
\begin{Proposition}[Fundamental property of the de Rham-like differential]
    The de Rham-like differential $\delta^{(N)}$ squares to zero
    \begin{equation}
        \delta^{(N)} \circ \delta^{(N)} =0.
    \end{equation}
\end{Proposition}
\begin{proof}
    Since the $i$-th differentials satisfy $d^{(i)} \circ d^{(i)}=0 \ \ \forall i \in [1,N]$ and $d^{(i)} \circ d^{(j)}=d^{(j)} \circ d^{(i)} \ \ \forall i \neq j$ we get
    \begin{equation}
    \begin{aligned}
        \delta^{(N)} \circ \delta^{(N)}&=d^{(N)} \circ d^{(N-1)} \circ ... \circ d^{(2)} \circ d^{(1)} \circ d^{(N)} \circ d^{(N-1)} \circ ... \circ d^{(2)} \circ d^{(1)}=\\
        &=d^{(N)} \circ d^{(N)} \circ d^{(N-1)} \circ d^{(N-1)} \circ ... d^{(1)} \circ d^{(1)} =0.
    \end{aligned}
    \end{equation}
\end{proof}
Let us focus on the case $N=2$ where the de Rham-like differential is given by $\delta^{(2)}:=d^{(2)}\circ d^{(1)}$. Therefore we have the 2-de Rham complex
\begin{Definition}[2-de Rham-like complex]
    A de Rham-like complex for the biform space on a differential manifold $(M,\mathcal{A})$ such that $D=dim(M)$, or 2-de Rham complex, is the cochain complex with differential given by the de Rham-like differential $\delta^{(2)}$.
\end{Definition}
We stress that in the perspective of this definition, the standard de Rham complex given by the first column of diagram \ref{derhamcom2multi} could be defined as a 1-de Rham-like complex with differential $d^{(1)}=\delta^{(1)}$. \\
In diagram \ref{derhamcom2multi} some 2-de Rham-like complexes are highlighted with colored arrows. However, we note that the 2-de Rham-like complex with blue arrows have length\footnote{We mean the number of arrows between non-trivial modules.} $D$ while the one with orange arrows $D-1$ and go on. Hence, we are going to consider augmented cochain complexes in order to have cochain complexes all of the same length $D$ and all starting from $\Omega^{0 \otimes 0}$. Therefore we have the following 
\begin{Definition}[$k$-augmented 2-de Rham-like complex]
  We define the $k$-augmented 2-de Rham-like complex as the 2-de Rham-like complex with length $D-k$ augmented by the first $k$ terms of the  1-de Rham-like complex with differential $\delta^{(1)}$, or equivalently, as the 2-de Rham-like complex with length $D-k$ augmented by the first $k$ terms of the standard de Rham complex with differential $d^{(1)}$.
\end{Definition}
To give an example of $k$-augmented 2-de Rham-like complex let us consider the 2-de Rham complex with orange arrows, it has length $D-1$ for every fixed $D$ and so we have the 
1-augmented 2-de Rham-like complex
\begin{equation}
\begin{tikzpicture}[x=0.65pt,y=0.75pt,yscale=-1,xscale=1]

\draw    (56.62,53) -- (122.33,52.78) ;
\draw [shift={(124.33,52.77)}, rotate = 179.81] [color={rgb, 255:red, 0; green, 0; blue, 0 }  ][line width=0.75]    (10.93,-3.29) .. controls (6.95,-1.4) and (3.31,-0.3) .. (0,0) .. controls (3.31,0.3) and (6.95,1.4) .. (10.93,3.29)   ;
\draw [color={rgb, 255:red, 245; green, 166; blue, 35 }  ,draw opacity=1 ]   (168.62,54) -- (234.33,53.78) ;
\draw [shift={(236.33,53.77)}, rotate = 179.81] [color={rgb, 255:red, 245; green, 166; blue, 35 }  ,draw opacity=1 ][line width=0.75]    (10.93,-3.29) .. controls (6.95,-1.4) and (3.31,-0.3) .. (0,0) .. controls (3.31,0.3) and (6.95,1.4) .. (10.93,3.29)   ;
\draw [color={rgb, 255:red, 245; green, 166; blue, 35 }  ,draw opacity=1 ]   (283.62,54) -- (349.33,53.78) ;
\draw [shift={(351.33,53.77)}, rotate = 179.81] [color={rgb, 255:red, 245; green, 166; blue, 35 }  ,draw opacity=1 ][line width=0.75]    (10.93,-3.29) .. controls (6.95,-1.4) and (3.31,-0.3) .. (0,0) .. controls (3.31,0.3) and (6.95,1.4) .. (10.93,3.29)   ;
\draw [color={rgb, 255:red, 245; green, 166; blue, 35 }  ,draw opacity=1 ]   (395.62,54) -- (461.33,53.78) ;
\draw [shift={(463.33,53.77)}, rotate = 179.81] [color={rgb, 255:red, 245; green, 166; blue, 35 }  ,draw opacity=1 ][line width=0.75]    (10.93,-3.29) .. controls (6.95,-1.4) and (3.31,-0.3) .. (0,0) .. controls (3.31,0.3) and (6.95,1.4) .. (10.93,3.29)   ;
\draw [color={rgb, 255:red, 245; green, 166; blue, 35 }  ,draw opacity=1 ]   (488.62,54) -- (554.33,53.78) ;
\draw [shift={(556.33,53.77)}, rotate = 179.81] [color={rgb, 255:red, 245; green, 166; blue, 35 }  ,draw opacity=1 ][line width=0.75]    (10.93,-3.29) .. controls (6.95,-1.4) and (3.31,-0.3) .. (0,0) .. controls (3.31,0.3) and (6.95,1.4) .. (10.93,3.29)   ;

\draw (129.38,45.4) node [anchor=north west][inner sep=0.75pt]    {$\Omega ^{1\otimes 0}$};
\draw (470.81,55.4) node [anchor=north west][inner sep=0.75pt]    {$...$};
\draw (16,45.4) node [anchor=north west][inner sep=0.75pt]    {$\Omega ^{0\otimes 0}$};
\draw (77.62,34.4) node [anchor=north west][inner sep=0.75pt]    {$\delta^{( 1)}$};
\draw (241.38,45.4) node [anchor=north west][inner sep=0.75pt]    {$\Omega ^{2\otimes 1}$};
\draw (193.62,35.4) node [anchor=north west][inner sep=0.75pt]    {$\delta^{( 2)}$};
\draw (356.38,45.4) node [anchor=north west][inner sep=0.75pt]    {$\Omega ^{3\otimes 2}$};
\draw (308.62,35.4) node [anchor=north west][inner sep=0.75pt]    {$\delta^{( 2)}$};
\draw (420.62,35.4) node [anchor=north west][inner sep=0.75pt]    {$\delta^{( 2)}$};
\draw (513.62,35.4) node [anchor=north west][inner sep=0.75pt]    {$\delta^{( 2)}$};
\draw (560.38,45.4) node [anchor=north west][inner sep=0.75pt]    {$\Omega ^{D\otimes D-1}$};
\end{tikzpicture}
;
\end{equation}
in an analog way, considering the 2-de Rham complex with green arrows we have the 2-augmented 2-de Rham-like complex
\begin{equation}
    \begin{tikzpicture}[x=0.65pt,y=0.75pt,yscale=-1,xscale=1]

\draw    (56.62,53) -- (122.33,52.78) ;
\draw [shift={(124.33,52.77)}, rotate = 179.81] [color={rgb, 255:red, 0; green, 0; blue, 0 }  ][line width=0.75]    (10.93,-3.29) .. controls (6.95,-1.4) and (3.31,-0.3) .. (0,0) .. controls (3.31,0.3) and (6.95,1.4) .. (10.93,3.29)   ;
\draw [color={rgb, 255:red, 65; green, 117; blue, 5 }  ,draw opacity=1 ]   (283.62,54) -- (349.33,53.78) ;
\draw [shift={(351.33,53.77)}, rotate = 179.81] [color={rgb, 255:red, 65; green, 117; blue, 5 }  ,draw opacity=1 ][line width=0.75]    (10.93,-3.29) .. controls (6.95,-1.4) and (3.31,-0.3) .. (0,0) .. controls (3.31,0.3) and (6.95,1.4) .. (10.93,3.29)   ;
\draw [color={rgb, 255:red, 65; green, 117; blue, 5 }  ,draw opacity=1 ]   (395.62,54) -- (461.33,53.78) ;
\draw [shift={(463.33,53.77)}, rotate = 179.81] [color={rgb, 255:red, 65; green, 117; blue, 5 }  ,draw opacity=1 ][line width=0.75]    (10.93,-3.29) .. controls (6.95,-1.4) and (3.31,-0.3) .. (0,0) .. controls (3.31,0.3) and (6.95,1.4) .. (10.93,3.29)   ;
\draw [color={rgb, 255:red, 65; green, 117; blue, 5 }  ,draw opacity=1 ]   (488.62,54) -- (554.33,53.78) ;
\draw [shift={(556.33,53.77)}, rotate = 179.81] [color={rgb, 255:red, 65; green, 117; blue, 5 }  ,draw opacity=1 ][line width=0.75]    (10.93,-3.29) .. controls (6.95,-1.4) and (3.31,-0.3) .. (0,0) .. controls (3.31,0.3) and (6.95,1.4) .. (10.93,3.29)   ;
\draw    (169.62,54) -- (235.33,53.78) ;
\draw [shift={(237.33,53.77)}, rotate = 179.81] [color={rgb, 255:red, 0; green, 0; blue, 0 }  ][line width=0.75]    (10.93,-3.29) .. controls (6.95,-1.4) and (3.31,-0.3) .. (0,0) .. controls (3.31,0.3) and (6.95,1.4) .. (10.93,3.29)   ;

\draw (129.38,45.4) node [anchor=north west][inner sep=0.75pt]    {$\Omega ^{1\otimes 0}$};
\draw (470.81,55.4) node [anchor=north west][inner sep=0.75pt]    {$...$};
\draw (16,45.4) node [anchor=north west][inner sep=0.75pt]    {$\Omega ^{0\otimes 0}$};
\draw (77.62,34.4) node [anchor=north west][inner sep=0.75pt]    {$\delta^{( 1)}$};
\draw (241.38,45.4) node [anchor=north west][inner sep=0.75pt]    {$\Omega ^{2\otimes 0}$};
\draw (356.38,45.4) node [anchor=north west][inner sep=0.75pt]    {$\Omega ^{3\otimes 1}$};
\draw (308.62,35.4) node [anchor=north west][inner sep=0.75pt]    {$\delta^{( 2)}$};
\draw (420.62,35.4) node [anchor=north west][inner sep=0.75pt]    {$\delta^{( 2)}$};
\draw (513.62,35.4) node [anchor=north west][inner sep=0.75pt]    {$\delta^{( 2)}$};
\draw (560.38,45.4) node [anchor=north west][inner sep=0.75pt]    {$\Omega ^{D\otimes D-2}$};
\draw (190.62,35.4) node [anchor=north west][inner sep=0.75pt]    {$\delta^{( 1)}$};
\end{tikzpicture}.
\end{equation}
At this point we can generalize our definitions to the general $N$ case; therefore, we have the following
\begin{Definition}[$N$-de Rham-like complex]
    A de Rham-like complex for the $N$-multi-form space on a differential manifold $(M,\mathcal{A})$ such that $D=dim(M)$, or $N$-de Rham complex, is the cochain complex with differential given by the de Rham-like differential $\delta^{( N)}$.
\end{Definition}
and, for the same reasons as before, their augmented cochain complexes
\begin{Definition}[$(k_1,...,k_{N-1})$-augmented $N$-de Rham-like complex]\label{Naugcomp}
The $(k_1,...,k_{N-1})$-augmented $N$-de Rham-like complex is the $N$-de Rham-like complex with length $D-(k_1+...+k_{N-1})$ augmented by the first $k_1+...+k_{N-1}$ terms of the $(k_1,...,k_{N-2})$-augmented $(N-1)$-de Rham-like complex.
\end{Definition}
Let us see consider the following clarifying example.
\subsubsection*{Example 1: the augmenting of the $N$-de Rham-like complex passing for $\Omega^{q_1 \otimes ... \otimes q_N}$}
Let us consider the $N$-de Rham-like complex passing for $\Omega^{q_1 \otimes ... \otimes q_N}$, this is given by
\begin{equation}
\tikzset{every picture/.style={line width=0.75pt}} 
\begin{tikzpicture}[x=0.65pt,y=0.75pt,yscale=-1,xscale=1]

\draw    (159.82,184.28) -- (197.35,184.28) ;
\draw [shift={(199.35,184.28)}, rotate = 180] [color={rgb, 255:red, 0; green, 0; blue, 0 }  ][line width=0.75]    (10.93,-3.29) .. controls (6.95,-1.4) and (3.31,-0.3) .. (0,0) .. controls (3.31,0.3) and (6.95,1.4) .. (10.93,3.29)   ;
\draw    (228.82,184.28) -- (266.35,184.28) ;
\draw [shift={(268.35,184.28)}, rotate = 180] [color={rgb, 255:red, 0; green, 0; blue, 0 }  ][line width=0.75]    (10.93,-3.29) .. controls (6.95,-1.4) and (3.31,-0.3) .. (0,0) .. controls (3.31,0.3) and (6.95,1.4) .. (10.93,3.29)   ;
\draw    (350.82,184.28) -- (388.35,184.28) ;
\draw [shift={(390.35,184.28)}, rotate = 180] [color={rgb, 255:red, 0; green, 0; blue, 0 }  ][line width=0.75]    (10.93,-3.29) .. controls (6.95,-1.4) and (3.31,-0.3) .. (0,0) .. controls (3.31,0.3) and (6.95,1.4) .. (10.93,3.29)   ;
\draw    (420.82,184.28) -- (458.35,184.28) ;
\draw [shift={(460.35,184.28)}, rotate = 180] [color={rgb, 255:red, 0; green, 0; blue, 0 }  ][line width=0.75]    (10.93,-3.29) .. controls (6.95,-1.4) and (3.31,-0.3) .. (0,0) .. controls (3.31,0.3) and (6.95,1.4) .. (10.93,3.29)   ;

\draw (9.97,172.65) node [anchor=north west][inner sep=0.75pt]    {$\Omega ^{q_{1} -q_{N} \otimes ...\otimes q_{N-1} -q_{N} \otimes 0}$};
\draw (173.22,161.54) node [anchor=north west][inner sep=0.75pt]    {$\delta ^{( N)}$};
\draw (208.92,180.37) node [anchor=north west][inner sep=0.75pt]    {$...$};
\draw (233.12,161.54) node [anchor=north west][inner sep=0.75pt]    {$\delta ^{( N)}$};
\draw (277,172.65) node [anchor=north west][inner sep=0.75pt]    {$\Omega ^{q_{1} \otimes ...\otimes q_{N}}$};
\draw (355.12,160.54) node [anchor=north west][inner sep=0.75pt]    {$\delta ^{( N)}$};
\draw (397.92,180.37) node [anchor=north west][inner sep=0.75pt]    {$...$};
\draw (425.12,160.54) node [anchor=north west][inner sep=0.75pt]    {$\delta ^{( N)}$};
\draw (466.59,172.68) node [anchor=north west][inner sep=0.75pt]    {$\Omega ^{D\otimes D-( q_{1} -q_{2}) \otimes ...\otimes D-( q_{1} -q_{N})}$};
\end{tikzpicture}.
\label{Nderhames}
\end{equation}
This complex has length $D-(q_1-q_N)$ and we want to construct the $(q_1-q_2,...,q_{N-1}-q_N)$-augmented $N$-de Rham-like complex. Therefore we need to add the first $\sum_{i=1}^{N-1} q_i-q_{i+1}=q_1-q_N$ terms of the $(q_1-q_2,...,q_{N-2}-q_{N-1})$-augmented $(N-1)$-de Rham-like complex. These first $q_1-q_N$ terms are 
\begin{equation}
\tikzset{every picture/.style={line width=0.75pt}} 
\begin{tikzpicture}[x=0.75pt,y=0.75pt,yscale=-1,xscale=1]

\draw    (192.82,53.28) -- (230.35,53.28) ;
\draw [shift={(232.35,53.28)}, rotate = 180] [color={rgb, 255:red, 0; green, 0; blue, 0 }  ][line width=0.75]    (10.93,-3.29) .. controls (6.95,-1.4) and (3.31,-0.3) .. (0,0) .. controls (3.31,0.3) and (6.95,1.4) .. (10.93,3.29)   ;
\draw    (266.82,53.28) -- (304.35,53.28) ;
\draw [shift={(306.35,53.28)}, rotate = 180] [color={rgb, 255:red, 0; green, 0; blue, 0 }  ][line width=0.75]    (10.93,-3.29) .. controls (6.95,-1.4) and (3.31,-0.3) .. (0,0) .. controls (3.31,0.3) and (6.95,1.4) .. (10.93,3.29)   ;
\draw    (348.82,65.28) -- (348.68,101.37) ;
\draw [shift={(348.67,103.37)}, rotate = 270.22] [color={rgb, 255:red, 0; green, 0; blue, 0 }  ][line width=0.75]    (10.93,-3.29) .. controls (6.95,-1.4) and (3.31,-0.3) .. (0,0) .. controls (3.31,0.3) and (6.95,1.4) .. (10.93,3.29)   ;
\draw    (306.67,123.37) -- (266.67,123.37) ;
\draw [shift={(264.67,123.37)}, rotate = 360] [color={rgb, 255:red, 0; green, 0; blue, 0 }  ][line width=0.75]    (10.93,-3.29) .. controls (6.95,-1.4) and (3.31,-0.3) .. (0,0) .. controls (3.31,0.3) and (6.95,1.4) .. (10.93,3.29)   ;
\draw    (236.67,123.37) -- (196.67,123.37) ;
\draw [shift={(194.67,123.37)}, rotate = 360] [color={rgb, 255:red, 0; green, 0; blue, 0 }  ][line width=0.75]    (10.93,-3.29) .. controls (6.95,-1.4) and (3.31,-0.3) .. (0,0) .. controls (3.31,0.3) and (6.95,1.4) .. (10.93,3.29)   ;
\draw    (133.82,135.28) -- (133.68,171.37) ;
\draw [shift={(133.67,173.37)}, rotate = 270.22] [color={rgb, 255:red, 0; green, 0; blue, 0 }  ][line width=0.75]    (10.93,-3.29) .. controls (6.95,-1.4) and (3.31,-0.3) .. (0,0) .. controls (3.31,0.3) and (6.95,1.4) .. (10.93,3.29)   ;
\draw    (193.82,193.28) -- (231.35,193.28) ;
\draw [shift={(233.35,193.28)}, rotate = 180] [color={rgb, 255:red, 0; green, 0; blue, 0 }  ][line width=0.75]    (10.93,-3.29) .. controls (6.95,-1.4) and (3.31,-0.3) .. (0,0) .. controls (3.31,0.3) and (6.95,1.4) .. (10.93,3.29)   ;
\draw    (267.82,193.28) -- (305.35,193.28) ;
\draw [shift={(307.35,193.28)}, rotate = 180] [color={rgb, 255:red, 0; green, 0; blue, 0 }  ][line width=0.75]    (10.93,-3.29) .. controls (6.95,-1.4) and (3.31,-0.3) .. (0,0) .. controls (3.31,0.3) and (6.95,1.4) .. (10.93,3.29)   ;
\draw    (309.67,264.37) -- (269.67,264.37) ;
\draw [shift={(267.67,264.37)}, rotate = 360] [color={rgb, 255:red, 0; green, 0; blue, 0 }  ][line width=0.75]    (10.93,-3.29) .. controls (6.95,-1.4) and (3.31,-0.3) .. (0,0) .. controls (3.31,0.3) and (6.95,1.4) .. (10.93,3.29)   ;
\draw    (231.67,264.37) -- (191.67,264.37) ;
\draw [shift={(189.67,264.37)}, rotate = 360] [color={rgb, 255:red, 0; green, 0; blue, 0 }  ][line width=0.75]    (10.93,-3.29) .. controls (6.95,-1.4) and (3.31,-0.3) .. (0,0) .. controls (3.31,0.3) and (6.95,1.4) .. (10.93,3.29)   ;

\draw (113.97,40.65) node [anchor=north west][inner sep=0.75pt]    {$\Omega ^{0\otimes ...\otimes 0\otimes 0}$};
\draw (197.12,29.54) node [anchor=north west][inner sep=0.75pt]    {$\delta ^{( 1)}$};
\draw (243.92,50.37) node [anchor=north west][inner sep=0.75pt]    {$...$};
\draw (271.12,29.54) node [anchor=north west][inner sep=0.75pt]    {$\delta ^{( 1)}$};
\draw (315,39.65) node [anchor=north west][inner sep=0.75pt]    {$\Omega ^{q_{1} -q_{2} \otimes 0\otimes ...\otimes 0}$};
\draw (314,110.65) node [anchor=north west][inner sep=0.75pt]    {$\Omega ^{q_{1} -q_{2} +1\otimes 1\otimes 0\otimes ...\otimes 0}$};
\draw (356.12,73.54) node [anchor=north west][inner sep=0.75pt]    {$\delta ^{( 2)}$};
\draw (279.12,99.54) node [anchor=north west][inner sep=0.75pt]    {$\delta ^{( 2)}$};
\draw (243.92,120.37) node [anchor=north west][inner sep=0.75pt]    {$...$};
\draw (209.12,99.54) node [anchor=north west][inner sep=0.75pt]    {$\delta ^{( 2)}$};
\draw (71,110.65) node [anchor=north west][inner sep=0.75pt]    {$\Omega ^{q_{1} -q_{3} \otimes q_{2} -q_{3} \otimes ...\otimes 0}$};
\draw (8,180.65) node [anchor=north west][inner sep=0.75pt]    {$\Omega ^{q_{1} -q_{3} +1\otimes q_{2} -q_{3} +1\otimes 1\otimes 0\otimes ...\otimes 0}$};
\draw (141.12,143.54) node [anchor=north west][inner sep=0.75pt]    {$\delta ^{( 3)}$};
\draw (391.5,217.96) node [anchor=north west][inner sep=0.75pt]  [rotate=-89.87]  {$...$};
\draw (0.09,251.68) node [anchor=north west][inner sep=0.75pt]    {$\Omega ^{q_{1} -q_{N} \otimes q_{2} -q_{N} \otimes ...\otimes q_{N-1} -q_{N} \otimes 0}$};
\draw (198.12,169.54) node [anchor=north west][inner sep=0.75pt]    {$\delta ^{( 3)}$};
\draw (243.92,189.37) node [anchor=north west][inner sep=0.75pt]    {$...$};
\draw (272.12,169.54) node [anchor=north west][inner sep=0.75pt]    {$\delta ^{( 3)}$};
\draw (318,179.65) node [anchor=north west][inner sep=0.75pt]    {$\Omega ^{q_{1} -q_{4} \otimes q_{2} -q_{4} \otimes q_{3} -q_{4} \otimes ...\otimes 0}$};
\draw (270.12,240.54) node [anchor=north west][inner sep=0.75pt]    {$\delta ^{( N-1)}$};
\draw (192.12,240.54) node [anchor=north west][inner sep=0.75pt]    {$\delta ^{( N-1)}$};
\draw (243.92,259.37) node [anchor=north west][inner sep=0.75pt]    {$...$};
\draw (299.5,217.96) node [anchor=north west][inner sep=0.75pt]  [rotate=-89.87]  {$...$};
\draw (262.5,217.96) node [anchor=north west][inner sep=0.75pt]  [rotate=-89.87]  {$...$};
\draw (218.5,217.96) node [anchor=north west][inner sep=0.75pt]  [rotate=-89.87]  {$...$};
\draw (128.5,217.96) node [anchor=north west][inner sep=0.75pt]  [rotate=-89.87]  {$...$};
\draw (322,251.68) node [anchor=north west][inner sep=0.75pt]    {$\Omega ^{q_{1} -q_{N-1} +1\otimes q_{2} -q_{N-1} +1\otimes ...\otimes 1\otimes 0}$};
\end{tikzpicture}.
\label{augcompl}
\end{equation}
In the end, adding the augmented complex \eqref{augcompl} to the original $N$-de Rham-like complex \eqref{Nderhames} we get the $(q_1-q_2,...,q_{N-1}-q_N)$-augmented $N$-de Rham-like complex. 
\\

For every $N$-de Rham complex we can define the de Rham-like cohomology groups as
\begin{Definition}[De Rham-like cohomology groups of a $k$-augmented $N$-de Rham-like complex]\label{coholike}
    Given a $(k_1,...,k_{N-1})$-augmented $N$-de Rham-like complex its de Rham-like cohomology groups are
    \begin{equation}
    H^{p_1,...,p_N}:=\frac{Z^{p_1,...,p_N}}{B^{p_1,...,p_N}},
    \end{equation}
    where 
    \begin{equation}
    \begin{aligned}
        Z^{p_1,...,p_N}&:=\{X \in \Omega^{p_1 \otimes ...\otimes p_N} | \delta^{(N)}X=0\},\\
        B^{p_1,...,p_N}&:=\{X=\delta^{(N)}Y \in \Omega^{p_1\otimes ...\otimes p_N} | Y \in \Omega^{p_1-1\otimes ...\otimes p_N-1}\}.
        \end{aligned}    
    \end{equation}
if $p_1,...,p_N >0$ and 
\begin{equation}
    H^{p_1,...,p_i,0,...,0}:=\frac{Z^{p_1,...,p_i,0,...,0}}{B^{p_1,...,p_i,0,...,0}},
\end{equation}
 where 
    \begin{equation}
    \begin{aligned}
        Z^{p_1,...,p_i,0,...,0}&:=\{X \in \Omega^{p_1 \otimes ...\otimes p_i\otimes 0 \otimes ... \otimes 0} | \delta^{(i+1)}X=0\},\\
        B^{p_1,...,p_i,0,...,0}&:=\{X=\delta^{(i)}Y \in \Omega^{p_1 \otimes ...\otimes p_i\otimes 0 \otimes ... \otimes 0} | Y \in \Omega^{p_1-1 \otimes ...\otimes p_i-1\otimes 0 \otimes ... \otimes 0}\}.
        \end{aligned}    
    \end{equation}
if only $p_1,...,p_i >0$     
\end{Definition}
These groups can be considered, for differential manifolds, as the generalization of de Rham cohomology groups; however, to make fully meaningful this conclusion we should prove that de Rham-like cohomology groups are topological invariants. In this perspective we first prove a Poincaré-like lemma for differential mixed symmetry tensors and then the main theorem about the isomorphism between de Rham-like cohomology groups and de Rham cohomology groups using abstract de Rham theorem  (see Appendix \ref{appfasci} for a review). 

\begin{Lemma}[Poincaré-like lemma]
    Let $U \subset \mathbb{R}^n$ a open polyinterval (product of $n$ open intervals even unlimited of $\mathbb{R}$). For every $(k_1,...,k_{N-1})$-augmented $N$-de Rham-like complex and $(p_1,...,p_N)\neq (0,...,0)$ then every closed differential mixed symmetry tensor $T$ is exact.
\end{Lemma}
\begin{proof}
    Unless translations, it is not restrictive to assume that $0 \in U$. Let us proceed by induction on $N$. For $N=1$ this is just the Poincaré lemma. Let us assume $N>1$ and let us consider the case $p_1,...,p_N >0$ first. 
    Every mixed symmetry tensor $T \in \Omega^{p_1 \otimes ... \otimes p_N}$ can be written as
    \begin{equation}
        T=\sum_{I_1,...,I_N}T_{I_1,...,I_N}dx_{I_1} \otimes ... \otimes dx_{I_N} \qquad |I_1|=p_1,...,|I_N|=p_N
        \label{eqmstform}
    \end{equation}
    We now define the subspace $\Omega^{p_1 \otimes ... \otimes p_N}_{m_1,...,m_N}$ with $m_1,...,m_N < n$ given by the mixed symmetry tensors \eqref{eqmstform} with $T_{I_1,...,I_N}=0$ if $I_i \not\subset \{1,...,m_i\} \quad \forall i \in [1,N]$,  hold simultaneously. We now proceed by induction on $(m_1,..,m_N)$. If $m_1=...=m_N=0$ there are only vanishing mixed symmetry tensors and there is nothing to show. If $m_i=0$ for some $i$ then the $I_i$ indexes do not appear and we are dealing with mixed symmetry tensors that belong to the $(N-1)$-multi-form space and by the induction hypothesis on $N$ we get the result. If $m_1,...,m_N >0$ then we write $T$ as
    \begin{equation}
        T=\sum_{I_1,...,I_N}T_{I_1,...,I_N} dx_{m_1} \wedge dx_{I_1}\otimes ... \otimes dx_{m_N} \wedge dx_{I_N}+\sum_{J_1,...,J_N}\Tilde{T}_{J_1,...,J_N}dx_{J_1} \otimes ... \otimes dx_{J_N}
        \label{defdiT}
    \end{equation}
    where $|I_i|=p_i-1$, $|J_i|=p_i$ and $I_i,J_i \subset \{1,...,m_i-1\} \quad \forall i \in [1,N]$. From the condition $\delta^{(N)}T=0$ we get that 
    \begin{equation}
        \frac{\partial T_{I_1,...,I_N}}{\partial x_{h_1}...\partial x_{h_N}}=0, \quad h_1>m_1,...,h_N>m_N.
    \end{equation}
    At this point we construct the sequence of $C^{\infty}$ functions (here we use that $U$ is a polyinterval)
    \begin{equation}
    \begin{aligned}
     C^{(N)}_{I_1...I_N}(x_1,...,x_n)&:=\int_0^{x_{m_N}} C^{(N-1)}_{I_1...I_N}(x_1,...x_{m_N-1},t,x_{m_N+1},...,x_n)dt, \\
        C^{(1)}_{I_1...I_N}(x_1,...,x_n)&:=\int_0^{x_{m_1}} T_{I_1...I_N}(x_1,...,x_{m_1-1},t,x_{m_1+1},...,x_n)dt
     \end{aligned}    
     \label{seqdic}
    \end{equation}
   such that 
   \begin{equation}
    \frac{\partial C^{(N)}_{I_1,...,I_N}}{\partial x_{m_1}...\partial x_{m_N}}=T_{I_1,...,I_N}; \qquad \frac{\partial C^{(N)}_{I_1,...,I_N}}{\partial x_{h_1}...\partial x_{h_N}}=0, \quad h_1>m_1,...,h_N>m_N.
   \end{equation}
   We then define
   \begin{equation}
       C=\sum_{{I_1,...,I_N}}C^{(N)}_{I_1,...,I_N}dx_{I_1} \otimes ... \otimes dx_{I_N}, \quad |I_i|=p_i-1, \ I_i \subset \{1,...,m_i-1\} \quad \forall \ i \in [1,N];
   \end{equation}
   therefore
   \begin{equation}
       T-\delta^{(N)}C \in \Omega^{p_1 \otimes ... \otimes p_N}_{m_1-1,...,m_N-1}
   \end{equation}
   since the differential of $c$ cancel out the first addendum in \eqref{defdiT}. Since $\delta^{(N)}(T-\delta^{(N)}C)=0$ by the induction hypothesis on $(m_1,...,m_N)$ there exist a $S \in \Omega^{p_1-1 \otimes ... \otimes p_N-1}$ such that $\delta^{(N)}S=T-\delta^{(N)}C$ and so $T=\delta^{(N)}(S-C)$. \\
   If only $p_1,..,p_j >0$ and the condition is $\delta^{(j+1)}T=0 \in \Omega^{p_1+1 \otimes ... \otimes p_j+1 \otimes 1 \otimes 0 \otimes ... \otimes 0}$ the same reasonings hold except that now it is enough to consider the sequence \ref{seqdic} up to its $j$-th element such that 
   \begin{equation}
    \frac{\partial C^{(j)}_{I_1,...,I_j}}{\partial x_{m_1}...\partial x_{m_j}}=T_{I_1,...,I_j}; \qquad \frac{\partial C^{(j)}_{I_1,...,I_j}}{\partial x_{h_1}...\partial x_{h_j}}=0, \quad h_1>m_1,...,h_j>m_j.
   \end{equation}
    As before we define
   \begin{equation}
   C=\sum_{{I_1,...,I_j}}C^{(j)}_{I_1,...,I_j}dx_{I_1} \otimes ... \otimes dx_{I_j}, \quad |I_i|=p_i-1, \ I_i \subset \{1,...,m_i-1\} \quad \forall \ i \in [1,j];
   \end{equation}
   therefore
   \begin{equation}
       T-\delta^{(j)}C \in \Omega^{p_1 \otimes ... \otimes p_j\otimes 0 \otimes 0... \otimes 0}_{m_1-1,...,m_j-1}
   \end{equation}
   Since $\delta^{(j+1)}(T-\delta^{(j)}C)=0$ by the induction hypothesis on $(m_1,...,m_j)$ there exist a $S \in \Omega^{p_1-1 \otimes ... \otimes p_j-1 \otimes 0 \otimes ... \otimes 0 }$ such that (in this case the condition $\delta^{(j+1)}(T-\delta^{(j)}C)=0$ implies a weaker condition since some $N$-multi-form degrees are zero) $\delta^{(j)}S=T-\delta^{(j)}C$ and so $T=\delta^{(j)}(S-C)$. 

\end{proof}
\begin{Theorem}[Isomorphism between de Rham-like cohomology groups and de Rham cohomology groups]\label{con4.3.1}
Given a differential manifold $(M,\mathcal{A})$ and a $(k_1,...,k_{N-1})$-augmented $N$-de Rham-like complex its de Rham-like cohomology groups are isomorphic to de Rham cohomology groups.
\end{Theorem}
\begin{proof}
    The theorem follows as an application of abstract de Rham theorem. In fact, thanks to Poincaré and Poincaré-like lemmas, both the de Rham and every $(k_1,...,k_{N-1})$-augmented $N$-de Rham-like complexes are exact. Therefore they are all exact sequences of shaves. Moreover the atlas $\mathcal{A}$ induces a structure sheaf $\mathcal{E}$ and the elements of de Rham complex and  of every $(k_1,...,k_{N-1})$-augmented $N$-de Rham-like complexes can be viewed as an $\mathcal{E}$-module, therefore they are all fine sheaves and hence acyclic. Since the underlying field is $\mathbb{R}$ we have different acyclic resolutions of the constant sheaf $\mathbb{R}_X$ and de Rham cohomology groups and de Rham-like cohomology groups of every $(k_1,...,k_{N-1})$-augmented $N$-de Rham-like complexes are coincide with the Cech cohomology groups of the constant sheaf $\mathbb{R}_X$ and hence isomorphic one to the others. 
\end{proof}
A first obvious observation that lead to this theorem is that for every fixed $(k_1,...,k_{N-1})$-augmented $N$-de Rham-like complex there are exactly $D$ de Rham-like cohomology groups by construction, where $D$ is the dimension of the differential manifold we are considering. Moreover, with reference to the example discussed above, the first $q_1-q_2$ de Rham-like cohomology groups are exactly the standard de Rham cohomology groups since the first line of diagram \eqref{augcompl} is the standard de Rham complex. Furthermore, Theorem \ref{con4.3.1} is motivated by the observation that, on a differential manifold with dimension $D$, we can construct countable infinity many $N$-multi-form space, that is, one for every $N \in \mathbb{N} \setminus \{0\}$. Any of them lead to de Rham-like cohomology groups and we end up with countable infinity many of these groups. On the one hand, it is very unlike that only for $N=1$, i.e. the de Rham complex, the de Rham-like cohomology groups furnishes interesting information about the topology of the manifold and, on the other hand, it is unlikely that all these countable infinity many groups give different information about the topology. Therefore, seems quite natural that there exist an isomorphism between de Rham-like cohomology groups and de Rham cohomology groups.  \\

\subsection{The extension of Theorem \ref{THM3.1} to mixed symmetry tensors}\label{exthm}
Once we have the generalization of de Rham complex for the $N$-multi-form space, we can extend the Theorem \ref{THM3.1} to the mixed symmetry tensor cases using the algebraic topology interpretation of the theorem. The idea is essentially the same. Thanks to implications come from the fact we are interested in asymptotic symmetries and the triviality of some de Rham-like cohomology groups we conclude that the de Rham-like differential is an isomorphism between the space of gauge fields and the space of their field strengths. Moreover, constructing an homotopy between specific $(k_1,...,k_{N-1})$-augmented $N$-de Rham-like complexes we can conclude the existence and uniqueness of a duality map in the case of well defined asymptotic charges computed in mixed symmetry tensor gauge theories which are duals.\\
Suppose we are interested in studying asymptotic symmetries in a gauge theory whose gauge field is a mixed symmetry tensor field $T \in \Omega^{q_1 \otimes ... \otimes q_N}$ and its field strength is $H:=\delta^{(N)}T \in \Omega^{q_1+1 \otimes ... \otimes q_N+1}$. Therefore, let us consider the $(q_1-q_2,...,q_{N-1}-q_N)$-augmented $N$-de Rham-like complex and, in a specific way, the last terms, i.e.
\begin{equation}
\tikzset{every picture/.style={line width=0.75pt}} 
\begin{tikzpicture}[x=0.65pt,y=0.75pt,yscale=-1,xscale=1]

\draw    (159.82,184.28) -- (197.35,184.28) ;
\draw [shift={(199.35,184.28)}, rotate = 180] [color={rgb, 255:red, 0; green, 0; blue, 0 }  ][line width=0.75]    (10.93,-3.29) .. controls (6.95,-1.4) and (3.31,-0.3) .. (0,0) .. controls (3.31,0.3) and (6.95,1.4) .. (10.93,3.29)   ;
\draw    (228.82,184.28) -- (266.35,184.28) ;
\draw [shift={(268.35,184.28)}, rotate = 180] [color={rgb, 255:red, 0; green, 0; blue, 0 }  ][line width=0.75]    (10.93,-3.29) .. controls (6.95,-1.4) and (3.31,-0.3) .. (0,0) .. controls (3.31,0.3) and (6.95,1.4) .. (10.93,3.29)   ;
\draw    (350.82,184.28) -- (388.35,184.28) ;
\draw [shift={(390.35,184.28)}, rotate = 180] [color={rgb, 255:red, 0; green, 0; blue, 0 }  ][line width=0.75]    (10.93,-3.29) .. controls (6.95,-1.4) and (3.31,-0.3) .. (0,0) .. controls (3.31,0.3) and (6.95,1.4) .. (10.93,3.29)   ;
\draw    (420.82,184.28) -- (458.35,184.28) ;
\draw [shift={(460.35,184.28)}, rotate = 180] [color={rgb, 255:red, 0; green, 0; blue, 0 }  ][line width=0.75]    (10.93,-3.29) .. controls (6.95,-1.4) and (3.31,-0.3) .. (0,0) .. controls (3.31,0.3) and (6.95,1.4) .. (10.93,3.29)   ;

\draw (9.97,172.65) node [anchor=north west][inner sep=0.75pt]    {$\Omega ^{q_{1} -q_{N} \otimes ...\otimes q_{N-1} -q_{N} \otimes 0}$};
\draw (177.12,160.54) node [anchor=north west][inner sep=0.75pt]    {$\delta ^{( N)}$};
\draw (208.92,180.37) node [anchor=north west][inner sep=0.75pt]    {$...$};
\draw (233.12,161.54) node [anchor=north west][inner sep=0.75pt]    {$\delta ^{( N)}$};
\draw (277,172.65) node [anchor=north west][inner sep=0.75pt]    {$\Omega ^{q_{1} \otimes ...\otimes q_{N}}$};
\draw (355.12,160.54) node [anchor=north west][inner sep=0.75pt]    {$\delta ^{( N)}$};
\draw (397.92,180.37) node [anchor=north west][inner sep=0.75pt]    {$...$};
\draw (425.12,160.54) node [anchor=north west][inner sep=0.75pt]    {$\delta ^{( N)}$};
\draw (466.59,172.68) node [anchor=north west][inner sep=0.75pt]    {$\Omega ^{D\otimes D-( q_{1} -q_{2}) \otimes ...\otimes D-( q_{1} -q_{N})}$};
\end{tikzpicture}.
\end{equation}
Requiring we are interested in asymptotic symmetries of a gauge theory whose gauge field is $T$ and requiring the vanishing of $H^{q_1,...,q_N}$ and $H^{q_1+1,...,q_N+1}$ means, following the very same reasoning\footnote{That is, we have $H=0 \Leftrightarrow T=\delta^{(N)}A$ for some $A \in \Omega^{q_1-1\otimes ... \otimes q_N-1}$; hence we need to throw away all those elements $T \in \Omega^{q_1\otimes ... \otimes q_N}$ such that $T=\delta^{(N)}A$ for some $A \in \Omega^{q_1-1\otimes ... \otimes q_N-1}$. Moreover, only the zero form can have vanishing field strength but, again for exactness $B=0 \Leftrightarrow A=\delta^{(N)}C$ for some $C \in \Omega^{q_1-2\otimes ... \otimes q_N-2}$.} of the $p$-form case above, that the module $\Omega^{q_1+1\otimes ... \otimes q_N+1}$ can be replaced by $\Omega^{q_1+1\otimes ... \otimes q_N+1}_{\text{AS}}:=\{H \in \Omega^{q_1+1\otimes ... \otimes q_N+1} | H=\delta^{(N)}B, H \neq 0\} \cup \{H=0\}$. Moreover both $\Omega^{q_1-1\otimes ... \otimes q_N-1}$ and $\Omega^{q_1+2\otimes ... \otimes q_N+2}$ can be replaced by $0$. Therefore we have the following short exact sequence 
\begin{equation}
\tikzset{every picture/.style={line width=0.75pt}} 
\begin{tikzpicture}[x=0.75pt,y=0.75pt,yscale=-1,xscale=1]

\draw    (113.33,175.17) -- (163.33,175.17) ;
\draw [shift={(165.33,175.17)}, rotate = 180] [color={rgb, 255:red, 0; green, 0; blue, 0 }  ][line width=0.75]    (10.93,-3.29) .. controls (6.95,-1.4) and (3.31,-0.3) .. (0,0) .. controls (3.31,0.3) and (6.95,1.4) .. (10.93,3.29)   ;
\draw    (240.33,175.17) -- (290.33,175.17) ;
\draw [shift={(292.33,175.17)}, rotate = 180] [color={rgb, 255:red, 0; green, 0; blue, 0 }  ][line width=0.75]    (10.93,-3.29) .. controls (6.95,-1.4) and (3.31,-0.3) .. (0,0) .. controls (3.31,0.3) and (6.95,1.4) .. (10.93,3.29)   ;
\draw    (397.33,175.17) -- (447.33,175.17) ;
\draw [shift={(449.33,175.17)}, rotate = 180] [color={rgb, 255:red, 0; green, 0; blue, 0 }  ][line width=0.75]    (10.93,-3.29) .. controls (6.95,-1.4) and (3.31,-0.3) .. (0,0) .. controls (3.31,0.3) and (6.95,1.4) .. (10.93,3.29)   ;

\draw (120.12,154.54) node [anchor=north west][inner sep=0.75pt]    {$\delta ^{( N)}$};
\draw (100,166.5) node [anchor=north west][inner sep=0.75pt]    {$0$};
\draw (172,161.4) node [anchor=north west][inner sep=0.75pt]    {$\Omega _{\text{AS}}^{q_{1} \otimes ...\otimes q_{N}}$};
\draw (299,161.4) node [anchor=north west][inner sep=0.75pt]    {$\Omega _{\text{AS}}^{q_{1} +1\otimes ...\otimes q_{N} +1}$};
\draw (452,166.5) node [anchor=north west][inner sep=0.75pt]    {$0$};
\draw (250.12,154.54) node [anchor=north west][inner sep=0.75pt]    {$\delta ^{( N)}$};
\draw (403.12,154.54) node [anchor=north west][inner sep=0.75pt]    {$\delta ^{( N)}$};
\end{tikzpicture}
\label{shortexactAS}
\end{equation}
which teaches us that the differential $\delta^{(N)}$ between $\Omega _{\text{AS}}^{q_{1} \otimes ...\otimes q_{N}}$ and $\Omega _{\text{AS}}^{q_{1}+1 \otimes ...\otimes q_{N}+1}$ is an isomorphism. \\
In order to extend theorem \ref{THM3.1}, let us proceed with the specific case of the graviton in $D=5$ and its dual descriptions: the Curtright three indexes field, or hooke field, and the Riemann-like field. Therefore, we need to consider two copies of the 0-augmented 2-de Rham-like complex (one for the graviton field and one for the Riemann-like field) and a copy of the 1-augmented 2-de Rham-like complex (for the Curtright hooke field). Let us suppose we have computed the asymptotic charges $Q_0,Q_1,Q_2$ of all descriptions; the diagram we consider to show the existence and uniqueness of maps between dual descriptions is the following
\begin{equation}
    \tikzset{every picture/.style={line width=0.75pt}} 
\begin{tikzpicture}[x=0.75pt,y=0.75pt,yscale=-1,xscale=1]

\draw    (279.28,346.5) -- (278.34,296.51) ;
\draw [shift={(278.3,294.51)}, rotate = 88.92] [color={rgb, 255:red, 0; green, 0; blue, 0 }  ][line width=0.75]    (10.93,-3.29) .. controls (6.95,-1.4) and (3.31,-0.3) .. (0,0) .. controls (3.31,0.3) and (6.95,1.4) .. (10.93,3.29)   ;
\draw    (276.98,243.52) -- (276.04,193.53) ;
\draw [shift={(276,191.53)}, rotate = 88.92] [color={rgb, 255:red, 0; green, 0; blue, 0 }  ][line width=0.75]    (10.93,-3.29) .. controls (6.95,-1.4) and (3.31,-0.3) .. (0,0) .. controls (3.31,0.3) and (6.95,1.4) .. (10.93,3.29)   ;
\draw    (277.24,140.54) -- (276.3,90.55) ;
\draw [shift={(276.26,88.55)}, rotate = 88.92] [color={rgb, 255:red, 0; green, 0; blue, 0 }  ][line width=0.75]    (10.93,-3.29) .. controls (6.95,-1.4) and (3.31,-0.3) .. (0,0) .. controls (3.31,0.3) and (6.95,1.4) .. (10.93,3.29)   ;
\draw    (189.28,346.5) -- (188.34,296.51) ;
\draw [shift={(188.3,294.51)}, rotate = 88.92] [color={rgb, 255:red, 0; green, 0; blue, 0 }  ][line width=0.75]    (10.93,-3.29) .. controls (6.95,-1.4) and (3.31,-0.3) .. (0,0) .. controls (3.31,0.3) and (6.95,1.4) .. (10.93,3.29)   ;
\draw    (186.98,243.52) -- (186.04,193.53) ;
\draw [shift={(186,191.53)}, rotate = 88.92] [color={rgb, 255:red, 0; green, 0; blue, 0 }  ][line width=0.75]    (10.93,-3.29) .. controls (6.95,-1.4) and (3.31,-0.3) .. (0,0) .. controls (3.31,0.3) and (6.95,1.4) .. (10.93,3.29)   ;
\draw    (187.24,140.54) -- (186.3,90.55) ;
\draw [shift={(186.26,88.55)}, rotate = 88.92] [color={rgb, 255:red, 0; green, 0; blue, 0 }  ][line width=0.75]    (10.93,-3.29) .. controls (6.95,-1.4) and (3.31,-0.3) .. (0,0) .. controls (3.31,0.3) and (6.95,1.4) .. (10.93,3.29)   ;
\draw    (369.28,346.5) -- (368.34,296.51) ;
\draw [shift={(368.3,294.51)}, rotate = 88.92] [color={rgb, 255:red, 0; green, 0; blue, 0 }  ][line width=0.75]    (10.93,-3.29) .. controls (6.95,-1.4) and (3.31,-0.3) .. (0,0) .. controls (3.31,0.3) and (6.95,1.4) .. (10.93,3.29)   ;
\draw    (366.98,243.52) -- (366.04,193.53) ;
\draw [shift={(366,191.53)}, rotate = 88.92] [color={rgb, 255:red, 0; green, 0; blue, 0 }  ][line width=0.75]    (10.93,-3.29) .. controls (6.95,-1.4) and (3.31,-0.3) .. (0,0) .. controls (3.31,0.3) and (6.95,1.4) .. (10.93,3.29)   ;
\draw    (367.24,140.54) -- (366.3,90.55) ;
\draw [shift={(366.26,88.55)}, rotate = 88.92] [color={rgb, 255:red, 0; green, 0; blue, 0 }  ][line width=0.75]    (10.93,-3.29) .. controls (6.95,-1.4) and (3.31,-0.3) .. (0,0) .. controls (3.31,0.3) and (6.95,1.4) .. (10.93,3.29)   ;
\draw    (296.77,165.66) -- (339.77,165.66) ;
\draw [shift={(341.77,165.66)}, rotate = 180] [color={rgb, 255:red, 0; green, 0; blue, 0 }  ][line width=0.75]    (10.93,-3.29) .. controls (6.95,-1.4) and (3.31,-0.3) .. (0,0) .. controls (3.31,0.3) and (6.95,1.4) .. (10.93,3.29)   ;
\draw    (251.77,165.66) -- (207.77,165.66) ;
\draw [shift={(205.77,165.66)}, rotate = 360] [color={rgb, 255:red, 0; green, 0; blue, 0 }  ][line width=0.75]    (10.93,-3.29) .. controls (6.95,-1.4) and (3.31,-0.3) .. (0,0) .. controls (3.31,0.3) and (6.95,1.4) .. (10.93,3.29)   ;
\draw    (250.77,270.66) -- (206.77,270.66) ;
\draw [shift={(204.77,270.66)}, rotate = 360] [color={rgb, 255:red, 0; green, 0; blue, 0 }  ][line width=0.75]    (10.93,-3.29) .. controls (6.95,-1.4) and (3.31,-0.3) .. (0,0) .. controls (3.31,0.3) and (6.95,1.4) .. (10.93,3.29)   ;
\draw    (295.77,270.66) -- (338.77,270.66) ;
\draw [shift={(340.77,270.66)}, rotate = 180] [color={rgb, 255:red, 0; green, 0; blue, 0 }  ][line width=0.75]    (10.93,-3.29) .. controls (6.95,-1.4) and (3.31,-0.3) .. (0,0) .. controls (3.31,0.3) and (6.95,1.4) .. (10.93,3.29)   ;
\draw    (392,271) .. controls (433.34,292.88) and (432.78,377.38) .. (394.18,409.54) ;
\draw [shift={(393,410.5)}, rotate = 321.7] [color={rgb, 255:red, 0; green, 0; blue, 0 }  ][line width=0.75]    (10.93,-3.29) .. controls (6.95,-1.4) and (3.31,-0.3) .. (0,0) .. controls (3.31,0.3) and (6.95,1.4) .. (10.93,3.29)   ;
\draw    (159,271) .. controls (120.05,291.31) and (119.66,379.72) .. (158.8,409.61) ;
\draw [shift={(160,410.5)}, rotate = 215.69] [color={rgb, 255:red, 0; green, 0; blue, 0 }  ][line width=0.75]    (10.93,-3.29) .. controls (6.95,-1.4) and (3.31,-0.3) .. (0,0) .. controls (3.31,0.3) and (6.95,1.4) .. (10.93,3.29)   ;
\draw    (293,280) .. controls (334.34,301.88) and (335.02,374.4) .. (296.45,406.32) ;
\draw [shift={(295.27,407.27)}, rotate = 321.7] [color={rgb, 255:red, 0; green, 0; blue, 0 }  ][line width=0.75]    (10.93,-3.29) .. controls (6.95,-1.4) and (3.31,-0.3) .. (0,0) .. controls (3.31,0.3) and (6.95,1.4) .. (10.93,3.29)   ;
\draw    (302.77,421.66) -- (345.77,421.66) ;
\draw [shift={(347.77,421.66)}, rotate = 180] [color={rgb, 255:red, 0; green, 0; blue, 0 }  ][line width=0.75]    (10.93,-3.29) .. controls (6.95,-1.4) and (3.31,-0.3) .. (0,0) .. controls (3.31,0.3) and (6.95,1.4) .. (10.93,3.29)   ;
\draw    (253.77,421.66) -- (209.77,421.66) ;
\draw [shift={(207.77,421.66)}, rotate = 360] [color={rgb, 255:red, 0; green, 0; blue, 0 }  ][line width=0.75]    (10.93,-3.29) .. controls (6.95,-1.4) and (3.31,-0.3) .. (0,0) .. controls (3.31,0.3) and (6.95,1.4) .. (10.93,3.29)   ;

\draw (258.53,340.1) node [anchor=north west][inner sep=0.75pt]  [rotate=-268.92]  {$\delta ^{( 2)}$};
\draw (273.23,347.31) node [anchor=north west][inner sep=0.75pt]  [rotate=-359.84]  {$0$};
\draw (256.22,256.68) node [anchor=north west][inner sep=0.75pt]  [rotate=-0.36]  {$\Omega _{\text{AS}}^{1\otimes 1}$};
\draw (255.91,153.71) node [anchor=north west][inner sep=0.75pt]  [rotate=-0.32]  {$\Omega _{\text{AS}}^{2\otimes 2}$};
\draw (270.99,69.29) node [anchor=north west][inner sep=0.75pt]  [rotate=-0.45]  {$0$};
\draw (256.17,234.13) node [anchor=north west][inner sep=0.75pt]  [rotate=-268.92]  {$\delta ^{( 2)}$};
\draw (256.51,135.14) node [anchor=north west][inner sep=0.75pt]  [rotate=-268.92]  {$\delta ^{( 2)}$};
\draw (168.53,340.1) node [anchor=north west][inner sep=0.75pt]  [rotate=-268.92]  {$\delta ^{( 2)}$};
\draw (183.23,347.31) node [anchor=north west][inner sep=0.75pt]  [rotate=-359.84]  {$0$};
\draw (166.22,256.68) node [anchor=north west][inner sep=0.75pt]  [rotate=-0.36]  {$\Omega _{\text{AS}}^{2\otimes 1}$};
\draw (165.91,153.71) node [anchor=north west][inner sep=0.75pt]  [rotate=-0.32]  {$\Omega _{\text{AS}}^{3\otimes 2}$};
\draw (180.99,69.29) node [anchor=north west][inner sep=0.75pt]  [rotate=-0.45]  {$0$};
\draw (166.17,234.13) node [anchor=north west][inner sep=0.75pt]  [rotate=-268.92]  {$\delta ^{( 2)}$};
\draw (166.51,135.14) node [anchor=north west][inner sep=0.75pt]  [rotate=-268.92]  {$\delta ^{( 2)}$};
\draw (348.53,340.1) node [anchor=north west][inner sep=0.75pt]  [rotate=-268.92]  {$\delta ^{( 2)}$};
\draw (363.23,347.31) node [anchor=north west][inner sep=0.75pt]  [rotate=-359.84]  {$0$};
\draw (346.22,256.68) node [anchor=north west][inner sep=0.75pt]  [rotate=-0.36]  {$\Omega _{\text{AS}}^{2\otimes 2}$};
\draw (345.91,153.71) node [anchor=north west][inner sep=0.75pt]  [rotate=-0.32]  {$\Omega _{\text{AS}}^{3\otimes 3}$};
\draw (360.99,69.29) node [anchor=north west][inner sep=0.75pt]  [rotate=-0.45]  {$0$};
\draw (346.17,234.13) node [anchor=north west][inner sep=0.75pt]  [rotate=-268.92]  {$\delta ^{( 2)}$};
\draw (346.51,135.14) node [anchor=north west][inner sep=0.75pt]  [rotate=-268.92]  {$\delta ^{( 2)}$};
\draw (307,143.4) node [anchor=north west][inner sep=0.75pt]    {$\star $};
\draw (220,144.4) node [anchor=north west][inner sep=0.75pt]    {$\star _{L}$};
\draw (209,249.4) node [anchor=north west][inner sep=0.75pt]    {$\star _{L} |_{3}$};
\draw (297,249.4) node [anchor=north west][inner sep=0.75pt]    {$\star |_{3}$};
\draw (353,406.4) node [anchor=north west][inner sep=0.75pt]    {$\mathbb{C}^{n_{2,2}}$};
\draw (265,406.4) node [anchor=north west][inner sep=0.75pt]    {$\mathbb{C}^{n_{1,1}}$};
\draw (170,406.4) node [anchor=north west][inner sep=0.75pt]    {$\mathbb{C}^{n_{2,1}}$};
\draw (328,344.4) node [anchor=north west][inner sep=0.75pt]    {$\pi _{0}$};
\draw (110,344.4) node [anchor=north west][inner sep=0.75pt]    {$\pi _{1}$};
\draw (425,344.4) node [anchor=north west][inner sep=0.75pt]    {$\pi _{2}$};
\draw (315,398.4) node [anchor=north west][inner sep=0.75pt]    {$f_{2}$};
\draw (225,398.4) node [anchor=north west][inner sep=0.75pt]    {$f_{1}$};
\end{tikzpicture}
\end{equation}
where $n_{2,1}=dim(\Omega_{\text{AS}}^{2 \otimes 1})$, $n_{1,1}=dim(\Omega_{\text{AS}}^{1 \otimes 1})$ and $n_{2,2}=dim(\Omega_{\text{AS}}^{2 \otimes 2})$. The first part of the proof is essentially given by the discussion around the short exact complex \eqref{shortexactAS} since then both $\star_L|_3$ and $\star|_3$ are isomorphisms. The second part of the proof goes through constructing the maps $\pi_0,\pi_1,\pi_2$ and proving that they are isomorphisms. Both these points are essentially identical to the case of theorem \ref{THM3.1}. Indeed the vector in the complex spaces can be constructed having the first component given by the charge while the others given by the independent entries of the matrix whose represent the mixed symmetry tensor in a coordinate chart. Hence to show the maps $\pi_0,\pi_1,\pi_2$ are isomorphisms we can follow, quite exactly, the steps given in the proof of theorem \ref{THM3.1}.

In the general case we can prove the following 
\begin{Theorem}[Existence and uniqueness of a set of duality maps for well defined charges]\label{THM4.3.1}
    Let $(M,\mathcal{A})$ be a differential manifold of dimension $D$ and let $\Omega^{n_1 \otimes ... \otimes n_N}(M)$ be the $N$-multi-form space on it. Let $T \in \Omega^{p_1 \otimes ... \otimes p_N}$ with $p_k\leq \big[\frac{D-2}{2}\big] \ \forall k \in [1,N]$ be a mixed symmetry tensor carrying the irreducible representation associated to the Young tableau $\lambda_0=\{p_1,...,p_N\}$ having $r$ dual descriptions and let $\lambda_1=\{q_1=D-2-p_1,p_2,...,p_N\},...,\lambda_r=\{q_1=D-2-p_1,...,q_r=D-2-p_r,...,q_N=D-2-p_N\}$ be the Young tableaux associated to the irreducible representation carried by the dual mixed symmetry tensors. Let $Q_0,...,Q_r$ be the well defined charges of the original and of the dual descriptions and let be $n_{n_1,...,n_N}=dim(\Omega_{\textnormal{AS}}^{n_1\otimes...\otimes n_N})$. Then there exist a set of duality maps $\{f_1,...,f_r\}$ such that the following diagram commute
    \begin{equation}
    \tikzset{every picture/.style={line width=0.75pt}} 
\begin{tikzpicture}[x=0.62pt,y=0.75pt,yscale=-1,xscale=1]

\draw    (427.14,353.04) -- (426.24,309.83) ;
\draw [shift={(426.2,307.83)}, rotate = 88.81] [color={rgb, 255:red, 0; green, 0; blue, 0 }  ][line width=0.75]    (10.93,-3.29) .. controls (6.95,-1.4) and (3.31,-0.3) .. (0,0) .. controls (3.31,0.3) and (6.95,1.4) .. (10.93,3.29)   ;
\draw    (424.94,263.49) -- (424.04,220.28) ;
\draw [shift={(424,218.28)}, rotate = 88.81] [color={rgb, 255:red, 0; green, 0; blue, 0 }  ][line width=0.75]    (10.93,-3.29) .. controls (6.95,-1.4) and (3.31,-0.3) .. (0,0) .. controls (3.31,0.3) and (6.95,1.4) .. (10.93,3.29)   ;
\draw    (425.19,173.93) -- (424.29,130.72) ;
\draw [shift={(424.25,128.72)}, rotate = 88.81] [color={rgb, 255:red, 0; green, 0; blue, 0 }  ][line width=0.75]    (10.93,-3.29) .. controls (6.95,-1.4) and (3.31,-0.3) .. (0,0) .. controls (3.31,0.3) and (6.95,1.4) .. (10.93,3.29)   ;
\draw    (121.75,353.91) -- (120.86,310.7) ;
\draw [shift={(120.81,308.7)}, rotate = 88.81] [color={rgb, 255:red, 0; green, 0; blue, 0 }  ][line width=0.75]    (10.93,-3.29) .. controls (6.95,-1.4) and (3.31,-0.3) .. (0,0) .. controls (3.31,0.3) and (6.95,1.4) .. (10.93,3.29)   ;
\draw    (119.55,264.36) -- (118.66,221.15) ;
\draw [shift={(118.61,219.15)}, rotate = 88.81] [color={rgb, 255:red, 0; green, 0; blue, 0 }  ][line width=0.75]    (10.93,-3.29) .. controls (6.95,-1.4) and (3.31,-0.3) .. (0,0) .. controls (3.31,0.3) and (6.95,1.4) .. (10.93,3.29)   ;
\draw    (119.8,174.8) -- (118.9,131.59) ;
\draw [shift={(118.86,129.59)}, rotate = 88.81] [color={rgb, 255:red, 0; green, 0; blue, 0 }  ][line width=0.75]    (10.93,-3.29) .. controls (6.95,-1.4) and (3.31,-0.3) .. (0,0) .. controls (3.31,0.3) and (6.95,1.4) .. (10.93,3.29)   ;
\draw    (589.89,353.04) -- (588.99,309.83) ;
\draw [shift={(588.95,307.83)}, rotate = 88.81] [color={rgb, 255:red, 0; green, 0; blue, 0 }  ][line width=0.75]    (10.93,-3.29) .. controls (6.95,-1.4) and (3.31,-0.3) .. (0,0) .. controls (3.31,0.3) and (6.95,1.4) .. (10.93,3.29)   ;
\draw    (587.69,263.49) -- (586.79,220.28) ;
\draw [shift={(586.75,218.28)}, rotate = 88.81] [color={rgb, 255:red, 0; green, 0; blue, 0 }  ][line width=0.75]    (10.93,-3.29) .. controls (6.95,-1.4) and (3.31,-0.3) .. (0,0) .. controls (3.31,0.3) and (6.95,1.4) .. (10.93,3.29)   ;
\draw    (587.93,173.93) -- (587.04,130.72) ;
\draw [shift={(586.99,128.72)}, rotate = 88.81] [color={rgb, 255:red, 0; green, 0; blue, 0 }  ][line width=0.75]    (10.93,-3.29) .. controls (6.95,-1.4) and (3.31,-0.3) .. (0,0) .. controls (3.31,0.3) and (6.95,1.4) .. (10.93,3.29)   ;
\draw    (482.18,198.39) -- (523.26,198.39) ;
\draw [shift={(525.26,198.39)}, rotate = 180] [color={rgb, 255:red, 0; green, 0; blue, 0 }  ][line width=0.75]    (10.93,-3.29) .. controls (6.95,-1.4) and (3.31,-0.3) .. (0,0) .. controls (3.31,0.3) and (6.95,1.4) .. (10.93,3.29)   ;
\draw    (364.48,197.81) -- (174.96,197.52) ;
\draw [shift={(172.96,197.52)}, rotate = 0.09] [color={rgb, 255:red, 0; green, 0; blue, 0 }  ][line width=0.75]    (10.93,-3.29) .. controls (6.95,-1.4) and (3.31,-0.3) .. (0,0) .. controls (3.31,0.3) and (6.95,1.4) .. (10.93,3.29)   ;
\draw    (382.67,287.38) -- (178.84,285.66) ;
\draw [shift={(176.84,285.64)}, rotate = 0.48] [color={rgb, 255:red, 0; green, 0; blue, 0 }  ][line width=0.75]    (10.93,-3.29) .. controls (6.95,-1.4) and (3.31,-0.3) .. (0,0) .. controls (3.31,0.3) and (6.95,1.4) .. (10.93,3.29)   ;
\draw    (474.54,286.52) -- (533.81,286.52) ;
\draw [shift={(535.81,286.52)}, rotate = 180] [color={rgb, 255:red, 0; green, 0; blue, 0 }  ][line width=0.75]    (10.93,-3.29) .. controls (6.95,-1.4) and (3.31,-0.3) .. (0,0) .. controls (3.31,0.3) and (6.95,1.4) .. (10.93,3.29)   ;
\draw    (642,293.5) .. controls (681.58,312.52) and (684.02,380.81) .. (647.13,408.67) ;
\draw [shift={(646,409.5)}, rotate = 324.34] [color={rgb, 255:red, 0; green, 0; blue, 0 }  ][line width=0.75]    (10.93,-3.29) .. controls (6.95,-1.4) and (3.31,-0.3) .. (0,0) .. controls (3.31,0.3) and (6.95,1.4) .. (10.93,3.29)   ;
\draw    (93.55,298.69) .. controls (56.26,316.35) and (55.12,383.01) .. (92.57,408.79) ;
\draw [shift={(93.72,409.56)}, rotate = 213.13] [color={rgb, 255:red, 0; green, 0; blue, 0 }  ][line width=0.75]    (10.93,-3.29) .. controls (6.95,-1.4) and (3.31,-0.3) .. (0,0) .. controls (3.31,0.3) and (6.95,1.4) .. (10.93,3.29)   ;
\draw    (446.02,296.08) .. controls (485.6,315.1) and (486.25,378.17) .. (449.32,405.93) ;
\draw [shift={(448.19,406.76)}, rotate = 324.34] [color={rgb, 255:red, 0; green, 0; blue, 0 }  ][line width=0.75]    (10.93,-3.29) .. controls (6.95,-1.4) and (3.31,-0.3) .. (0,0) .. controls (3.31,0.3) and (6.95,1.4) .. (10.93,3.29)   ;
\draw    (460.58,416.62) -- (560.05,416.62) ;
\draw [shift={(562.05,416.62)}, rotate = 180] [color={rgb, 255:red, 0; green, 0; blue, 0 }  ][line width=0.75]    (10.93,-3.29) .. controls (6.95,-1.4) and (3.31,-0.3) .. (0,0) .. controls (3.31,0.3) and (6.95,1.4) .. (10.93,3.29)   ;
\draw    (395.11,417.82) -- (159.64,417.53) ;
\draw [shift={(157.64,417.53)}, rotate = 0.07] [color={rgb, 255:red, 0; green, 0; blue, 0 }  ][line width=0.75]    (10.93,-3.29) .. controls (6.95,-1.4) and (3.31,-0.3) .. (0,0) .. controls (3.31,0.3) and (6.95,1.4) .. (10.93,3.29)   ;
\draw    (394.15,176.07) -- (327.59,140.93) ;
\draw [shift={(325.82,140)}, rotate = 27.83] [color={rgb, 255:red, 0; green, 0; blue, 0 }  ][line width=0.75]    (10.93,-3.29) .. controls (6.95,-1.4) and (3.31,-0.3) .. (0,0) .. controls (3.31,0.3) and (6.95,1.4) .. (10.93,3.29)   ;
\draw    (395.11,273.47) -- (327.59,237.93) ;
\draw [shift={(325.82,237)}, rotate = 27.76] [color={rgb, 255:red, 0; green, 0; blue, 0 }  ][line width=0.75]    (10.93,-3.29) .. controls (6.95,-1.4) and (3.31,-0.3) .. (0,0) .. controls (3.31,0.3) and (6.95,1.4) .. (10.93,3.29)   ;
\draw    (404.68,409.13) -- (317.57,360.97) ;
\draw [shift={(315.82,360)}, rotate = 28.94] [color={rgb, 255:red, 0; green, 0; blue, 0 }  ][line width=0.75]    (10.93,-3.29) .. controls (6.95,-1.4) and (3.31,-0.3) .. (0,0) .. controls (3.31,0.3) and (6.95,1.4) .. (10.93,3.29)   ;
\draw    (310.82,293) -- (309.85,241.74) ;
\draw [shift={(309.81,239.74)}, rotate = 88.91] [color={rgb, 255:red, 0; green, 0; blue, 0 }  ][line width=0.75]    (10.93,-3.29) .. controls (6.95,-1.4) and (3.31,-0.3) .. (0,0) .. controls (3.31,0.3) and (6.95,1.4) .. (10.93,3.29)   ;
\draw    (308.82,202) -- (307.86,145) ;
\draw [shift={(307.82,143)}, rotate = 89.03] [color={rgb, 255:red, 0; green, 0; blue, 0 }  ][line width=0.75]    (10.93,-3.29) .. controls (6.95,-1.4) and (3.31,-0.3) .. (0,0) .. controls (3.31,0.3) and (6.95,1.4) .. (10.93,3.29)   ;
\draw    (308.8,105.85) -- (307.9,62.64) ;
\draw [shift={(307.86,60.64)}, rotate = 88.81] [color={rgb, 255:red, 0; green, 0; blue, 0 }  ][line width=0.75]    (10.93,-3.29) .. controls (6.95,-1.4) and (3.31,-0.3) .. (0,0) .. controls (3.31,0.3) and (6.95,1.4) .. (10.93,3.29)   ;
\draw    (278.55,238.73) .. controls (277.3,292.45) and (242.38,315.33) .. (279.15,340.24) ;
\draw [shift={(280.29,340.99)}, rotate = 213.13] [color={rgb, 255:red, 0; green, 0; blue, 0 }  ][line width=0.75]    (10.93,-3.29) .. controls (6.95,-1.4) and (3.31,-0.3) .. (0,0) .. controls (3.31,0.3) and (6.95,1.4) .. (10.93,3.29)   ;

\draw (406.94,349.45) node [anchor=north west][inner sep=0.75pt]  [rotate=-268.92]  {$\delta ^{( N)}$};
\draw (421.09,352.75) node [anchor=north west][inner sep=0.75pt]  [rotate=-359.84]  {$0$};
\draw (399.8,273.26) node [anchor=north west][inner sep=0.75pt]  [rotate=-0.36]  {$\Omega _{\textnormal{AS}}^{p_{1} \otimes ...\otimes p_{N}}$};
\draw (418.95,110.98) node [anchor=north west][inner sep=0.75pt]  [rotate=-0.45]  {$0$};
\draw (404.69,257.3) node [anchor=north west][inner sep=0.75pt]  [rotate=-268.92]  {$\delta ^{( N)}$};
\draw (405.01,171.22) node [anchor=north west][inner sep=0.75pt]  [rotate=-268.92]  {$\delta ^{( N)}$};
\draw (101.55,350.32) node [anchor=north west][inner sep=0.75pt]  [rotate=-268.92]  {$\delta ^{( N)}$};
\draw (115.7,353.62) node [anchor=north west][inner sep=0.75pt]  [rotate=-359.84]  {$0$};
\draw (96.33,274.13) node [anchor=north west][inner sep=0.75pt]  [rotate=-0.36]  {$\Omega _{\textnormal{AS}}^{q_{1} \otimes ...\otimes p_{N}}$};
\draw (113.56,111.85) node [anchor=north west][inner sep=0.75pt]  [rotate=-0.45]  {$0$};
\draw (99.3,258.16) node [anchor=north west][inner sep=0.75pt]  [rotate=-268.92]  {$\delta ^{( N)}$};
\draw (99.62,172.09) node [anchor=north west][inner sep=0.75pt]  [rotate=-268.92]  {$\delta ^{( N)}$};
\draw (569.69,349.19) node [anchor=north west][inner sep=0.75pt]  [rotate=-268.92]  {$\delta ^{( N)}$};
\draw (583.84,352.75) node [anchor=north west][inner sep=0.75pt]  [rotate=-359.84]  {$0$};
\draw (581.7,110.98) node [anchor=north west][inner sep=0.75pt]  [rotate=-0.45]  {$0$};
\draw (567.43,257.3) node [anchor=north west][inner sep=0.75pt]  [rotate=-268.92]  {$\delta ^{( N)}$};
\draw (567.75,171.22) node [anchor=north west][inner sep=0.75pt]  [rotate=-268.92]  {$\delta ^{( N)}$};
\draw (492.52,178.04) node [anchor=north west][inner sep=0.75pt]    {$\star ^{( r)}$};
\draw (187.02,177.11) node [anchor=north west][inner sep=0.75pt]    {$\star ^{( 1)}$};
\draw (188.91,263.07) node [anchor=north west][inner sep=0.75pt]    {$\star ^{( 1)} |_{D-2}$};
\draw (480.5,262.74) node [anchor=north west][inner sep=0.75pt]    {$\star ^{( r)} |_{D-2}$};
\draw (405.68,403.49) node [anchor=north west][inner sep=0.75pt]    {$\mathbb{C}^{n_{p_{1} ,...,p_{N}}}$};
\draw (94.55,404.36) node [anchor=north west][inner sep=0.75pt]    {$\mathbb{C}^{n_{q_{1} ,...,p_{N}}}$};
\draw (504.9,394.31) node [anchor=north west][inner sep=0.75pt]    {$f_{r}$};
\draw (256.06,394.31) node [anchor=north west][inner sep=0.75pt]    {$f_{1}$};
\draw (375.27,185.42) node [anchor=north west][inner sep=0.75pt]  [rotate=-0.36]  {$\Omega _{\textnormal{AS}}^{p_{1} +1\otimes ...\otimes p_{N} +1}$};
\draw (69.89,186.29) node [anchor=north west][inner sep=0.75pt]  [rotate=-0.36]  {$\Omega _{\textnormal{AS}}^{q_{1} +1\otimes ...\otimes p_{N} +1}$};
\draw (341.05,130.28) node [anchor=north west][inner sep=0.75pt]    {$\star ^{( i)}$};
\draw (336.58,226.06) node [anchor=north west][inner sep=0.75pt]    {$\star ^{( i)} |_{D-2}$};
\draw (356.73,358.31) node [anchor=north west][inner sep=0.75pt]    {$f_{i}$};
\draw (290.55,281.37) node [anchor=north west][inner sep=0.75pt]  [rotate=-268.92]  {$\delta ^{( N)}$};
\draw (304.7,297.67) node [anchor=north west][inner sep=0.75pt]  [rotate=-359.84]  {$0$};
\draw (264.33,210.17) node [anchor=north west][inner sep=0.75pt]  [rotate=-0.36]  {$\Omega _{\textnormal{AS}}^{q_{1} \otimes ...\otimes q_{i} \otimes ...\otimes p_{N}}$};
\draw (302.56,42.9) node [anchor=north west][inner sep=0.75pt]  [rotate=-0.45]  {$0$};
\draw (288.3,189.21) node [anchor=north west][inner sep=0.75pt]  [rotate=-268.92]  {$\delta ^{( N)}$};
\draw (288.62,103.13) node [anchor=north west][inner sep=0.75pt]  [rotate=-268.92]  {$\delta ^{( N)}$};
\draw (283.55,335.4) node [anchor=north west][inner sep=0.75pt]    {$\mathbb{C}^{n_{q_{1} ,...,q_{i} ,...,p_{N}}}$};
\draw (191.38,90.04) node [anchor=north west][inner sep=0.75pt]  [rotate=-329.38]  {$...$};
\draw (493.8,82.66) node [anchor=north west][inner sep=0.75pt]  [rotate=-34.08]  {$...$};
\draw (238.33,112.17) node [anchor=north west][inner sep=0.75pt]  [rotate=-0.36]  {$\Omega _{\textnormal{AS}}^{q_{1} +1\otimes ...\otimes q_{i} +1\otimes ...\otimes p_{N} +1}$};
\draw (479,345.4) node [anchor=north west][inner sep=0.75pt]    {$\pi _{0}$};
\draw (677,342.4) node [anchor=north west][inner sep=0.75pt]    {$\pi _{r}$};
\draw (38,345.4) node [anchor=north west][inner sep=0.75pt]    {$\pi _{1}$};
\draw (237,306.4) node [anchor=north west][inner sep=0.75pt]    {$\pi _{i}$};
\draw (533.33,185.17) node [anchor=north west][inner sep=0.75pt]  [rotate=-0.36]  {$\Omega _{\textnormal{AS}}^{q_{1} +1\otimes ...\otimes q_{r} +1\otimes ...\otimes p_{N} +1}$};
\draw (544.33,272.17) node [anchor=north west][inner sep=0.75pt]  [rotate=-0.36]  {$\Omega _{\textnormal{AS}}^{q_{1} \otimes ...\otimes q_{r} \otimes ...\otimes p_{N}}$};
\draw (567.55,403.4) node [anchor=north west][inner sep=0.75pt]    {$\mathbb{C}^{n_{q_{1} ,...,q_{r} ,...,p_{N}}}$};
\end{tikzpicture}
\end{equation}
Moreover, every $f_i \in \{f_1,...,f_r\}$ admits a unique restriction to a 1-dimensional subspace such that 
\begin{equation}
    f_i|_{Q_i}: Q_0 \mapsto Q_i, \qquad f_i^{-1}|_{Q_i}: Q_i \mapsto Q_0. 
\end{equation}    
\end{Theorem}
\begin{proof}
    First of all, from the exact short sequences we get that $\delta^{(N)}$ is an isomorphism from the space of gauge fields and the space of their field strengths. Since for every $i \in [1,r]$, $\star^{(i)}$ is an isomorphism we can define 
    \begin{equation}
        \star|_{D-2}^{(i)}:=(\delta^{(N)})^{-1} \circ \star^{(i)} \circ \delta^{(N)},
    \end{equation}
which sends $\Omega _{\textnormal{AS}}^{p_{1} \otimes ...\otimes p_{N}}$ in a subspace of $\Omega _{\textnormal{AS}}^{q_{1}\otimes ...\otimes q_{i}\otimes ...\otimes p_{N}}$ of dimension $n_{p_1,...,p_N}$ since it is an isomorphism. Let us construct the $\pi_i$ maps as 
\begin{equation}
    v^{(i)}:=\begin{cases}
        \pi_i(T_i)=0 \ \ \ \textnormal{if} \  T_i\equiv 0;\\
        \pi_i(T_i)=Q^{(T_i)}_{i}e^{(i)}_1+\sum_{k=2}^{n_{p_1,...,p_N}} b^{(T_i)}_ke^{(i)}_k \ \ \ \textnormal{otherwise},
    \end{cases}
\end{equation}
where $T_i \in \Omega _{\textnormal{AS}}^{q_{1}\otimes ...\otimes q_{i}\otimes ...\otimes p_{N}}$, $e^{(i)}_1,...,e^{(i)}_{n_{p_1,...,p_N}}$ is an orthonormal base of $\mathbb{C}^{n_{p_1,...,p_N}}$ and $b^{(T_i)}_k\in \mathbb{R}$ are some of the independent entries of the in coordinate representation of the form chosen in such a way that different mixed symmetry tensors have a different string objects. These maps are all linear bijections; therefore the $f_i$ maps are well defined as $f_i:=\pi_i \circ \star^{(i)}_{D-2} \circ \pi_0^{-1}$. Hence, $f_i \in Aut(\mathbb{C}^{n_{p_1,...,p_N}})=\mathrm{GL}(\mathbb{C}^{n_{p_1,...,p_N}}) \cong \mathrm{GL}(n_{p_1,...,p_N},\mathbb{C})$. By a general theorem of linear algebra, every $f_i$ can be chosen to be lower triangular; calling $\eta_i \in \mathbb{C} \setminus \{0\}$ the top left element of $f_i$, we have
\begin{equation}
\eta_i=\frac{Q^{(T_i)}_{i}}{Q^{(T_0)}_{0}};
    \label{etai}
\end{equation}
written in other way 
\begin{equation}
Q^{(T_i)}_{i}=f_i|_{Q_i}(Q^{(T_0)}_{0})
\end{equation}
where $f_i|_{Q_i}(\bullet)=\eta_i \bullet$. 
\end{proof}

\section{Physical applications}\label{phy}
Here we discuss some physical applications of the existence and uniqueness of the duality maps and of the isomorphism between de Rham and de Rham-like cohomology groups. 
\subsection{Mixed symmetry memory effects}
Memory effects \cite{Strominger2018,Godazgar2022,Kumar2022,Mitman2024,Pasterski_2016,Comp_re_2020} are a class of observable phenomena that characterize the passage of radiation that interferes with a test charge and whose effect persists even when the radiation is extinguished. For example, a pair of test masses may undergo a non-vanishing relative displacement after the passage of gravitational radiation or a small electric charge initially at rest, may display a non-vanishing velocity after it is hit by electromagnetic radiation. These effects have also been proposed in Yang-Mills theories and in the context of the interaction of a 2-form with a test string.

Let \( T \) be a mixed symmetry tensor, and
\( H = \delta^{(N)} T \) its gauge–invariant field strength. We assume linear source–free equations of motion
\begin{equation}\label{eq:eom}
\bigl(\delta^{(N)}\bigr)^{\dagger} H = 0 , \qquad \delta^{(N)} H = 0 \, .
\end{equation}
The first one is the analogue of the equation on motion while the second one is the Bianchi identity ; for example it reduces  to the Maxwell equation without sources if we are considering a $p=1$ form. We can impose a radiative fall–off of the components with one \( u \) index defining  an object analogue to the Bondi news tensor. This mixed news is a smooth angular tensor \( N(u,x) \) with the same Young symmetry of the field strength $H$. The natural choice for the radiative component of the field strength $H_{\text{rad}}$ is a radiation fall-off of the kind $1/r^{\frac{D-2}{2}}$; in Bondi coordinates and using an orthonormal basis in the sphere we have
\begin{equation}\label{eq:newsfalloff}
H_{\text{rad}}(u,r,x)
= \frac{1}{r^{\frac{D-2}{2}}}\, N(u,x)
+ \mathcal{O}\!\left(r^{-\frac{D}{2}}\right).
\end{equation}

We define the electric–like charge at retarded time \( u \) as the pairing between \( \varepsilon \) and \( H \) on a large radius \( S^{D-2} \)
\begin{equation}\label{eq:charge}
Q(u) := \int_{S^{D-2}_{u}} \langle \varepsilon(x) \,,\, H_{\text{rad}}(r,u,x) \rangle,
\end{equation}
where \( \varepsilon \) is an appropriate combination of asymptotic symmetry parameters\footnote{For a mixed symmetry  tensor there can be more than one gauge parameter.} on the sphere, such that its leading order term satisfy \( \partial_{u}\varepsilon=0 \) and  \( \langle \cdot,\cdot \rangle \) is the natural contraction induced by the metric on \( S^{D-2} \), the Levi-Civita tensor and the Young structure of the mixed field. Using Stokes’ theorem in a null tube between \( u=u_i \) and \( u=u_f \) yields the balance law
\begin{equation}\label{eq:balance}
\Delta Q = Q(u_f)-Q(u_i)
= \int_{\mathcal{I}^{+} \times [u_i,u_f]} \langle \varepsilon(x) \,,\,  \partial_{u} H_{\text{rad}}(r,u,x)\rangle\, \mathrm{d}u\, \mathrm{d}\Omega_{D-2} \, .
\end{equation}
Using the ansatz \eqref{eq:newsfalloff} and integrating by parts in $u$, we finds directly
\begin{equation}
\int_{u_i}^{u_f} \langle \varepsilon(x), \partial_u N(u,x) \rangle \, du
= \Big[ \langle \varepsilon(x), N(u,x) \rangle \Big]_{u_i}^{u_f}
- \int_{u_i}^{u_f} \langle \partial_u \varepsilon(x), N(u,x) \rangle \, du;
\end{equation}
by assumption \(\partial_u \varepsilon(x) = 0\) so the second term vanishes, 
\begin{equation}
\int_{u_i}^{u_f} \langle \varepsilon(x), \partial_u N(u,x) \rangle \, du
= \langle \varepsilon(x), N(u_f,x) \rangle
- \langle \varepsilon(x), N(u_i,x) \rangle \, .
\end{equation}
Therefore, in the end we get
\begin{equation}\label{eq:balanceNews}
\Delta Q
= \int_{S^{D-2}} \left\langle \varepsilon(x) \,,\, \int_{u_i}^{u_f} N(u,x)\,\mathrm{d}u \right\rangle \mathrm{d}\Omega_{D-2} \, .
\end{equation}
When the news has support in a finite interval of retarded time (a burst), we generally have \( \Delta Q\neq 0 \).  
This can be considered as the mixed memory effect at the level of the charge. In Theorem \ref{THM4.3.1}, when the relevant de Rham–like cohomologies vanish and the short exact sequence properties hold, the operator \( \delta^{(N)} \) induces an isomorphism between the space of asymptotic potentials and field strengths. The duality map \( f \) then carries \eqref{eq:balanceNews} into the corresponding identity in the dual channel. For example:
\begin{equation}\label{eq:duality}
\Delta Q_0
= Q_0(u_f)-Q_0(u_i)
\;\;\Longleftrightarrow\;\;
\Delta Q_j
=  Q_j(u_f)-Q_j(u_i)=f_j(Q_j(u_f))-f_j(Q_j(u_i)).
\end{equation}
The charge of the $j$-th description will be expressed in term of the its radiative data and asymptotic parameters. In other words, the existence and uniqueness of duality maps implies that the memory measured in one channel uniquely determines the memory in the dual channel. This could open the door to an experimental verification of extra dimensions: by measuring the gravitational memory effect, one could deduce the existence, based on the number of dimensions, of a dual memory effect which, if measured in turn, would indicate an amount of extra dimensions. For example, if we lived in $D=5$ dimensions, then the graviton would have dual memory effects for a Curtright and a Riemann-like field. Measuring the dual memory effects would confirm $D=5$ dimensions.
\subsection{Higher symmetries and defects}
Higher symmetries \cite{Luo2023,Barkeshli2023,ArmasJain2024,Yuan2023,SchaferNameki2023,AxionMaxwell2024} are a new frontier of theoretical physics. Here, the idea is that, in the presence of extended defects, the asymptotic charges of fields with mixed symmetry are classified by de Rham-like cohomology groups which are isomorphic to the ordinary de Rham cohomology of the "punctured" space. This provides a natural quantization of charges. 

Let \(M = \mathbb{R}^D \setminus W
\)
where \( W \subset \mathbb{R}^D \) is a smooth closed submanifold of dimension 
\(\dim (W) = D-k\), hence of codimension \(k \geq 2\).  By deformation, any tubular neighborhood of \( W \) in \( M \) retracts onto the \( S^{k-1} \)-bundle over \( W \). For local computations in a punctured normal slice, we can use the model
\begin{equation}
M \cong_{\text{loc}} \big(\mathbb{R}^{D-k} \times (\mathbb{R}^k \setminus \{0\})\big) \cong \mathbb{R}^{D-k} \times S^{k-1} \times \mathbb{R}_+ .
\end{equation}
We have
\begin{equation}
H^m_{\mathrm{dR}}(M) \cong_{\text{loc}} H^m_{\mathrm{dR}}(S^{k-1}) ,
\end{equation}
and therefore
\begin{equation}
H^m_{\mathrm{dR}}(M) \cong_{\text{loc}} 
\begin{cases}
\mathbb{R}, & m=0 \ \text{or} \ m=k-1 , \\
0, & \text{otherwise} .
\end{cases}
\end{equation}
The generator in degree \(k-1\) is the normalized volume form \(\omega_{S^{k-1}}\) on the spheres surrounding \( W \). From Theorem \ref{con4.3.1} we have a natural isomorphism
\begin{equation}\label{eqcom}
H^{(\cdot,...)}_{\mathrm{dR-like}}(M) \cong H^{(\cdot)}_{\mathrm{dR}}(M) ,
\end{equation}
so that the topological classes counting the mixed charges coincide with the usual ones in the appropriate degree. 

Now, let $T$ a mixed symmetry  gauge field and consider the gauge-invariant field strength $H = \delta^{(N)} T $. 
In the absence of sources we have \(\delta^{(N)} H=0\); with a defect \( W \) carrying integer charge \( n\in \mathbb{Z}\), the Bianchi identity deforms to
\begin{equation}
\delta^{(N)} H = 2\pi n \, \Delta_W ,
\end{equation}
where \(\Delta_W\) is the Poincaré current of \( W \), Young-projected consistently with the symmetry of the mixed symmetry field. This is the analogue of the most elementary case in $\mathbb{R}^3$ with a pointlike magnetic monopole at the origin. The Bianchi identity acquire a non homogeneous term of the form $2\pi n \, \delta^{(3)}(x) $ with
the delta has support only at the origin. Therefore $H$ is no longer exact globally. Furthermore, since $M$ has no trivial de Rham cohomolgy in degree $k-1$ we can consider
\begin{equation}
\frac{1}{2\pi}\int_{S^{k-1}} [H] = n \in \mathbb{Z}.
\end{equation}
Where $[H]$ is the cohomology class associated to the mixed symmetry tensor field strength which is in fact a de Rham cohomology class due \eqref{eqcom}. This quantity depends only on the cohomology class in \(H^{k-1}_{\mathrm{dR-like}}(M) \cong \mathbb{R}\) and is therefore topologically protected. This point in the direction of mixed symmetry tensor symmetries where the charged object are extended operators that can be considered much similar to those charged under higher form symmetries.

\subsection{Holography}
In AdS/CFT \cite{Natsuume:2014sfa,Maldacena_1999,Skenderis:2002wp,Aharony:1999ti,Witten:1998qj,Gubser:1998bc}, bulk asymptotic charges correspond to global symmetries and conserved currents of the dual boundary CFT. The generalization to mixed symmetry tensors implies that new classes of conserved quantities arise, associated with boundary operators supported on submanifolds of various codimensions.

Consider $\mathrm{AdS}_{D+1}$ with Poincar\'e patch metric
\begin{equation}
  ds^{2} \;=\; \frac{L^{2}}{z^{2}}\bigl(dz^{2} + \eta_{\mu\nu}\,dx^{\mu}dx^{\nu}\bigr) \, , 
  \qquad \mu=0,\ldots,D-1 \, , \quad z>0 \, .
\end{equation}
Let $T$ be a mixed symmetry tensor in the bulk and let $H= \delta^{(N)} T$ be its field strength. Following the standard lore of the holographic dictionary, near the boundary $z\to 0$, solutions of the equation of motion can be expanded as
\begin{equation}
  T(z,x) \;\sim\; z^{\Delta_{-}}\, t^{(0)}(x) \;+\; z^{\Delta_{+}}\, t^{(1)}(x) \;+\; \cdots \, ,
\end{equation}
with $\Delta_{\pm}$ appropriate exponents which could rise from a generalized mass dimension relation depending on the Young tableau of $T$. The leading term $t^{(0)}$ sources a boundary operator, while $t^{(1)}$ is related to its vacuum expectation value.

The holographic dictionary identifies
\begin{equation}
  Z_{\text{bulk}}\bigl[t^{(0)}\bigr]
  \;=\; \Bigl\langle \exp\!\Bigl(\int_{\partial \mathrm{AdS}} \!\!\langle t^{(0)}, \, O \rangle \Bigr) \Bigr\rangle_{\!\text{CFT}},
\end{equation}
hence the bulk partition function with boundary data $t^{(0)}$ coincides with the generating functional of CFT correlators in the presence of a source $t^{(0)}$ coupled to the operator $O$. Here $\langle t^{(0)}, O \rangle$ denotes the natural pairing between indices consistent with the Young symmetries of the field. Gauge invariance in the bulk, expressed by a differential operators acting on gauge parameters, induces invariance of the generating functional. This invariance is expressed by a boundaru gauge variation of the boundary data. As a result, the CFT operator satisfies the Ward identities implemented by the boundary versions of the differential operators. The asymptotic charge in the bulk is
\begin{equation}
  Q\;=\; \int_{\partial \mathrm{AdS}} \!\langle \varepsilon (x), H(z,x) \rangle \, ,
\end{equation}
with $\varepsilon$ a combination of asymptotic gauge parameters. Its variation under bulk dynamics yields
\begin{equation}
  \delta Q \;=\; \int_{\partial \mathrm{AdS}} \!\langle \varepsilon , \, \delta_{\partial} O \rangle \, .
\end{equation}
where the operator $\delta_\partial$ contains all the boundary operators inherited from the differential operators implementing the gauge transformations. 
By the Ward identities, $\delta Q_{\varepsilon}=0$ and thus, the bulk conservation of $Q_{\varepsilon}$ is holographically equivalent to the CFT Ward identity. This is true if the theory is anomalous-free and there are no boundary or contact terms that break the symmetry.

Let us now focus on $D=5$, where there exist a duality between the graviton and the Curtright field. In a $D=5$ AdS bulk (dual to a $4d$ CFT), the graviton $h_{\mu\nu}$ is holographically dual to the stress tensor $T_{\mu\nu}$, while the Curtright field $C_{\mu\nu\mid\rho}$ is dual to a mixed symmetry operator $O_{\mu\nu\mid\rho}$ with Young tableau $(2,1)$. Since the universal covering of AdS$_5$, which is the physically relevant space free from closed time-like curves, is topologically $\mathbb{R}^5$, theorems \ref{THM4.3.1} holds and implies a unique map
\begin{equation}
  Q_{h} \;\Longleftrightarrow\; Q_{C} \, .
\end{equation}
Holographically speaking, this translates into the relation of two Ward identities
\begin{equation}
  \partial_{\mu} T^{\mu}{}_{\nu} \;=\; 0 
  \;\Longleftrightarrow\; 
  \delta_{\partial}^{(2)} O_{\mu\nu\mid\rho} \;=\; 0 \, .
\end{equation}
Thus, the duality in the bulk enforces non-trivial constraints among distinct operator sectors of the boundary CFT, coupling the stress tensor with exotic mixed-symmetry operators. In general this relations may implies that certain operator algebras in the CFT must close into each other, thereby reducing the freedom in bootstrap equations.
\subsection{Condensed matter and fractons}
Fracton phases of matter \cite{PretkoChenYou2020,NandkishoreHermele2019,Gromov2019,Prem2019,Shirley2018,Vijay2016} are characterized by higher-rank conservation laws that restrict the mobility of excitations. In continuum dual descriptions, these conservation laws are naturally written in terms of symmetric or mixed-symmetry tensor gauge fields. The generalized de Rham-like formalism provides a clean cohomological interpretation of such conservation laws.

A prototypical example is the scalar-charge rank-2 $U(1)$ theory, with symmetric tensor gauge potential $A_{ij}$ and Gauss law
\begin{equation}
\partial_i \partial_j E^{ij} = \rho \, .
\label{eq:gauss}
\end{equation}
Here $E^{ij}$ is the generalized electric field, symmetric in $i,j$. In absence of sources we have
\begin{equation}
\partial_i \partial_j E^{ij} = 0 \, .
\label{eq:constraint}
\end{equation}
This implies not only charge conservation
\begin{equation}
\frac{d}{dt} \int d^3x \, \rho(x) = 0 \, ,
\label{eq:charge_cons}
\end{equation}
but also dipole conservation
\begin{equation}
\frac{d}{dt} \int d^3x \, x^k \rho(x) = 0 \, .
\label{eq:dipole_cons}
\end{equation}
The constraint \eqref{eq:constraint} is exactly of the form
\begin{equation}
\delta^{(2)} T = 0 \, , \qquad T \in \Omega^{1 \otimes 1} \, ,
\end{equation}
where $T$ plays the role of the gauge field (a bi-form, projected to the symmetric sector). Thus, the immobility of isolated fractons is rephrased as the closure condition in the de Rham-like complex.

More exotic fracton models (e.i. with vector charges or multipole conservation beyond dipole) could be described by higher mixed-symmetry tensors
\begin{equation}
T \in \Omega^{q_1} \otimes \cdots \otimes \Omega^{q_N} \,
\end{equation}
properly Young projected and the closure condition
\begin{equation}
\delta^{(N)} T = 0
\end{equation}
implies conservation of higher multipole moments. This cohomological rewriting could lead to a purely cohomological formulation of the allowed ground-state degeneracy of a given fracton model.
\subsection{String theory}
In the tensionless limit of string theory, an infinite tower of massless fields emerges, including not only 
$p$-form gauge fields (Kalb--Ramond, RR forms) but also mixed-symmetry tensors \cite{Grana2005,Douglas2007,VanRietZoccarato2023,Shukla2021,SagnottiTsulaia2003,LindstromZabzine2003}. Their asymptotic charges can be naturally classified by the generalized de Rham-like cohomology introduced in this work.

Consider type II string theory compactified on a torus $T^d$. A $p$-brane wrapping a cycle 
$\Sigma_p \subset T^d$ carries a charge
\begin{equation}
Q^{(p)} = \int_{\Sigma_p} F_{p+1} \,,
\end{equation}
where $F_{p+1} = dC_p$ is the RR flux. The integrality condition
\begin{equation}
Q^{(p)} \in 2\pi \mathbb{Z}
\end{equation}
comes from the fact that $[F_{p+1}] \in H^{p+1}(T^d,\mathbb{Z})$. For a mixed-symmetry tensor $T$ the field strength
\begin{equation}
H = \delta^{(N)} T
\end{equation}
defines a generalized flux. $H$ defines a de Rham-like cohomology class which by Theorem \ref{con4.3.1} is also a de Rham cohomology class. Therefore could exist also mixed symmetry tensor flux in string theory compactification and they can be classified in the usual way.

In type II theories, T-duality acts by reshuffling cohomology classes on $T^d$. From the above theorems, it follows that:
\begin{itemize}
  \item brane wrapping numbers (topological charges) in one channel map uniquely into mixed-symmetry charges in the dual channel;
  \item the spectrum of allowed charges is determined entirely by $H^\bullet_{\mathrm{dR}}(T^d,\mathbb{Z})$, independent of the field representation.
\end{itemize}

For instance, on $T^2$ a string wrapping along $y_1$ (an element of $H_{\text{dR}}^1$) can dualize into a mixed bi-form flux (an element of $H^{(2,2)}_{\text{dR}-\text{like} }\cong H^2_{\text{dR}}$), in agreement with T-duality predictions.

The classification of mixed fluxes provides a systematic way to include exotic brane states predicted in the tensionless limit. The cohomological framework prevents overcounting: all exotic charges reduce to integer classes in the ordinary cohomology of the compactification manifold.

\section{Conclusions}
In this work, we discussed the asymptotic symmetries and the Young machinery duality in the realm of mixed symmetry tensor gauge theories where the gauge field is a mixed symmetry tensor $T$. We discuss how to extend Theorem \ref{THM3.1} to mixed symmetry tensors case. In Paragraph \ref{derhamcomgen}, we discuss and develop a generalization of the de Rham complex to the case of mixed symmetry tensors. This leads to the Definition \ref{Naugcomp} of $(k_1,...,k_{N-1})$-augmented $N$-de Rham-like complex and to Definition \ref{coholike} of de Rham-like cohomology groups which characterize the topology of the differential manifold under consideration thanks to Theorem \ref{con4.3.1}.  These abstract mathematical tools are employed, in Paragraph \ref{exthm}, to prove Theorem \ref{THM4.3.1} on the existence and uniqueness of a set of duality maps for well defined charges. This theorem is the direct generalization of Theorem \ref{THM3.1} and, once we have at our disposal the $(k_1,...,k_{N-1})$-augmented $N$-de Rham-like complexes, can be prove using the same ideas of Theorem \ref{THM3.1}. The physical meaning is that given a description and its asymptotic charge, this charge has access to physical information related to asymptotic charges of the dual formulations. This means that the symmetries of a gauge theory are intimately related to the symmetries of the dual formulations, and under suitable topological conditions there exists a unique way to associate the charges of these symmetries with each other. Therefore, to understand the physics and the physical information associated with a particular gauge theory it is essential to know the mathematical properties of the space on which the theory is built and not only the content in fields and the properties of those fields. The property of being able to uniquely map the asymptotic charges of the dual descriptions of a specific formulation of a mixed symmetry tensor gauge theory could allow a mathematical classification of the spaces on which these theories are formulated. In this sense, it would be essential to construct particular cohomology classes that, when vanishing, make possible the existence and uniqueness of the duality map and, possibly, a vice versa.\\
Furthermore we provide some physical applications that can shed new light in various fields of physics. We discuss mixed symmetry memory effects, which naturally emerge at the intersection of geometry, field theory and quantum gravity. Thanks to the duality theorem, these effects are not just formal artifacts but acquire a striking physical relevance: they open the intriguing possibility of probing extra dimensions through the very act of measuring memory effects in their dual descriptions. This idea resonates with the long-standing dream of unveiling hidden structures of spacetime by means of subtle, yet robust, physical imprints.
The applications to higher symmetries and string theory become particularly compelling once one recalls the isomorphism between de Rham-like cohomology and the traditional de Rham cohomology. In essence, extended defects do not merely couple to conventional fields but can also source mixed symmetry fields. Remarkably, this spectrum is could be fully dictated by the de Rham cohomology of the compactification manifold, encoding in purely topological terms the algebra of possible excitations and their consistency.
Furthermore, the closure property of the de Rham-like complex seems to provide the natural framework for a cohomological reinterpretation of fracton dynamics. Within this setting, the immobility and restricted motion of fractons could emerge as a reflection of hidden cohomological structures.
Finally, the holographic perspective elevates the discussion to an even broader stage. In the dual conformal field theory, the operator algebras must close upon themselves, enforcing nontrivial consistency conditions. This closure principle could reduces the arbitrariness in the bootstrap equations, carving out a more rigid and predictive landscape of possible CFTs. In this way, memory effects, higher symmetries, and holography are interlaced into a unified narrative: they point towards a profound cohomological skeleton underlying quantum field theories and string dynamics, a skeleton that may ultimately guide us toward the architecture of spacetime itself. All these physical applications can be explored separately and lead to a deeper understanding of these fields of physics.
\appendix 

\section{Basic elements of shaves theory}\label{appfasci}
We give a brief review of shaves theory and their cohomology with the aim to proof the abstract de Rham theorem. In the following we consider X as a Hausdorff topological space where every open set is paracompact. An example are all the metrizable spaces.
\subsection{Sheaves}
Intuitively, a sheaf is a tool for systematically tracking data (such as sets, abelian groups, rings) attached to the open sets of a topological space and defined locally with regard to them. More formally
\begin{Definition}[Sheaf]
A sheaf of abelian groups $\mathcal{F}$ on X is the datum of an abelian group $\mathcal{F}(U)$ for every open set $U \subset X$ and a group homomorphism (called restrictions) $\rho_{UV}: \mathcal{F}(U) \rightarrow \mathcal{F}(V)$ for every inclusion $V \subset U$. The couple $(U,s)$ with $U$ open set of $X$ and $s \in \mathcal{F}(U)$ is called a section. The datum has to satisfy the following conditions:
\begin{itemize}
    \item $\mathcal{F}(\emptyset)=\emptyset$;
    \item $\rho_{UU}: \mathcal{F}(U) \rightarrow \mathcal{F}(U)$ is the identity for every $U$;
    \item if $W \subset V \subset U$ then $\rho_{UW}=\rho_{UV} \circ \rho_{VW}$;
    \item if $U=\bigcup_i U_i$ and $s \in \mathcal{F}(U)$ than $s=0$ if and only if $\rho_{UU_i}(s)=0 \  \forall U_i$; 
    \item given, for every $i$, a section $s_i \in \mathcal{F}(U_i)$ such that $\rho_{U_iU_i \cap U_j}(s_i)=\rho_{U_jU_i \cap U_j}(s_j)$ than there exist a section $s \in \mathcal{F}(U)$ such that $\rho_{UU_i}(s)=s_i \ \forall i$.
\end{itemize}
\end{Definition}
Examples of sheaves are the sheaf of locally constant functions on a field $\mathbb{K}$, denoted by $\mathbb{K}_X$, and the sheaf of discontinues sections of a sheaf $\mathcal{F}$, denoted by $\mathcal{DF}$ and defined by $\mathcal{DF}(U):=\prod_{x \in U} \mathcal{F}_x$.
\begin{Definition}[Germs and stalks of a sheaf]
    Given two sections $(U,s)$ and $(V,t)$ they are equivalent if there exist an open $W\subset U \cap V$ such that $\rho_{UW}(s)=\rho_{VW}(t)$. The equivalence class $[U,s]$ is called a germ and the abelian group $\mathcal{F}_x$ on the set of germs around the point $x \in X$ is called a stalk.
\end{Definition}

\begin{Definition}[Shaves morphism and its support]
    A  shave morphism of abelian groups $f:\mathcal{F} \rightarrow \mathcal{F}$ is a family of group homomorphisms $f_U : \mathcal{F}(U) \rightarrow \mathcal{F}(U)$ for every open $U \subset X$ which commutes with the restrictions. A  shave morphism naturally induces a morphism on the stalks $f_x : \mathcal{F}_x \rightarrow \mathcal{F}_x $ and the support of a shave morphism is given by 
    \begin{equation}
        Supp(f)=\{x \in X \ | \  f_x(x) \neq 0\}
    \end{equation}
\end{Definition}

\begin{Definition}[Exact sequence of shaves]
    A sequence of shaves on $X$
\begin{equation}
    \tikzset{every picture/.style={line width=0.75pt}} 
\begin{tikzpicture}[x=0.75pt,y=0.75pt,yscale=-1,xscale=1]

\draw    (80,191) -- (138.73,190.84) ;
\draw [shift={(140.73,190.83)}, rotate = 179.84] [color={rgb, 255:red, 0; green, 0; blue, 0 }  ][line width=0.75]    (10.93,-3.29) .. controls (6.95,-1.4) and (3.31,-0.3) .. (0,0) .. controls (3.31,0.3) and (6.95,1.4) .. (10.93,3.29)   ;
\draw    (260,191) -- (318.73,190.84) ;
\draw [shift={(320.73,190.83)}, rotate = 179.84] [color={rgb, 255:red, 0; green, 0; blue, 0 }  ][line width=0.75]    (10.93,-3.29) .. controls (6.95,-1.4) and (3.31,-0.3) .. (0,0) .. controls (3.31,0.3) and (6.95,1.4) .. (10.93,3.29)   ;
\draw    (170,191) -- (228.73,190.84) ;
\draw [shift={(230.73,190.83)}, rotate = 179.84] [color={rgb, 255:red, 0; green, 0; blue, 0 }  ][line width=0.75]    (10.93,-3.29) .. controls (6.95,-1.4) and (3.31,-0.3) .. (0,0) .. controls (3.31,0.3) and (6.95,1.4) .. (10.93,3.29)   ;

\draw (147,181.4) node [anchor=north west][inner sep=0.75pt]    {$\mathcal{F}$};
\draw (238,181.4) node [anchor=north west][inner sep=0.75pt]    {$\mathcal{E}$};
\draw (101,162.4) node [anchor=north west][inner sep=0.75pt]    {$f$};
\draw (277,166.4) node [anchor=north west][inner sep=0.75pt]    {$h$};
\draw (190,163.4) node [anchor=north west][inner sep=0.75pt]    {$g$};
\draw (49,187.4) node [anchor=north west][inner sep=0.75pt]    {$...$};
\draw (334,187.4) node [anchor=north west][inner sep=0.75pt]    {$...$};
\end{tikzpicture}
\end{equation}
is called exact if $\forall x \in X$ the sequence of stalk
\begin{equation}
\tikzset{every picture/.style={line width=0.75pt}} 
\begin{tikzpicture}[x=0.75pt,y=0.75pt,yscale=-1,xscale=1]

\draw    (80,191) -- (138.73,190.84) ;
\draw [shift={(140.73,190.83)}, rotate = 179.84] [color={rgb, 255:red, 0; green, 0; blue, 0 }  ][line width=0.75]    (10.93,-3.29) .. controls (6.95,-1.4) and (3.31,-0.3) .. (0,0) .. controls (3.31,0.3) and (6.95,1.4) .. (10.93,3.29)   ;
\draw    (260,191) -- (318.73,190.84) ;
\draw [shift={(320.73,190.83)}, rotate = 179.84] [color={rgb, 255:red, 0; green, 0; blue, 0 }  ][line width=0.75]    (10.93,-3.29) .. controls (6.95,-1.4) and (3.31,-0.3) .. (0,0) .. controls (3.31,0.3) and (6.95,1.4) .. (10.93,3.29)   ;
\draw    (170,191) -- (228.73,190.84) ;
\draw [shift={(230.73,190.83)}, rotate = 179.84] [color={rgb, 255:red, 0; green, 0; blue, 0 }  ][line width=0.75]    (10.93,-3.29) .. controls (6.95,-1.4) and (3.31,-0.3) .. (0,0) .. controls (3.31,0.3) and (6.95,1.4) .. (10.93,3.29)   ;

\draw (147,180.4) node [anchor=north west][inner sep=0.75pt]    {$\mathcal{F}_{x}$};
\draw (238,181.4) node [anchor=north west][inner sep=0.75pt]    {$\mathcal{E}_{x}$};
\draw (101,162.4) node [anchor=north west][inner sep=0.75pt]    {$f_{x}$};
\draw (277,166.4) node [anchor=north west][inner sep=0.75pt]    {$h_{x}$};
\draw (190,163.4) node [anchor=north west][inner sep=0.75pt]    {$g_{x}$};
\draw (49,187.4) node [anchor=north west][inner sep=0.75pt]    {$...$};
\draw (334,187.4) node [anchor=north west][inner sep=0.75pt]    {$...$};
\end{tikzpicture}
\end{equation}
is exact.

\end{Definition}

\begin{Definition}[Partition of the identity]
    Let $\mathcal{F}$ be a sheaf and $\mathcal{U}=\{U_i\}_{i \in I}$ an open covering of $X$. A partition of the identity of $\mathcal{F}$ subordinated to $\mathcal{U}$ is a family of shaves morphisms $f_i: \mathcal{F} \rightarrow \mathcal{F}$ such that
    \begin{itemize}
        \item $\overline{Supp(f_i)} \subset U_i \ \forall i$; 
        \item $\{Supp(f_i)\}_{i\in I}$ forms a locally finite covering of $X$;
        \item $\sum_{i \in I} f_i=id_{\mathcal{F}}$.
    \end{itemize}
\end{Definition}

\begin{Definition}[fine sheaf]
    A sheaf $\mathcal{F}$ on $X$ is fine if it admits a partition of the identity subordinate to every open covering $\mathcal{U}$ of $X$.
\end{Definition}

\subsection{Shaves cohomology}
Čech cohomology is a cohomology theory based on the intersection properties of open covers of a topological space. Given an open covering $\mathcal{U}=\{U_i\}_{i \in I}$ we call $U_{i_0...i_q}=U_{i_0} \cap ... \cap U_{i_q}$. 
\begin{Definition}[Čech cochain and Čech differential]
    We define the Čech $q$-cochain (with $q>0$) as an element of 
\begin{equation}
    \Check{C}^q(\mathcal{U},\mathcal{F}):=\prod_{(i_0,...,i_q) \in I^{q+1}}\mathcal{F}(U_{i_0...i_q}),
\end{equation}
and Čech differential $\Check{\delta} : \Check{C}^q(\mathcal{U},\mathcal{F}) \rightarrow \Check{C}^{q+1} (\mathcal{U},\mathcal{F})$ as
\begin{equation}
    (\Check{\delta} c)_{i_0...i_q+1}:=\sum_{k=0}^{q+1}(-1)^kc_{i_0...i_{k-1}i_{k+1}...i_q+1}|U_{i_0...i_q+1}.
\end{equation}
\end{Definition}
Is a matter of computations to show that $\Check{\delta} \circ \Check{\delta}=0$. Hence $(\Check{C}^*,\Check{\delta})$ is a cochain complex known as Čech complex and its cohomology is known as Čech or shaves cohomology
\begin{Definition}[Čech cohomology]
    The q-th Čech cohomology group is 
    \begin{equation}
\Check{H}^q(\mathcal{U},\mathcal{F}):=\frac{\Check{Z}^q(\mathcal{U},\mathcal{F})}{\Check{B}^q(\mathcal{U},\mathcal{F})}=\frac{\{c \in \Check{C}^q(\mathcal{U},\mathcal{F}) | \Check{\delta} c=0\}}{\{\Check{\delta} c \in \Check{C}^{q}(\mathcal{U},\mathcal{F}) | c \in \Check{C}^{q-1}(\mathcal{U},\mathcal{F})\}}.
    \end{equation}
\end{Definition}
Thanks to a mathematical procedure known as colimit \cite{lane1998categories} we can eliminate the dependence on the open covering of $X$ in the Čech complex (and hence also in its cohomology groups) by replacing it directly with the dependence on the topological space $X$.
\begin{Definition}[Acyclic sheaf]
    A sheaf $\mathcal{F}$ on $X$ is called acyclic if $\Check{H}^q(X,\mathcal{F})=0 \ \forall q>0$.
\end{Definition}
\begin{Theorem}[Every fine sheaf is acyclic]
    Let $\mathcal{F}$ be a fine sheaf on $X$ then $\mathcal{F} $ is acyclic.
\end{Theorem}
\begin{proof}
    We show a stronger fact from which the theorem follows easily. The stronger assertion we want to prove is that if $\mathcal{F}$ is a sheaf on $X$ which admits a partition of the identity subordinated to an open covering $\mathcal{U}=\{U_i\}_{i \in I}$ of $X$ then $\Check{H}^q(\mathcal{U},\mathcal{F})=0 \ \forall q>0.$ In fact, let us consider the shaves $\mathcal{F}_i$ by placing $\mathcal{F}_i(U):=\mathcal{F}(U_i \cap U)$ for every open $U \subset X$ and the shaves morphisms $g_i: \mathcal{F}_i \rightarrow \mathcal{F}$ defined, for every $s \in \mathcal{F}_i(U)$, by 
    \begin{equation}
        g_i(s):=
        \begin{cases}
             f_i(s) \qquad on \ U \cap U_i\\
             0  \qquad on \ U\setminus \overline{Supp(f_i)}
        \end{cases};
    \end{equation}
where the $f_i$ are the shave morphisms of the partition of identity subordinated to $\mathcal{U}$. Choosing a cocycle $a \in \Check{Z}^q(\mathcal{U},\mathcal{F})$ and defining $b \in \Check{C}^{q-1}(\mathcal{U},\mathcal{F})$ as 
\begin{equation}
    b_{i_1...i_n}:=\sum_jg_j(a_{ji_1...i_n})
\end{equation}
is easy to show that $\Check{\delta}b=a$. The theorem now follows since a fine sheaf $\mathcal{F}$ admits a partition of the identity subordinated to every open covenging of $X$, hence passing to the colimit we get $\Check{H}^q(X,\mathcal{F})=0 \ \forall q>0.$ and so $\mathcal{F}$ is acyclic. 
\end{proof}

\begin{Theorem}[Long exact sequence in Čech cohomology]
    For every short exact sequence of shaves 
\begin{equation}
\tikzset{every picture/.style={line width=0.75pt}} 
\begin{tikzpicture}[x=0.75pt,y=0.75pt,yscale=-1,xscale=1]

\draw    (80,191) -- (138.73,190.84) ;
\draw [shift={(140.73,190.83)}, rotate = 179.84] [color={rgb, 255:red, 0; green, 0; blue, 0 }  ][line width=0.75]    (10.93,-3.29) .. controls (6.95,-1.4) and (3.31,-0.3) .. (0,0) .. controls (3.31,0.3) and (6.95,1.4) .. (10.93,3.29)   ;
\draw    (260,191) -- (318.73,190.84) ;
\draw [shift={(320.73,190.83)}, rotate = 179.84] [color={rgb, 255:red, 0; green, 0; blue, 0 }  ][line width=0.75]    (10.93,-3.29) .. controls (6.95,-1.4) and (3.31,-0.3) .. (0,0) .. controls (3.31,0.3) and (6.95,1.4) .. (10.93,3.29)   ;
\draw    (170,191) -- (228.73,190.84) ;
\draw [shift={(230.73,190.83)}, rotate = 179.84] [color={rgb, 255:red, 0; green, 0; blue, 0 }  ][line width=0.75]    (10.93,-3.29) .. controls (6.95,-1.4) and (3.31,-0.3) .. (0,0) .. controls (3.31,0.3) and (6.95,1.4) .. (10.93,3.29)   ;
\draw    (350,191) -- (408.73,190.84) ;
\draw [shift={(410.73,190.83)}, rotate = 179.84] [color={rgb, 255:red, 0; green, 0; blue, 0 }  ][line width=0.75]    (10.93,-3.29) .. controls (6.95,-1.4) and (3.31,-0.3) .. (0,0) .. controls (3.31,0.3) and (6.95,1.4) .. (10.93,3.29)   ;

\draw (147,180.4) node [anchor=north west][inner sep=0.75pt]    {$\mathcal{E}$};
\draw (238,181.4) node [anchor=north west][inner sep=0.75pt]    {$\mathcal{F}$};
\draw (277,166.4) node [anchor=north west][inner sep=0.75pt]    {$h$};
\draw (190,166.4) node [anchor=north west][inner sep=0.75pt]    {$g$};
\draw (62,182.4) node [anchor=north west][inner sep=0.75pt]    {$0$};
\draw (418,182.4) node [anchor=north west][inner sep=0.75pt]    {$0$};
\draw (328,181.4) node [anchor=north west][inner sep=0.75pt]    {$\mathcal{G}$};
\end{tikzpicture}
\end{equation}
induce a long exact sequence in Čech cohomology
\begin{equation}
\tikzset{every picture/.style={line width=0.75pt}} 
\begin{tikzpicture}[x=0.75pt,y=0.75pt,yscale=-1,xscale=1]

\draw    (121.88,188.6) -- (141,188.94) ;
\draw [shift={(143,188.98)}, rotate = 181.02] [color={rgb, 255:red, 0; green, 0; blue, 0 }  ][line width=0.75]    (10.93,-3.29) .. controls (6.95,-1.4) and (3.31,-0.3) .. (0,0) .. controls (3.31,0.3) and (6.95,1.4) .. (10.93,3.29)   ;
\draw    (28.88,188.6) -- (48,188.94) ;
\draw [shift={(50,188.98)}, rotate = 181.02] [color={rgb, 255:red, 0; green, 0; blue, 0 }  ][line width=0.75]    (10.93,-3.29) .. controls (6.95,-1.4) and (3.31,-0.3) .. (0,0) .. controls (3.31,0.3) and (6.95,1.4) .. (10.93,3.29)   ;
\draw    (217.88,188.6) -- (237,188.94) ;
\draw [shift={(239,188.98)}, rotate = 181.02] [color={rgb, 255:red, 0; green, 0; blue, 0 }  ][line width=0.75]    (10.93,-3.29) .. controls (6.95,-1.4) and (3.31,-0.3) .. (0,0) .. controls (3.31,0.3) and (6.95,1.4) .. (10.93,3.29)   ;
\draw    (308.88,189.6) -- (328,189.94) ;
\draw [shift={(330,189.98)}, rotate = 181.02] [color={rgb, 255:red, 0; green, 0; blue, 0 }  ][line width=0.75]    (10.93,-3.29) .. controls (6.95,-1.4) and (3.31,-0.3) .. (0,0) .. controls (3.31,0.3) and (6.95,1.4) .. (10.93,3.29)   ;
\draw    (412.88,190.6) -- (432,190.94) ;
\draw [shift={(434,190.98)}, rotate = 181.02] [color={rgb, 255:red, 0; green, 0; blue, 0 }  ][line width=0.75]    (10.93,-3.29) .. controls (6.95,-1.4) and (3.31,-0.3) .. (0,0) .. controls (3.31,0.3) and (6.95,1.4) .. (10.93,3.29)   ;
\draw    (521.88,190.6) -- (541,190.94) ;
\draw [shift={(543,190.98)}, rotate = 181.02] [color={rgb, 255:red, 0; green, 0; blue, 0 }  ][line width=0.75]    (10.93,-3.29) .. controls (6.95,-1.4) and (3.31,-0.3) .. (0,0) .. controls (3.31,0.3) and (6.95,1.4) .. (10.93,3.29)   ;
\draw (55.08,178.44) node [anchor=north west][inner sep=0.75pt]    {$\Check{H}^{q}( X,\mathcal{E})$};
\draw (220.01,164.91) node [anchor=north west][inner sep=0.75pt]    {$h^{*}$};
\draw (126.09,161.5) node [anchor=north west][inner sep=0.75pt]    {$g^{*}$};
\draw (4.4,185.81) node [anchor=north west][inner sep=0.75pt]    {$...$};
\draw (146.73,178.34) node [anchor=north west][inner sep=0.75pt]    {$\Check{H}^{q}( X,\mathcal{F})$};
\draw (241.77,179.34) node [anchor=north west][inner sep=0.75pt]    {$\Check{H}^{q}( X,\mathcal{G})$};
\draw (332.96,179.34) node [anchor=north west][inner sep=0.75pt]    {$\Check{H}^{q+1}( X,\mathcal{E})$};
\draw (524.47,165.81) node [anchor=north west][inner sep=0.75pt]    {$h^{*}$};
\draw (417.9,162.4) node [anchor=north west][inner sep=0.75pt]    {$g^{*}$};
\draw (439.57,179.34) node [anchor=north west][inner sep=0.75pt]    {$\Check{H}^{q+1}( X,\mathcal{F})$};
\draw (549.94,188.81) node [anchor=north west][inner sep=0.75pt]    {$...$};
\end{tikzpicture}   
\end{equation}
\end{Theorem}
\begin{proof}
The short exact sequence of shaves induces for every open covering $\mathcal{U}$, by replacing $\mathcal{G}$ with the image $\Tilde{\mathcal{G}}$ of $f$, a short exact sequence of Čech complexes 
\begin{equation}
\tikzset{every picture/.style={line width=0.75pt}} 
\begin{tikzpicture}[x=0.75pt,y=0.75pt,yscale=-1,xscale=1]

\draw    (68,190) -- (126.73,189.84) ;
\draw [shift={(128.73,189.83)}, rotate = 179.84] [color={rgb, 255:red, 0; green, 0; blue, 0 }  ][line width=0.75]    (10.93,-3.29) .. controls (6.95,-1.4) and (3.31,-0.3) .. (0,0) .. controls (3.31,0.3) and (6.95,1.4) .. (10.93,3.29)   ;
\draw    (335,191) -- (393.73,190.84) ;
\draw [shift={(395.73,190.83)}, rotate = 179.84] [color={rgb, 255:red, 0; green, 0; blue, 0 }  ][line width=0.75]    (10.93,-3.29) .. controls (6.95,-1.4) and (3.31,-0.3) .. (0,0) .. controls (3.31,0.3) and (6.95,1.4) .. (10.93,3.29)   ;
\draw    (203,191) -- (261.73,190.84) ;
\draw [shift={(263.73,190.83)}, rotate = 179.84] [color={rgb, 255:red, 0; green, 0; blue, 0 }  ][line width=0.75]    (10.93,-3.29) .. controls (6.95,-1.4) and (3.31,-0.3) .. (0,0) .. controls (3.31,0.3) and (6.95,1.4) .. (10.93,3.29)   ;
\draw    (469,191) -- (527.73,190.84) ;
\draw [shift={(529.73,190.83)}, rotate = 179.84] [color={rgb, 255:red, 0; green, 0; blue, 0 }  ][line width=0.75]    (10.93,-3.29) .. controls (6.95,-1.4) and (3.31,-0.3) .. (0,0) .. controls (3.31,0.3) and (6.95,1.4) .. (10.93,3.29)   ;

\draw (134,179.4) node [anchor=north west][inner sep=0.75pt]    {$\Check{C}^{*}(\mathcal{U} ,\mathcal{E})$};
\draw (349,166.4) node [anchor=north west][inner sep=0.75pt]    {$h^{*}$};
\draw (221,164.4) node [anchor=north west][inner sep=0.75pt]    {$g^{*}$};
\draw (49,182.5) node [anchor=north west][inner sep=0.75pt]    {$0$};
\draw (537,182.5) node [anchor=north west][inner sep=0.75pt]    {$0$};
\draw (399,179.4) node [anchor=north west][inner sep=0.75pt]    {$\Check{C}^{*}(\mathcal{U} ,\Tilde{\mathcal{G}})$};
\draw (266,179.4) node [anchor=north west][inner sep=0.75pt]    {$\Check{C}^{*}(\mathcal{U} ,\mathcal{F})$};
\end{tikzpicture},
\end{equation}
which induces, by standard theorem of cohomological algebra, a long exact sequence in cohomology
\begin{equation}
\tikzset{every picture/.style={line width=0.75pt}} 
\begin{tikzpicture}[x=0.75pt,y=0.75pt,yscale=-1,xscale=1]

\draw    (121.88,188.6) -- (141,188.94) ;
\draw [shift={(143,188.98)}, rotate = 181.02] [color={rgb, 255:red, 0; green, 0; blue, 0 }  ][line width=0.75]    (10.93,-3.29) .. controls (6.95,-1.4) and (3.31,-0.3) .. (0,0) .. controls (3.31,0.3) and (6.95,1.4) .. (10.93,3.29)   ;
\draw    (28.88,188.6) -- (48,188.94) ;
\draw [shift={(50,188.98)}, rotate = 181.02] [color={rgb, 255:red, 0; green, 0; blue, 0 }  ][line width=0.75]    (10.93,-3.29) .. controls (6.95,-1.4) and (3.31,-0.3) .. (0,0) .. controls (3.31,0.3) and (6.95,1.4) .. (10.93,3.29)   ;
\draw    (217.88,188.6) -- (237,188.94) ;
\draw [shift={(239,188.98)}, rotate = 181.02] [color={rgb, 255:red, 0; green, 0; blue, 0 }  ][line width=0.75]    (10.93,-3.29) .. controls (6.95,-1.4) and (3.31,-0.3) .. (0,0) .. controls (3.31,0.3) and (6.95,1.4) .. (10.93,3.29)   ;
\draw    (308.88,189.6) -- (328,189.94) ;
\draw [shift={(330,189.98)}, rotate = 181.02] [color={rgb, 255:red, 0; green, 0; blue, 0 }  ][line width=0.75]    (10.93,-3.29) .. controls (6.95,-1.4) and (3.31,-0.3) .. (0,0) .. controls (3.31,0.3) and (6.95,1.4) .. (10.93,3.29)   ;
\draw    (412.88,190.6) -- (432,190.94) ;
\draw [shift={(434,190.98)}, rotate = 181.02] [color={rgb, 255:red, 0; green, 0; blue, 0 }  ][line width=0.75]    (10.93,-3.29) .. controls (6.95,-1.4) and (3.31,-0.3) .. (0,0) .. controls (3.31,0.3) and (6.95,1.4) .. (10.93,3.29)   ;
\draw    (521.88,190.6) -- (541,190.94) ;
\draw [shift={(543,190.98)}, rotate = 181.02] [color={rgb, 255:red, 0; green, 0; blue, 0 }  ][line width=0.75]    (10.93,-3.29) .. controls (6.95,-1.4) and (3.31,-0.3) .. (0,0) .. controls (3.31,0.3) and (6.95,1.4) .. (10.93,3.29)   ;
\draw (55.08,178.44) node [anchor=north west][inner sep=0.75pt]    {$\Check{H}^{q}( \mathcal{U},\mathcal{E})$};
\draw (220.01,164.91) node [anchor=north west][inner sep=0.75pt]    {$h^{*}$};
\draw (126.09,161.5) node [anchor=north west][inner sep=0.75pt]    {$g^{*}$};
\draw (4.4,185.81) node [anchor=north west][inner sep=0.75pt]    {$...$};
\draw (146.73,178.34) node [anchor=north west][inner sep=0.75pt]    {$\Check{H}^{q}( X,\mathcal{F})$};
\draw (241.77,179.34) node [anchor=north west][inner sep=0.75pt]    {$\Check{H}^{q}( \mathcal{U},\Tilde{\mathcal{G}})$};
\draw (332.96,179.34) node [anchor=north west][inner sep=0.75pt]    {$\Check{H}^{q+1}( \mathcal{U},\mathcal{E})$};
\draw (524.47,165.81) node [anchor=north west][inner sep=0.75pt]    {$h^{*}$};
\draw (417.9,162.4) node [anchor=north west][inner sep=0.75pt]    {$g^{*}$};
\draw (439.57,179.34) node [anchor=north west][inner sep=0.75pt]    {$\Check{H}^{q+1}( \mathcal{U},\mathcal{F})$};
\draw (549.94,188.81) node [anchor=north west][inner sep=0.75pt]    {$...$};
\end{tikzpicture}   
\end{equation}
At this point is enough to pass to the colimit, which preserves exactness, to get
\begin{equation}
\tikzset{every picture/.style={line width=0.75pt}} 
\begin{tikzpicture}[x=0.75pt,y=0.75pt,yscale=-1,xscale=1]

\draw    (121.88,188.6) -- (141,188.94) ;
\draw [shift={(143,188.98)}, rotate = 181.02] [color={rgb, 255:red, 0; green, 0; blue, 0 }  ][line width=0.75]    (10.93,-3.29) .. controls (6.95,-1.4) and (3.31,-0.3) .. (0,0) .. controls (3.31,0.3) and (6.95,1.4) .. (10.93,3.29)   ;
\draw    (28.88,188.6) -- (48,188.94) ;
\draw [shift={(50,188.98)}, rotate = 181.02] [color={rgb, 255:red, 0; green, 0; blue, 0 }  ][line width=0.75]    (10.93,-3.29) .. controls (6.95,-1.4) and (3.31,-0.3) .. (0,0) .. controls (3.31,0.3) and (6.95,1.4) .. (10.93,3.29)   ;
\draw    (217.88,188.6) -- (237,188.94) ;
\draw [shift={(239,188.98)}, rotate = 181.02] [color={rgb, 255:red, 0; green, 0; blue, 0 }  ][line width=0.75]    (10.93,-3.29) .. controls (6.95,-1.4) and (3.31,-0.3) .. (0,0) .. controls (3.31,0.3) and (6.95,1.4) .. (10.93,3.29)   ;
\draw    (308.88,189.6) -- (328,189.94) ;
\draw [shift={(330,189.98)}, rotate = 181.02] [color={rgb, 255:red, 0; green, 0; blue, 0 }  ][line width=0.75]    (10.93,-3.29) .. controls (6.95,-1.4) and (3.31,-0.3) .. (0,0) .. controls (3.31,0.3) and (6.95,1.4) .. (10.93,3.29)   ;
\draw    (412.88,190.6) -- (432,190.94) ;
\draw [shift={(434,190.98)}, rotate = 181.02] [color={rgb, 255:red, 0; green, 0; blue, 0 }  ][line width=0.75]    (10.93,-3.29) .. controls (6.95,-1.4) and (3.31,-0.3) .. (0,0) .. controls (3.31,0.3) and (6.95,1.4) .. (10.93,3.29)   ;
\draw    (521.88,190.6) -- (541,190.94) ;
\draw [shift={(543,190.98)}, rotate = 181.02] [color={rgb, 255:red, 0; green, 0; blue, 0 }  ][line width=0.75]    (10.93,-3.29) .. controls (6.95,-1.4) and (3.31,-0.3) .. (0,0) .. controls (3.31,0.3) and (6.95,1.4) .. (10.93,3.29)   ;
\draw (55.08,178.44) node [anchor=north west][inner sep=0.75pt]    {$\Check{H}^{q}( X,\mathcal{E})$};
\draw (220.01,164.91) node [anchor=north west][inner sep=0.75pt]    {$h^{*}$};
\draw (126.09,161.5) node [anchor=north west][inner sep=0.75pt]    {$g^{*}$};
\draw (4.4,185.81) node [anchor=north west][inner sep=0.75pt]    {$...$};
\draw (146.73,178.34) node [anchor=north west][inner sep=0.75pt]    {$\Check{H}^{q}( X,\mathcal{F})$};
\draw (241.77,179.34) node [anchor=north west][inner sep=0.75pt]    {$\Check{H}^{q}( X,\Tilde{\mathcal{G}})$};
\draw (332.96,179.34) node [anchor=north west][inner sep=0.75pt]    {$\Check{H}^{q+1}( X,\mathcal{E})$};
\draw (524.47,165.81) node [anchor=north west][inner sep=0.75pt]    {$h^{*}$};
\draw (417.9,162.4) node [anchor=north west][inner sep=0.75pt]    {$g^{*}$};
\draw (439.57,179.34) node [anchor=north west][inner sep=0.75pt]    {$\Check{H}^{q+1}(X,\mathcal{F})$};
\draw (549.94,188.81) node [anchor=north west][inner sep=0.75pt]    {$...$};
\end{tikzpicture}   
\end{equation}
and the only non-trivial step is to show that $\Check{H}^{q}( X,\Tilde{\mathcal{G}}) \cong \Check{H}^{q}( X,{\mathcal{G}})$ that can be done thanks to the good properties of Čech cohomology under changing of refinement function. 
\end{proof}

\subsection{The abstract de Rham theorem}

\begin{Definition}[Resolution of a sheaf]
    A resolution of a sheaf is a exact sequence of the form 
\begin{equation}
\tikzset{every picture/.style={line width=0.75pt}} 
\begin{tikzpicture}[x=0.75pt,y=0.75pt,yscale=-1,xscale=1]

\draw    (68,190) -- (126.73,189.84) ;
\draw [shift={(128.73,189.83)}, rotate = 179.84] [color={rgb, 255:red, 0; green, 0; blue, 0 }  ][line width=0.75]    (10.93,-3.29) .. controls (6.95,-1.4) and (3.31,-0.3) .. (0,0) .. controls (3.31,0.3) and (6.95,1.4) .. (10.93,3.29)   ;
\draw    (253,191) -- (311.73,190.84) ;
\draw [shift={(313.73,190.83)}, rotate = 179.84] [color={rgb, 255:red, 0; green, 0; blue, 0 }  ][line width=0.75]    (10.93,-3.29) .. controls (6.95,-1.4) and (3.31,-0.3) .. (0,0) .. controls (3.31,0.3) and (6.95,1.4) .. (10.93,3.29)   ;
\draw    (156,190) -- (214.73,189.84) ;
\draw [shift={(216.73,189.83)}, rotate = 179.84] [color={rgb, 255:red, 0; green, 0; blue, 0 }  ][line width=0.75]    (10.93,-3.29) .. controls (6.95,-1.4) and (3.31,-0.3) .. (0,0) .. controls (3.31,0.3) and (6.95,1.4) .. (10.93,3.29)   ;
\draw    (343,192) -- (401.73,191.84) ;
\draw [shift={(403.73,191.83)}, rotate = 179.84] [color={rgb, 255:red, 0; green, 0; blue, 0 }  ][line width=0.75]    (10.93,-3.29) .. controls (6.95,-1.4) and (3.31,-0.3) .. (0,0) .. controls (3.31,0.3) and (6.95,1.4) .. (10.93,3.29)   ;

\draw (178,168.4) node [anchor=north west][inner sep=0.75pt]    {$i$};
\draw (52,182.5) node [anchor=north west][inner sep=0.75pt]    {$0$};
\draw (321,182.4) node [anchor=north west][inner sep=0.75pt]    {$\mathcal{E}^{1}$};
\draw (224,181.4) node [anchor=north west][inner sep=0.75pt]    {$\mathcal{E}^{0}$};
\draw (135,180.4) node [anchor=north west][inner sep=0.75pt]    {$\mathcal{F}$};
\draw (270,170.4) node [anchor=north west][inner sep=0.75pt]    {$d$};
\draw (360,171.4) node [anchor=north west][inner sep=0.75pt]    {$d$};
\draw (412,188.5) node [anchor=north west][inner sep=0.75pt]    {$...$};
\end{tikzpicture}
\end{equation}
if every $\mathcal{E}^j$ is acyclic the resolution is called an acyclic resolution. 
\end{Definition}
We note that if we discard the sheaf $\mathcal{F}$ and we consider the global sections $\mathcal{E}^j(X)$ of shaves $\mathcal{E}^j$ we can define their cohomology in the standard way getting the cohomology groups $H^q(\mathcal{E}^*(X))$. For example if the $\mathcal{E}^*(X)$ is the singular cochain complex then $H^q(\mathcal{E}^*(X))$ will be the singular cohomology groups.
\begin{Theorem}[Abstract de Rham theorem]
    Given an acyclic resolution of the sheaf $\mathcal{F}$
\begin{equation}
\tikzset{every picture/.style={line width=0.75pt}} 
\begin{tikzpicture}[x=0.75pt,y=0.75pt,yscale=-1,xscale=1]

\draw    (68,190) -- (126.73,189.84) ;
\draw [shift={(128.73,189.83)}, rotate = 179.84] [color={rgb, 255:red, 0; green, 0; blue, 0 }  ][line width=0.75]    (10.93,-3.29) .. controls (6.95,-1.4) and (3.31,-0.3) .. (0,0) .. controls (3.31,0.3) and (6.95,1.4) .. (10.93,3.29)   ;
\draw    (253,191) -- (311.73,190.84) ;
\draw [shift={(313.73,190.83)}, rotate = 179.84] [color={rgb, 255:red, 0; green, 0; blue, 0 }  ][line width=0.75]    (10.93,-3.29) .. controls (6.95,-1.4) and (3.31,-0.3) .. (0,0) .. controls (3.31,0.3) and (6.95,1.4) .. (10.93,3.29)   ;
\draw    (156,190) -- (214.73,189.84) ;
\draw [shift={(216.73,189.83)}, rotate = 179.84] [color={rgb, 255:red, 0; green, 0; blue, 0 }  ][line width=0.75]    (10.93,-3.29) .. controls (6.95,-1.4) and (3.31,-0.3) .. (0,0) .. controls (3.31,0.3) and (6.95,1.4) .. (10.93,3.29)   ;
\draw    (343,192) -- (401.73,191.84) ;
\draw [shift={(403.73,191.83)}, rotate = 179.84] [color={rgb, 255:red, 0; green, 0; blue, 0 }  ][line width=0.75]    (10.93,-3.29) .. controls (6.95,-1.4) and (3.31,-0.3) .. (0,0) .. controls (3.31,0.3) and (6.95,1.4) .. (10.93,3.29)   ;

\draw (178,168.4) node [anchor=north west][inner sep=0.75pt]    {$i$};
\draw (52,182.5) node [anchor=north west][inner sep=0.75pt]    {$0$};
\draw (321,182.4) node [anchor=north west][inner sep=0.75pt]    {$\mathcal{E}^{1}$};
\draw (224,181.4) node [anchor=north west][inner sep=0.75pt]    {$\mathcal{E}^{0}$};
\draw (135,180.4) node [anchor=north west][inner sep=0.75pt]    {$\mathcal{F}$};
\draw (270,170.4) node [anchor=north west][inner sep=0.75pt]    {$d$};
\draw (360,171.4) node [anchor=north west][inner sep=0.75pt]    {$d$};
\draw (412,188.5) node [anchor=north west][inner sep=0.75pt]    {$...$};
\end{tikzpicture}
\end{equation}   
    there is an isomorphism 
    \begin{equation}
        \Check{H}^q(X,\mathcal{F}) \cong H^q(\mathcal{E}^*(X))
    \end{equation}
\end{Theorem}
\begin{proof}
    We show that for every resolution of the sheaf $\mathcal{F}$ there exist homomorphisms
    \begin{equation}
        \alpha_q : H^q(\mathcal{E}^*(X)) \rightarrow \Check{H}^q(X,\mathcal{F}) \qquad q\geq 0
    \end{equation}
such that 
\begin{itemize}
    \item if $ \Check{H}^{q-i-1}(X,\mathcal{E}^i)=0 \ \forall i \in [0,q-2]$ then $\alpha_q$ is injective; 
    \item if $ \Check{H}^{q-i}(X,\mathcal{E}^i)=0 \ \forall i \in [0,q-1]$ then $\alpha_q$ is surjective. 
\end{itemize}
Hence the theorem follow since if the resolution is acyclic both the condition are satisfied and the $\alpha_q$ is an isomorphism for every $q$. The assertion can be shown by induction.\\
Let us start with $q=0$. In this case, using the sheaf property and the definition of the Čech differential we can show that $\Check{H}^0(X,\mathcal{F})=\mathcal{F}(X)$ (the global section of the sheaf). Hence we have the following exact sequence (due to the left exactness of global sections)
\begin{equation}
    \tikzset{every picture/.style={line width=0.75pt}} 
\begin{tikzpicture}[x=0.75pt,y=0.75pt,yscale=-1,xscale=1]

\draw    (68,190) -- (126.73,189.84) ;
\draw [shift={(128.73,189.83)}, rotate = 179.84] [color={rgb, 255:red, 0; green, 0; blue, 0 }  ][line width=0.75]    (10.93,-3.29) .. controls (6.95,-1.4) and (3.31,-0.3) .. (0,0) .. controls (3.31,0.3) and (6.95,1.4) .. (10.93,3.29)   ;
\draw    (288,189) -- (346.73,188.84) ;
\draw [shift={(348.73,188.83)}, rotate = 179.84] [color={rgb, 255:red, 0; green, 0; blue, 0 }  ][line width=0.75]    (10.93,-3.29) .. controls (6.95,-1.4) and (3.31,-0.3) .. (0,0) .. controls (3.31,0.3) and (6.95,1.4) .. (10.93,3.29)   ;
\draw    (174,190) -- (232.73,189.84) ;
\draw [shift={(234.73,189.83)}, rotate = 179.84] [color={rgb, 255:red, 0; green, 0; blue, 0 }  ][line width=0.75]    (10.93,-3.29) .. controls (6.95,-1.4) and (3.31,-0.3) .. (0,0) .. controls (3.31,0.3) and (6.95,1.4) .. (10.93,3.29)   ;
\draw    (399,189) -- (457.73,188.84) ;
\draw [shift={(459.73,188.83)}, rotate = 179.84] [color={rgb, 255:red, 0; green, 0; blue, 0 }  ][line width=0.75]    (10.93,-3.29) .. controls (6.95,-1.4) and (3.31,-0.3) .. (0,0) .. controls (3.31,0.3) and (6.95,1.4) .. (10.93,3.29)   ;

\draw (196,168.4) node [anchor=north west][inner sep=0.75pt]    {$i$};
\draw (52,182.5) node [anchor=north west][inner sep=0.75pt]    {$0$};
\draw (352,180.4) node [anchor=north west][inner sep=0.75pt]    {$\mathcal{E}^{1}( X)$};
\draw (240,180.4) node [anchor=north west][inner sep=0.75pt]    {$\mathcal{E}^{0}( X)$};
\draw (132,180.4) node [anchor=north west][inner sep=0.75pt]    {$\mathcal{F}( X)$};
\draw (305,168.4) node [anchor=north west][inner sep=0.75pt]    {$d$};
\draw (416,168.4) node [anchor=north west][inner sep=0.75pt]    {$d$};
\draw (467,185.5) node [anchor=north west][inner sep=0.75pt]    {$...$};
\end{tikzpicture}
\end{equation}
which means that  
\begin{equation}
    \Check{H}^0(X,\mathcal{F})=\mathcal{F}(X) \cong Ker(d: \mathcal{E}^0(X) \rightarrow \mathcal{E}^1(X))=H^0(\mathcal{E}^*(X)).
\end{equation}
For the general $q>0$ we first define $\mathcal{F}'$ as the kernel of the morphism $d: \mathcal{E}^1(X) \rightarrow \mathcal{E}^2(X)$. Therefore we have a resolution 
\begin{equation}
\tikzset{every picture/.style={line width=0.75pt}} 
\begin{tikzpicture}[x=0.75pt,y=0.75pt,yscale=-1,xscale=1]

\draw    (68,190) -- (126.73,189.84) ;
\draw [shift={(128.73,189.83)}, rotate = 179.84] [color={rgb, 255:red, 0; green, 0; blue, 0 }  ][line width=0.75]    (10.93,-3.29) .. controls (6.95,-1.4) and (3.31,-0.3) .. (0,0) .. controls (3.31,0.3) and (6.95,1.4) .. (10.93,3.29)   ;
\draw    (253,191) -- (311.73,190.84) ;
\draw [shift={(313.73,190.83)}, rotate = 179.84] [color={rgb, 255:red, 0; green, 0; blue, 0 }  ][line width=0.75]    (10.93,-3.29) .. controls (6.95,-1.4) and (3.31,-0.3) .. (0,0) .. controls (3.31,0.3) and (6.95,1.4) .. (10.93,3.29)   ;
\draw    (156,190) -- (214.73,189.84) ;
\draw [shift={(216.73,189.83)}, rotate = 179.84] [color={rgb, 255:red, 0; green, 0; blue, 0 }  ][line width=0.75]    (10.93,-3.29) .. controls (6.95,-1.4) and (3.31,-0.3) .. (0,0) .. controls (3.31,0.3) and (6.95,1.4) .. (10.93,3.29)   ;
\draw    (343,192) -- (401.73,191.84) ;
\draw [shift={(403.73,191.83)}, rotate = 179.84] [color={rgb, 255:red, 0; green, 0; blue, 0 }  ][line width=0.75]    (10.93,-3.29) .. controls (6.95,-1.4) and (3.31,-0.3) .. (0,0) .. controls (3.31,0.3) and (6.95,1.4) .. (10.93,3.29)   ;

\draw (178,168.4) node [anchor=north west][inner sep=0.75pt]    {$i$};
\draw (52,182.5) node [anchor=north west][inner sep=0.75pt]    {$0$};
\draw (321,182.4) node [anchor=north west][inner sep=0.75pt]    {$\mathcal{E}^{2}$};
\draw (224,181.4) node [anchor=north west][inner sep=0.75pt]    {$\mathcal{E}^{1}$};
\draw (135,180.4) node [anchor=north west][inner sep=0.75pt]    {$\mathcal{F}'$};
\draw (270,170.4) node [anchor=north west][inner sep=0.75pt]    {$d$};
\draw (360,171.4) node [anchor=north west][inner sep=0.75pt]    {$d$};
\draw (412,188.5) node [anchor=north west][inner sep=0.75pt]    {$...$};
\end{tikzpicture},
\end{equation}   
which provides the inductive step by which we have a map 
\begin{equation}
    \alpha_{q-1}: H^{q-1}(\mathcal{E}^{*-1}(X))=H^{q}(\mathcal{E}^{*}(X)) \rightarrow \Check{H}^{q-1}(X,\mathcal{F}'),
\end{equation}
and a short exact sequence 
\begin{equation}
\tikzset{every picture/.style={line width=0.75pt}} 
\begin{tikzpicture}[x=0.75pt,y=0.75pt,yscale=-1,xscale=1]

\draw    (68,190) -- (126.73,189.84) ;
\draw [shift={(128.73,189.83)}, rotate = 179.84] [color={rgb, 255:red, 0; green, 0; blue, 0 }  ][line width=0.75]    (10.93,-3.29) .. controls (6.95,-1.4) and (3.31,-0.3) .. (0,0) .. controls (3.31,0.3) and (6.95,1.4) .. (10.93,3.29)   ;
\draw    (253,191) -- (311.73,190.84) ;
\draw [shift={(313.73,190.83)}, rotate = 179.84] [color={rgb, 255:red, 0; green, 0; blue, 0 }  ][line width=0.75]    (10.93,-3.29) .. controls (6.95,-1.4) and (3.31,-0.3) .. (0,0) .. controls (3.31,0.3) and (6.95,1.4) .. (10.93,3.29)   ;
\draw    (156,190) -- (214.73,189.84) ;
\draw [shift={(216.73,189.83)}, rotate = 179.84] [color={rgb, 255:red, 0; green, 0; blue, 0 }  ][line width=0.75]    (10.93,-3.29) .. controls (6.95,-1.4) and (3.31,-0.3) .. (0,0) .. controls (3.31,0.3) and (6.95,1.4) .. (10.93,3.29)   ;
\draw    (343,192) -- (401.73,191.84) ;
\draw [shift={(403.73,191.83)}, rotate = 179.84] [color={rgb, 255:red, 0; green, 0; blue, 0 }  ][line width=0.75]    (10.93,-3.29) .. controls (6.95,-1.4) and (3.31,-0.3) .. (0,0) .. controls (3.31,0.3) and (6.95,1.4) .. (10.93,3.29)   ;

\draw (178,168.4) node [anchor=north west][inner sep=0.75pt]    {$i$};
\draw (52,182.5) node [anchor=north west][inner sep=0.75pt]    {$0$};
\draw (321,182.4) node [anchor=north west][inner sep=0.75pt]    {$\mathcal{F} '$};
\draw (224,181.4) node [anchor=north west][inner sep=0.75pt]    {$\mathcal{E}^{0}$};
\draw (135,180.4) node [anchor=north west][inner sep=0.75pt]    {$\mathcal{F}$};
\draw (270,170.4) node [anchor=north west][inner sep=0.75pt]    {$d$};
\draw (411,182.5) node [anchor=north west][inner sep=0.75pt]    {$0$};
\end{tikzpicture},
\end{equation}
which gives us a long exact sequence in cohomology and composing with the map $\delta: \Check{H}^{q-1}(X,\mathcal{F}') \rightarrow \Check{H}^{q}(X,\mathcal{F})$ we get the map we are looking for.
\end{proof}
\subsection{The sheaf of differential mixed symmetry tensors}
Let us now set $(X,\mathcal{A})$ a differential manifold with differential maximal atlas $\mathcal{A}$. This atlas naturally furnish a structure sheaf $\mathcal{E}$ (this is the sheaf associated to the manifold; for example for the case of $\mathbb{R}^n$ this is the sheaf of $C^{\infty}$ functions). In the following definition $\chi(U)$ is the sheaf of differential vector field on $U$.

\begin{Definition}[Differential mixed symmetry tensor]
    We define the sheaf $\mathcal{D}\Omega^{p_1\otimes ... \otimes p_N}$ as the sheaf associated to the map $x \mapsto \Omega^{p_1\otimes ... \otimes p_N}(X)$. Then there exist, for every open $U \subset X$, a map $\langle-;-,...,-\rangle : \mathcal{D}\Omega^{p_1\otimes ... \otimes p_N}(U) \times \chi^{p_1}(U) \times ... \times \chi^{p_N}(U) \rightarrow \mathcal{D}\mathbb{R}_X$ defined by
    \begin{equation}
        \langle T;(\alpha_1,..,\alpha_{p_1}),...,(\beta_1,..,\beta_{p_N})\rangle (x):= (T(x);(\alpha_1(x),..,\alpha_{p_1}(x)),...,(\beta_1(x),..,\beta_{p_N}(x)))
    \end{equation}
which is nothing but the evaluation of the mixed symmetry tensor $T$ on sets vector fields. A differential mixed symmetry tensor is an element $T \in \mathcal{D}\Omega^{p_1\otimes ... \otimes p_N}(U)$ such that for every $V \subset U$ 
\begin{equation}
    \langle T;(\alpha_1,..,\alpha_{p_1}),...,(\beta_1,..,\beta_{p_N})\rangle |_V (x) \in \mathcal{E}(V).
\end{equation}
\end{Definition}
As for the case of vector fields and differential forms, differential mixed symmetry tensors give rise to a sheaf and, with a little abuse of notation, we will still indicate with $\Omega^{p_1\otimes ... \otimes p_N}$ the sheaf of differential mixed symmetry tensors. We also stress that $\Omega^{p_1\otimes ... \otimes p_N}$ is also an $\mathcal{E}$-module under the pointwise multiplication.  

\begin{Theorem}[Fine shaves on differential manifold]
Let $\mathcal{F}$ be an $\mathcal{E}$-module on $X$. If $X$ is a differential manifold than $\mathcal{F}$ is a fine sheaf.
\end{Theorem}
\begin{proof}
   Given an open covering $\mathcal{U}$, it is enough to consider as partition of the identity $\mathcal{F}$ subordinated to $\mathcal{U}$ the multiplication of the elements of $\mathcal{F}$ by the functions of the partition of the unity of $X$ subordinated to $\mathcal{U}$.
\end{proof}

\bibliographystyle{JHEP} 
\bibliography{biblio}

\end{document}